\documentclass[11pt]{article}

\usepackage{xspace}
\usepackage{graphicx}
\usepackage{waingarten}
\usepackage[linesnumbered,ruled]{algorithm2e}
\usepackage{thmtools}
\usepackage{thm-restate}

\def\FullBox{\hbox{\vrule width 8pt height 8pt depth 0pt}}
\newcommand{\QED}{\;\;\;\FullBox}
\renewenvironment{proof}{\noindent{{\textbf{Proof:}~}}} {\hfill\QED}
\providecommand{\email}[1]{\href{mailto:#1}{\nolinkurl{#1}\xspace}}
\newenvironment{proofof}[1]{\noindent{\bf Proof of {#1}:~}}{\hfill\(\QED\)}
\newcommand{\eqdef}{\stackrel{\rm def}{=}}

\newcommand{\SCOND}{\ensuremath{\mathsf{SCOND}}\xspace}

\newcommand{\pr}[1]{\text{\bf Pr}\normalfont\lbrack #1 \rbrack} 
\newcommand{\ex}[1]{\mathbf{E}\normalfont\lbrack #1 \rbrack}
\newcommand{\bpr}[1]{\text{\bf Pr}\normalfont \Big[#1 \Big]} 

\newcommand{\ignore}[1]{}

\title{Learning and Testing Junta Distributions with Subcube Conditioning}

\author {
  Xi Chen\thanks{Columbia University. \email{xichen@cs.columbia.edu}. 
  Supported by NSF IIS-1838154 and NSF CCF-1703925.}
  \and 
  Rajesh Jayaram\thanks{Carnegie Mellon University. \email{rkjayara@cs.cmu.edu}.  Rajesh Jayaram would like to thank the partial support from the Office of Naval Research (ONR) grant N00014-18-1-2562, and the National Science Foundation (NSF) under Grant No. CCF-1815840.}
  \and
  Amit Levi\thanks{Cheriton School of Computer Science, University of Waterloo. \email{amit.levi@uwaterloo.ca}. Research supported by the David R. Cheriton Graduate Scholarship.}
  \and
  Erik Waingarten\thanks{Columbia University. \email{eaw@cs.columbia.edu}. Supported by the NSF Graduate Research Fellowship (Grant No. DGE-16-44869), NSF CCF-1563155, and NSF CCF-1814873.}
}

\begin{document}
\maketitle

\begin{abstract}
We study the problems of learning and testing junta distributions on $\{-1,1\}^n$ with respect to the uniform distribution, where a distribution $p$ is a $k$-junta if its probability mass function $p(x)$  depends on a subset of at most $k$ variables. 
The main contribution is an algorithm for finding relevant coordinates in a $k$-junta distribution with 
subcube conditioning 
\cite{BC18, CCKLW20}. 
We give two applications: 
\begin{itemize}
\item 
An algorithm for learning 
  $k$-junta distributions with 
  $\tilde{O}(k/\eps^2) \log n + O(2^k/\eps^2)$\\ subcube conditioning queries, and
\item 
An algorithm for testing $k$-junta distributions with 
$\tilde{O}((k + \sqrt{n})/\eps^2)$ 
subcube \\ conditioning queries. 
%
 \end{itemize}
All our algorithms are optimal up to poly-logarithmic factors.


Our results show that subcube conditioning, 
as a natural model for accessing 
  high-dimensional distributions,
enables significant savings in 
learning and testing junta distributions compared to the standard sampling model.
This addresses an open question posed   by Aliakbarpour, Blais, and Rubinfeld~\cite{ABR17}.

\end{abstract}

\thispagestyle{empty}
\newpage
\tableofcontents
\thispagestyle{empty}
\newpage
\pagenumbering{arabic}

\mathchardef\mhyphen="2D
\newcommand{\Junta}[1]{\mathtt{Junta}{(#1)}}

\section{Introduction}

We consider the problems of \emph{learning and testing $k$-junta distributions}, as first studied by Aliakbarpour, Blais, and Rubinfeld \cite{ABR17}. Given $n \in \N$ and $k \leq n$, a distribution $p$ supported on $\{-1,1\}^n$ is a $k$-junta distribution (with respect to the uniform distribution) if the probability mass function $p(x) = \Prx_{\bz \sim p}[\bz = x]$ is a $k$-junta.\footnote{We say a function $f(x)$ over $\{-1,1\}^n$ is a $k$-junta (function) if
	it depends on a subset of no more than $k$ variables.
	More generally, \cite{ABR17} defines $k$-junta distributions with respect to a fixed distribution $q$. For $n \in \N$, $k \leq n$, and a fixed distribution $q$ supported on $\{-1,1\}^n$, a distribution $p$ over $\{-1,1\}^n$ is a $k$-junta distribution with respect to $q$ if there exist $k$ coordinates $i_1, \dots, i_k \in [n]$ such that for every $x \in \{-1,1\}^k$, the distributions $p$ and $q$ conditioned on coordinates $i_1,\dots, i_k$ being set according to $x$ are equal. When $q$ is the uniform distribution, the above definition is equivalent to the requirement that $p(x)$ is a $k$-junta function.} 
The goal of the learning problem is to design algorithms which, given access to an unknown $k$-junta distribution $p$ over $\{-1,1\}^n$, output a hypothesis distribution $\hat{p}$ that satisfies $\dtv(p, \hat{p}) \leq \eps$. In the testing problem, the goal is to design algorithms which, given access to an arbitrary distribution $p$, can distinguish between $p$ being a $k$-junta distribution, and being $\eps$-far from a $k$-junta distribution.\footnote{Here, two distributions $p$ and $q$ are $\eps$-far if $\dtv(p, q) \geq \eps$, and $p$ is $\eps$-far from being a $k$-junta distribution if every $k$-junta distribution is $\eps$-far from $p$.}

The study of computational aspects of juntas has spawned a large body of work (for instance, see \cite{MOS03, FKRSS04, CG04, LMMV05, AR07, AM08, AKL09, V15, B08, B09, B10, STW15, BC16, BCELR18, CSTWX17, S18, LCSSX18, LW19, DMN19,PRW20} and references therein).
These problems are motivated by the \emph{feature selection} problem in machine learning
(see e.g. \cite{GE2003,LM2012,CS2014}), and are classically referred to in theoretical computer science as ``learning in the presence of irrelevant information'' \cite{B94, BL97}.  
The landmark (open) problem is the ``junta problem'' \cite{B03b,MOS03,V15}: given an unknown $k$-junta $f\colon \{-1,1\}^n \to \{-1,1\}$, an algorithm receives independent samples $(\bx, f(\bx))$ where $\bx \sim \{-1,1\}^n$ is uniform, and the task is to learn $f$ (with respect to the uniform distribution). Aliakbarpour, Blais, and Rubinfeld \cite{ABR17} study the analogous problem for distributions: for an unknown $k$-junta distribution $p$ over $\{-1,1\}^n$, an algorithm receives independent samples $\bx \sim p$, and the task is to learn $p$ to within small distance in total variation. They obtain an algorithm with sample complexity $\tilde{O}(2^{2k}) \log n /\eps^4$ and running time $\tilde{O}(2^{2k}) \min\{ n^k, 2^n \} / \eps^4$, and observed that any algorithm for learning $k$-junta distributions may be used to solve the ``junta problem.'' 
Hence, running time significantly better than $n^k$ (in particular, polynomial upper bounds for $k = O(\log n)$) would constitute a major breakthrough in computational learning theory.

Turning to testing $k$-junta distributions, \cite{ABR17} give a tight bound of $\tilde{\Theta}(2^{n/2} /\eps^2)$ for the number of samples $\bx\sim p$ needed. We note that this ``curse of dimensionality'' is not unique to the~problem of
testing junta distributions, and already appears for the most basic testing task: 
testing whether a distribution on $\{-1,1\}^n$ is uniform \cite{P08,VV17}, which can be viewed as testing $k$-junta
distributions with $k=0$.
Works addressing this state-of-affairs have proceeded by either analyzing restricted classes of high dimensional distributions \cite{RS09, CDKS17, DP17, DDK19, GLP18, BBCSV20, DKP19}, or by augmenting the oracle \cite{BDKR05, CR14, CRS15, CFGM16, ABDK18, BC18, OS18}.

\textbf{Membership queries.}
It has been observed \cite{BL97, MOS03, B03b} that the classic ``junta problem" becomes significantly easier when allowing \emph{membership queries}.\footnote{In learning theory, a membership query refers to an oracle which returns $f(x)$ upon a query $x \in \{-1,1\}^n$.} 
In particular, a simple algorithm making $O(k\log n / \eps)$ queries will find at most $k$ relevant variables such that the function~is $\eps$-close to a junta function over those variables.\footnote{The algorithm iteratively builds a set $J \subset [n]$ of relevant variables by sampling pairs of points $\bx,\by \sim \{-1,1\}^n$ with $\bx_J = \by_J$; when $f(\bx) \neq f(\by)$, the algorithm performs a binary search to find a new relevant variable to add to $J$.} 
For the problem of testing junta functions (with membership queries),
the state-of-the-art algorithm \cite{B09} only has query complexity $\tilde{O}(k/\eps)$ with no dependency on $n$.
This leads to the following question that motivates our work: 
\begin{flushleft}\begin{quote}
		\emph{What is an appropriate ``membership query'' model for learning and testing junta distributions, and would such query access admit significant complexity savings?}  
\end{quote}\end{flushleft}

\ignore{This is in sharp contrast with state-of-the-art
	algorithms for testing $k$-junta functions \cite{} with \emph{membership queries}.\footnote{A membership query refers to an oracle which returns $f(x)$ upon an input $x \in \{-1,1\}^n$.}
	The same phenomenon has been observed \cite{BL97, MOS03, B03b} that the classic learning 
	``junta problem" becomes significantly easier when allowing membership queries.}


\textbf{Subcube conditioning queries.} This paper considers the \emph{subcube conditioning model}, first studied by \cite{BC18}.
A subcube conditioning query on a distribution $p$ over $\{-1,1\}^n$ is specified by a string (or 
a restriction as we call in the paper) $\rho \in \{-1,1,*\}^n$. 
The oracle returns~a sample $\bx \sim p$ conditioned  on every $i \in [n]$ 
with $\rho_i \neq *$ having $\bx_i = \rho_i$.
Equivalently, $\rho$ encodes a subcube of $\{-1,1\}^n$ by fixing  non-$*$ coordinates in $\rho$;
the oracle returns a sample $\bx\sim p$ conditioned on $\bx$ lying in the subcube.\footnote{We note that while this paper considers distributions supported on $\{-1,1\}^n$, \cite{BC18} study subcube conditioning in a general product domain $\Sigma^n$. There, a subcube conditioning query is specified by a sequence of $n$ subsets $A_1 \times \dots \times A_n$ where each $A_i \subset \Sigma$, and a sample $\bx \sim p$ conditioned on $\bx_i \in A_i$ for all $i \in [n]$. Extending results from $\{-1,1\}^n$ to $\Sigma^n$ is a direction for future work.} 
When the subcube encoded by $\rho$ is not supported in $p$, 
the oracle under the model of \cite{BC18}
returns a point drawn uniformly from the subcube. 
We remark that this modeling choice
is not important for this paper:
our algorithms only make queries $\rho$
that are consistent with a sample~$x$ previously drawn from $p$ (i.e., $\rho_i=x_i$ for every 
non-$*$ coordinate $i$).\footnote{This gives our algorithms a flavor of those under the \emph{active learning}\hspace{0.04cm}/\hspace{0.04cm}\emph{testing}
	model \cite{DAS05,Burr09,BBBY12}, adapted to the setting of distribution testing: 
	an algorithm can only zoom in onto a subcube using conditioning queries after it is 
	discovered by samples drawn from the distribution.
	Our lower bounds, on the other hand, apply to the original subcube conditioning model, which only makes them stronger.}

The subcube conditioning model seems particularly appropriate for computational tasks over distributions supported on (high-dimensional) product domains, and was suggested in \cite{CRS15} as an open direction for learning and testing distributions over $\{-1,1\}^n$. From the purely theoretical perspective, we find two aspects of subcube conditioning especially compelling. The first is that restrictions of distributions over product domains are themselves distributions over product domains, which enable algorithms and their analyses to proceed recursively. The second is that algorithms may proceed via the method of (random) restrictions, exploiting properties of distributions apparent only by considering subcubes.
{See more discussions on random restrictions in Section \ref{sec:overview}.}

From a practical perspective, subcube conditional queries arise in a number of applications. An important example is sampling from large joins in a relational database. For database joins, subcube conditioning has a natural interpretation: a sample from a join conditioned on a subcube
{(defined by fixing certain attributes in the join)} can be represented as a sample from another join, where conditioning is first applied to each relation individually.\footnote{For example, a sample from a large multi-way join $J = R_1 \Join \cdots  \Join  R_m$ of relations $R_1,\dots,R_m$ conditioned on fixing a subset of attributes according to a restriction $\rho$
	corresponds to a sample from the join query $J' = R_1' \Join \dots R_m'$, where each $R_i'$ is the restriction of the relation $R_i$ where attributes are fixed according to $\rho$.} Thus, subcube conditional sampling from a join can be implemented in the same time as uniform sampling from a join with a minor overhead.  Moreover, efficiently sampling from joins is an important task in database theory \cite{chaudhuri1999random,AGPR99,zhao2018random, CY20}, and can often be implemented substantially faster than the time required to compute the entire query (which may be exponential in the number of relations given as input to the join).

\paragraph{Other query models.} We briefly discuss other proposed access oracles for distributions. The evaluation oracle \cite{BDKR05,CR14} allows algorithms to query the probability mass function of an input, in addition to receiving random samples. We note the same ``binary search'' strategy prescribed for finding relevant variables in a $k$-junta function works well in this setting, making it too strong for learning juntas. \cite{OS18} considers probability-revealing samples, where the algorithm receives pairs $(\bx, p(\bx))$ with $\bx \sim p$. This model is too weak for the learning problem, since the reduction
of \cite{ABR17} from the $k$-junta problem to the $k$-junta distribution problem  applies   to this oracle as well.\footnote{In particular, consider an unknown $k$-junta function $f \colon \{-1,1\}^n \to \{-1,1\}$, and notice that with $\poly(2^{k})$ random samples, we may know exactly how many inputs $x \in \{-1,1\}^n$ have $f(x) = 1$. Then, the reduction of \cite{ABR17} constructs the distribution which is uniform over the inputs where $f(x) = 1$, so knowing the probability mass function at these points gives no additional information.}  Lastly, and most relevant to this paper, is the (general) conditional sampling model, introduced in \cite{CFGM13, CFGM16, CRS14, CRS15}, where an algorithm is allowed to specify a (arbitrary) subset $A$ of the domain and receive a sample conditioned on it lying in $A$. This model is more powerful than subcube conditioning, yet, looking ahead, our lower bounds for learning $k$-junta distributions will apply to this model as well, showing that conditioning on arbitrary sets $A \subseteq \{-1,1\}^n$ is no more powerful than that 
on subcubes for the learning problem.

\subsection{Our results}


\paragraph{Learning $k$-junta distributions.} Our main algorithmic contribution is a procedure that can, given subcube conditioning query access
to a $k$-junta distribution $p$ over $\{-1,1\}^n$,
identify a set $J\subset [n]$ of at most $k$ relevant variables such that
$p$ is close to a $k$-junta over $J$.
The number of queries needed to identify each relevant variable, on average, is 
roughly $\log n/\eps^2$. (We emphasize though that the main idea behind 
the algorithm is not based on binary search; see Section \ref{sec:overview} for an overview
of the algorithm.) 


\begin{restatable}[Identifying relevant variables]{theorem}{thmidentify}\label{maintheorem}
	There is a randomized algorithm, which takes subcube conditioning query access to an unknown distribution $p$ over $\{-1,1\}^n$, an integer $k\in \N$, and~a~parameter $\eps \in (0, 1/4]$. The~algorithm makes 
	$\tilde{O}(k/\eps^2)\cdot \log n$ 
	queries, runs in time $\tilde{O}(k/\eps^2)\cdot n\log n$~and outputs a set $\bJ \subset [n]$ with 
	the following guarantee. 
	If $p$ is a $k$-junta distribution then $|\bJ|\le k$ and 
	$p$ is $\eps$-close to a junta~distribution over variables in $\bJ$ with probability at least $2/3$.
\end{restatable}
It is known as folklore that, once such a set $J$ is identified,
the unknown $k$-junta distribution $p$ can be learnt easily using another batch of
$O(2^k/\eps^2)$ samples from $p$
and the same amount of running time. 
Together we obtain the following corollary, showing that subcube conditioning queries
enable significant speedup compared to state-of-the-art learning algorithms under the 
sampling model.


\begin{corollary}[Learning junta distributions]
	Under the subcube conditioning query model, 
	there~is a learning algorithm for $k$-junta distributions with
	query complexity $\tilde{O}(k/\eps^2)\cdot \log n+O(2^k/\eps^2)$ 
	and running time $\tilde{O}(k/\eps^2)\cdot n\log n+O(2^k/\eps^2)$.
\end{corollary}

We show that query complexities of both algorithms are almost~tight.
Indeed they are almost~tight even under the more powerful \emph{general conditioning query model},
which was introduced  simultaneously by \cite{CFGM13,CFGM16} and \cite{CRS14,CRS15}.
A  general conditioning query to $p$ is specified by an arbitrary subset $A$ of $\{-1,1\}^n$ (which 
is not necessarily a subcube)
and the oracle returns a sample $\bx\sim p$ conditioned on $\bx\in A$.

\begin{restatable}{theorem}{thmlowerboundone}\label{thm:learning-lb}
	Let $0 < \eps \leq 1/8$, $n \in \N$ and $0 < k \leq n-1$. Suppose an algorithm receives as input conditional query access to an unknown $k$-junta distribution $p$ supported on $\{-1,1\}^n$ and outputs a set $\bJ \subset [n]$ with $|\bJ| \leq k$ such that with probability at least $4/5$, $p$ is $\eps$-close to a junta distribution over $\bJ$. Then, the algorithm must make $\Omega(\log\binom{n}{k} / \eps^2)$ queries. \end{restatable}

\begin{restatable}{theorem}{thmlowerboundtwo}\label{thm:learning-lb-2}
	Let $0 < \eps \le 1/120$, $n \in \N$ and $0<k\le n-1$. Suppose an algorithm receives as input conditional query access to an unknown $k$-junta distribution $p$ over $\{-1,1\}^n$ and outputs a distribution $\hat{\bp}$ such that with probability at least $4/5$, $p$ is $\eps$-close to $\hat{\bp}$. Then, the algorithm must make $\Omega(\log\binom{n}{k} / \eps^2)+\Omega(2^{k}/\eps^2)$ queries.
\end{restatable}

\paragraph{Testing $k$-junta distributions} For the problem of testing junta distributions, 
we obtain matching upper and lower bounds for the query complexity under the
subcube conditioning query model.

\begin{restatable}[Testing junta distributions]{theorem}{thmtestingalg}\label{thm:testing}
	There is an algorithm, 
	which takes subcube conditioning access to an unknown distribution $p$ over $\{-1,1\}^n$, an integer $k\in \N$, 
	and $\eps \in$ $(0, 1/4]$. 
	It makes 
	$$
	\tilde{O}\left(\frac{k+\sqrt{n}}{\eps^2}\right)
	$$
	queries, runs in time $\tilde{O}(n(k+\sqrt{n})^2/\eps^4)$ and achieves the following guarantee:
	It accepts with probability at least $2/3$ if $p$ is a $k$-junta distribution,
	and rejects with probability at least $2/3$ if $p$ is $\eps$-far from a $k$-junta.
\end{restatable}


\begin{restatable}[Lower bound for junta testing]{theorem}{thmlowerboundthree}\label{thm:lb}
	There exist two absolute constants $\eps_0>0$ and $C_0\in \N$ such that
	for any setting of $0<\eps\le \eps_0$, $n\ge C_0$ and $0\le k\le n/2$,
	any algorithm which receives as input subcube conditioning query access to an unknown distribution $p$ supported on $\{-1,1\}^n$ and distinguishes with probability at least $2/3$ between the case when $p$ is a $k$-junta distribution and the case when $p$ is $\eps$-far from any $k$-junta distribution must make 
	at least ${\tilde{\Omega}(k + \sqrt{n})}/{\eps^2}$ many 
	queries.
	Furthermore, the lower bound holds even when $p$ is promised to be a product distribution.
\end{restatable}

An open problem posed by Aliakbarpour, Blais and Rubinfeld \cite{ABR17}
is whether their exponential lower bound for testing junta distributions under the sampling oracle can be bypassed
using general conditioning queries.~We answer the question positively with subcube conditioning queries.


\subsection{Technical overview}\label{sec:overview}

We give an overview of our results for learning and testing junta distributions.
All our algorithms heavily use \emph{random restrictions} drawn using samples from the unknown  distribution.
We start with some  notation for restrictions and how we apply them on a distribution.

Let $p$ be a distribution over $\{-1,1\}^n$ and let $\rho\in \{-1,1,*\}^n$ be a restriction.
We write $p_{|\rho}$ to denote the distribution obtained by applying the restriction $\rho$ on $p$:
it is supported on $\smash{\{-1,1\}^{\stars(\rho)}}$~where $\stars(\rho)$ is the set of $i\in [n]$ with $\rho_i=*$,
and $\by\sim p_{|\rho}$ is drawn by first drawing $\bx\sim p$ conditioned on $\bx_i=\rho_i$ for all $i\notin
\stars(\rho)$ and then setting $\by=\bx_{\stars(\rho)}$.
There will be mainly two ways we draw a random restriction $\brho$.
In the first scenario, we fix a set $S\subset [n]$ and draw a random restriction $\brho$ by 
first drawing $\bx\sim p$ and then setting $\brho_i=\bx_i$ for each $i\notin S$ and $\brho_i=*$ otherwise.
We denote this distribution of restrictions by $\calD_S(p)$.
The more sophisticated  way of drawing a random restriction $\brho$, given a parameter $\sigma\in (0,1)$, is 
to first draw $\bx\sim p$ and a random set $\bS\subseteq [n]$ by including each element independently with
probability $\sigma$.
We then set $\brho_i=\bx_i$ for each $i\notin \bS$ and $\brho_i=*$ otherwise.
We denote this distribution of restrictions by $\calD_\sigma(p)$

\textbf{Algorithm for identifying relevant variables.}
Given access to a distribution $p$, the algorithm proceeds 
by maintaining a set $J$ (initially empty) of
relevant\footnote{Unlike the Boolean function setting,
	we only know that variables in $J$ are relevant with high probability.}
variables found, 
and iteratively adding to $J$ until no more relevant variables are found. 
Hence, the key challenge is discovering new relevant variables when 
$p$ 
remains $\eps$-far from any $k$-junta distribution over $J$.
The latter condition implies 
$$
\E_{\brho\sim \calD_{\overline{J}}(p)}\Big[\dtv\big(p_{|\brho},\calU\big)\Big]\ge \eps,
$$
where $\calU$ denotes the uniform distribution (of the right dimension).
Assume, for convenience, that 
the algorithm samples a restriction $\rho$ with $\dtv(p_{|\rho},\calU)$ $\ge \eps$.  The major difficulty is that arbitrary correlations among (yet unknown) $k$ relevant variables may hide the non-uniform nature of $p_{|\rho}$.\footnote{For example, consider the $k$-junta distribution $p$ over $\{-1,1\}^n$ which is parameterized by a subset $S \subset [n]$ of size $k$ (denoting the relevant variables). A sample $\bx \sim p$ is uniform over all points $y \in \{-1,1\}^n$ where $\prod_{i\in S} y_i = 1$. Notice that $\dtv(p, \calU) \geq 1/2$, however, the distribution given by projecting $p$ onto any subset of coordinates which does not completely include all $S$ variables is exactly uniform. The silver lining (for this specific distribution) will be that if a restriction $\rho$ fixes all but one variable in $S$, i.e., $S \cap \stars(\rho) = \{ i \}$, then every sample $\bx \sim p_{|\rho}$ will have $\bx_i$ always set to the same value.}
For this, we leverage a set of recently-developed tools from \cite{CCKLW20} for analyzing mean vectors of random restrictions of distributions. Specifically, for an arbitrary distribution $p$ over $\{-1,1\}^n$, we denote $\mu(p) \in [-1,1]^n$ as the \emph{mean vector},
$$ \mu(p) \eqdef \Ex_{\bx \sim p}\left[ \bx \right] \in [-1,1]^n. $$
We prove the following structural lemma for distributions which are far-from $k$-juntas. At a high level, this lemma allows us to find relevant variables by only considering the marginal distributions on specific coordinates after applying random restrictions. \ignore{In general, $p_{|\rho}$ is not a product distribution,\footnote{For general distributions over $\{-1,1\}^n$, the relationship between the total variation distance to uniformity and the $\ell_2$-norm of mean vector is not true for general distributions. A simple example is the uniform distribution on vectors whose parity is $1$, which is very far from uniform yet has mean vector $0$.}}
%

\begin{restatable}[Main structural lemma]{lemma}{mainstructurallemma}\label{lem:main-structural}
	There is a universal constant $c>0$ such that the following holds.
	Let $p$ be any probability distribution supported over $\{-1,1\}^n$ for some $n \in \N$. 
	Let $J \subset [n]$~be a subset of variables such that $p$
	is $\eps$-far from being a junta distribution over variables in $J$ for some $\eps \in (0, 1/4]$.\footnote{We require $\eps\le 1/4$ just so that $\log (n/\eps)\ge 2$ even when $n=1$; this helps avoid an
		extra multiplicative constant needed on the right hand side of (\ref{eq:structuralbound}).}  
	Then for $\sigma =1/2$ we have
	\begin{equation}\label{eq:structuralbound} \sum_{j=1}^{ \lceil \log_2 2n \rceil} \Ex_{\brho \sim \calD_{\ol J}(p)} \left[ \Ex_{\bnu \sim \calD_{\sigma^j}(p_{|\brho})}\Big[\big\|\mu((p_{|\brho})_{|\bnu})\big\|_2 \Big] \right] \geq \dfrac{\eps}
	{ \log^{c} (n/\eps) }.
	\end{equation}
\end{restatable}


We will apply the main structural lemma to the distribution $p$ projected onto its $k$ relevant variables (so $n$ in Lemma~\ref{lem:main-structural} becomes $k$), which suggests the following algorithm: for each $j=1,\ldots,\lceil \log_2 2k\rceil$, 
draw $\brho$ and $\bnu$ as described above in the hopes that $\|\mu((p_{|\brho})_{|\bnu})\|_2 \geq \eps / \log^c(k/\eps)$. Once this occurs, since $\mu((p_{|\brho})_{|\bnu})$ contains at most $k$ non-zero coordinates, at least one coordinate $i \in \stars(\bnu)$ will have mean at least $\eps / (\sqrt{k}\log^c(k/\eps) )$ 
in magnitude. In other words, the $i$-th variable is relevant, and the marginal distribution on the $i$-th coordinate of $(p_{|\brho})_{|\bnu}$ is biased by at least \smash{$\tilde{\Omega}(\eps / \sqrt{k})$}. Taking \smash{$\tilde{O}(k / \eps^2) \cdot \log n$} random samples from $(p_{|\brho})_{|\bnu}$ is enough to identify all relevant coordinates whose marginal is at least \smash{$\tilde{\Omega}(\eps / \sqrt{k})$} to include into $J$; furthermore, (by the extra $(\log n)$-factor), we never include a non-biased coordinate in $J$. 
Notice, however, that all guarantees are only in expectation, and we need to employ a budget doubling strategy to achieve the nearly-optimal bound.

\textbf{Algorithm for testing junta distributions.}
The testing algorithm first runs the algorithm for identifying relevant variables, and then tests whether the distribution depends only on the relevant variables found. In particular, let $J$ be the set of variables it returns, and notice that the algorithm may immediately reject if $|J|>k$, 
since every variable in $J$ found by the algorithm is relevant (with high probability). 
The remaining task is 
distinguishing between the following two cases:
\begin{flushleft}\begin{enumerate}
		\item If $p$ is $\eps$-far from $k$-junta distributions, then by definition $p$ is $\eps$-far
		from any junta distribution over $J$.
		By the main structural lemma, there is some $j= 1,\ldots,\lceil\log_2 2n\rceil$
		\\ such that $\|\mu((p_{|\brho})_{|\bnu})\|_2$ is large (in expectation) when $\brho\sim \calD_{\overline{J}}(p)$
		and $\bnu\sim \calD_{\sigma^j}(p_{|\brho})$.
		\item If $p$ is a $k$-junta distribution, then 
		for every $j =1, \dots, \lceil \log_2 2n\rceil$, $(p_{|\brho})_{|\bnu}$ will (trivially)
		\\still be a $k$-junta distribution and $\|\mu((p_{|\brho})_{\bnu})\|_2$ will tend to be small (in expectation).\\ The intuition for the latter condition is that otherwise, the algorithm for finding relevant variables as sketched above would have identified more variables.
\end{enumerate}\end{flushleft}
\ignore{As a result, to finish the job (with simplifications for the convenience of the overview) 
	all we need is an algorithm that, given a distribution $q$ over $\{-1,1\}^n$,
	can distinguish between the case when $q$ is a $k$-junta distribution and has
	$\|\mu(q)\|_2\le \eps/100$ and the case when $q$ is a (general) distribution that has $\|\mu(q)\|_2\ge \eps$.
	We note that the assumption that $q$ is a $k$-junta in the first case is crucial;
	otherwise it becomes a tolerant testing problem and can have a much higher query complexity.}

To this end, we design a ``robust mean tester'' for juntas distributions.

\begin{restatable}[Robust mean testing for juntas]{theorem}{MeanTestingRestatable}
	\label{thm:MeanTesting++}
	There is an algorithm 
	which, given sample
	access~to a distribution $p$ on $\{-1, 1\}^n$, $k\in \N$ and  
	a parameter $\eps\in (0, 1)$, has the following
	behavior:
	\begin{flushleft}\begin{enumerate}
			\item If $p$ is a $k$-junta distribution with $\|\mu(p)\|_2 \leq \eps \sqrt{n} / 100$,
			the algorithm returns \emph{``\texttt{Is a $k$-junta}''} with probability at least $2/3$;
			\item If $p$ is a distribution that satisfies $\|\mu(p)\|_2 \geq \eps \sqrt{n}$,
			the algorithm returns \emph{``\texttt{Not a $k$-junta}''} with probability at least $2/3$.
	\end{enumerate}\end{flushleft}
	Moreover, the algorithms draws 
	\begin{align}
	q=O\left(\max\left\{\frac{k + \sqrt{n}}{\eps^2 n } ,\hspace{0.03cm} \frac{ k+\sqrt{n} }{\eps\sqrt{n}} \right\} 
	\right)\label{eq:q-setting}
	\end{align}
	samples from $p$ and runs in time $O(q^2 n)$. 
\end{restatable}   

The above theorem improves on a (non-robust) mean tester from \cite{CCKLW20} (which solves the case when $k=0$) in two ways. The first is that since $k \neq 0$, the case $p$ is a $k$-junta may have non-zero mean vector, and our algorithm distinguishes a constant factor gap between the $\ell_2$-norm of mean vectors.\footnote{This gives the robust mean tester a somewhat tolerant testing flavor. Removing the assumption of $p$ being a $k$-junta in the completeness case, and allowing arbitrary distributions with small $\ell_2$-norms on the mean vector would result in an $\Omega(1/\eps^2)$ lower bound (which is always much higher than (\ref{eq:q-setting})). Proof: for $x \in \{-1,1\}^n$, let $p_1$ and $p_2$ be distributions over $\{x, -x\}$ where $p_1$ is uniform and $p_2$ samples $x$ with probability $(1+\eps)/2$. These exhibit a gap in the mean vectors, but are indistinguishable with significantly fewer than $1/\eps^2$ samples.}  The second is that the algorithm runs in time $O(q^2 n)$ as opposed to $n^{O(\log n)}$, and gives optimal query complexity (whereas the result in \cite{CCKLW20} lost a triply-logarithmic factor).


\textbf{Lower bounds for identifying relevant variables and learning junta distributions.}
Both proofs of Theorem \ref{thm:learning-lb} and Theorem \ref{thm:learning-lb-2} follow from a reduction to the
one-way communication complexity of the indexing problem:
Alice receives a uniformly random string $\by\sim\{-1,1\}^m$; Bob receives a uniformly 
random index $\bi\sim [m]$; Alice needs to send a message to Bob so that Bob outputs $\by_{\bi}$.
This problem has a well known $\Omega(m)$ lower bound for any public-coin protocol that succeeds 
with probability at least $2/3$ \cite{MNSW95}.

We focus on Theorem~\ref{thm:learning-lb}, as the proof of Theorem~\ref{thm:learning-lb-2} follows a similar plan. We assume that there~is an algorithm $\calA$ for identifying relevant variables of any $k$-junta distribution
$p$ over $\{-1,1\}^n$ with $q$ general conditioning queries, and similarly to \cite{BCG19}, we will give a communication protocol which simulates $\calA$ to contradict communication complexity lower bounds.
Given an input string $y \in \{-1,1\}^m$ where $m=\Omega(\log \binom{n}{k})$, Alice 
builds a $k$-junta distribution $p_y$ over $\{-1,1\}^n$ such that Bob can decode $y$ by learning relevant variables of $p_y$. 
By \cite{HJMD07,BravermanGarg14} (specifically, Corollary 7.7 in \cite{CCBook}) and the nature of distribution $p_y$,
we compress the naive one-way communication protocol (where Alice sends $q$ samples using $qn$ bits) into a 
public-coin protocol with $O(q\eps^2)+O(1)$ communication bits. 

\textbf{Lower bound for testing junta distributions.}
Our lower bound instances will always consist of product distributions, which simplifies the lower bound proof in two ways. The first way is that subcube conditioning queries may be simulated by random samples, so that it suffices to prove a sample complexity lower bound. The second is that, even uniformity testing (which is the case of $k = 0$), has a lower bound of $\Omega(\sqrt{n} / \eps^2)$ samples \cite{CDKS17, CCKLW20}, so that it suffices to prove a lower bound of $\tilde{\Omega}(k)/\eps^2$. 
We prove an $\tilde{\Omega}(n)/\eps^2$ sample complexity lower bound for testing $k$-junta product distributions with $k=n/2$, and extend the result to all $k\le n/2$ with a padding argument.

The two distributions of ``hard'' instances, $\Dyes$ and $\Dno$, are quite delicate, as they must simultaneously satisfy the following guarantees. (i) A distribution $\bp \sim \Dyes$ is an $(n/2)$-junta product distribution with probability at least $1 - o_n(1)$,  i.e., $\mu(\bp)$ has at most $n/2$ non-zero coordinates (in particular, these are the relevant coordinates). (ii) A distribution $\bp \sim \Dno$ is $\eps$-far from any $(n/2)$-junta product distribution with probability $1-o_n(1)$, i.e., letting $\mu'$ be $\mu(\bp)$ after zeroing out the top half of coordinates, $\|\mu'\|_2 \geq \eps$. (iii) The joint distributions over significantly fewer than $n/\eps^2$ samples from a draw $\bp \sim \Dyes$ and $\bp \sim \Dno$, respectively, are $o_n(1)$ in total variation distance. The constructions proceed by randomly and independently setting $\mu(\bp)_i$ according to one of two possible distributions (one for $\Dyes$ and one for $\Dno$) such that the first $O(\log n / \log \log n)$ moments of each $\mu(\bp)_i$ match when $\bp \sim \Dyes$ and $\bp \sim \Dno$, which we show suffices for condition (iii).\footnote{The method of matching moments for distribution testing tasks is a well-known technique \cite{RRSS09, V11}, where the core is analyzing the solution of a Vandermonde system to construct hard instances. While our plan proceeds in a similar fashion, the specific technical details are rather intricate. In particular, seemingly innocuous changes to the Vandermonde system result in constructions which would not work.}
%

\section{Preliminaries}

We use boldface symbols to represent random variables, and non-boldface symbols for fixed values (potentially realizations of these random variables) --- see, e.g., $\brho$ versus $\rho$. Given  $n\in\N$,~we let $\calU_n$ denote the uniform distribution over $\bits^n$. Usually, as the support of $\calU_{n}$ will be clear from the context, we will drop the subscript and simply write $\calU$. We write $f(n) \lesssim g(n)$ if, for some $c > 0$, $f(n) \leq c \cdot g(n)$ for all $n \geq 1$ (the $\gtrsim$ symbol is defined similarly). We use the notation $\tilde{O}(f(n))$ to denote $O(f(n) \cdot\polylog(f(n)))$, and $\tilde{\Omega}(f(n))$ to denote $\Omega(f(n) / (1+ |\polylog(f(n))|))$.
The notation $[k]$ denotes the set of integers $\{1, \dots, k\}$. 


We introduce two useful operations on a distribution 
$p$ supported on $\bits^n$.

\begin{definition}[Projection]For any set $S\subseteq [n]$, we write $\ol S = [n]\setminus S$ and define the \emph{projected distribution} $p_{\ol S}$ supported on $\smash{\bits^{\ol S}}$ by letting $\by \sim p_{\ol S}$ be drawn as $\by=\bx_{\ol S}$ for $\bx\sim p$.
\end{definition}

\begin{definition}[Restriction]
	We refer to a string $\rho\in \{-1,1,*\}^n$ as a \emph{restriction} 
	and use $\stars(\rho)$ to denote the set of indices $i\in [n]$ with $\rho_i=*$.
	We denote by $p_{|\rho}$ the \emph{restricted} distribution supported on $\smash{\{-1,1\}^{\stars(\rho)}}$ given by $\bx_{\stars(\rho)} $ where $\bx$ is drawn from $p$ \emph{conditioned} on every $i \notin \stars (\rho)$ being set to $\rho_i$.
\end{definition}

The majority of the results in this work consider restrictions $\brho$ drawn randomly from 
one of the distributions that we define next.
\begin{definition}
	Let $n \in \N$ and $p$ be a distribution supported on $\{-1,1\}^n$. Given a set $S \subseteq [n]$ we let $\calD_{S}(p)$ be the distribution over restrictions $\rho \in \{-1,1,*\}^{n}$ given by letting $\brho \sim \calD_{S}(p)$ be sampled according to a sample $\bx \sim p$, and setting for all $i \in [n]$:
	$\brho_i=*$ if $i\in S$ and $\brho_i=\bx_i$ if $i \notin S$.  
	
	For any $\sigma \in (0,1)$ and a ground set $T$, we let $\calS_{\sigma}(T)$ be the distribution supported on subsets $S \subseteq T$ given by letting $\bS \sim \calS_{\sigma}(T)$ be the set which includes each $i \in T$ in $\bS$ independently with probability $\sigma$. We oftentimes write $\calS_{\sigma} = \calS_{\sigma}([n])$ when $n$ is clear from context.
	We let $\calD_{\sigma}(p)$ be the distribution supported on restrictions $\{-1,1,*\}^n$ given by letting $\brho \sim \calD_{\sigma}(p)$ be sampled by first sampling $\bS \sim \calS_{\sigma}$ and then outputting $\brho \sim \calD_{\bS}(p)$.
\end{definition}

\section{Finding Relevant Variables}
\newcommand{\FindRelevantVariables}{\texttt{FindRelevantVariables}}
\newcommand{\VarBudget}{\texttt{VariablesBudget}}
\newcommand{\SampleEdges}{\texttt{SampleEdges}}
\newcommand{\fail}{\texttt{fail}}
\newcommand{\reject}{\texttt{reject}}
\newcommand{\accept}{\texttt{accept}}

In this section we give our algorithm for identifying relevant variables
from junta distributions.
We restate our main structural lemma but delay its proof to Section \ref{sec:structural}.

\mainstructurallemma*

We emphasize that the parameter $n$ in our structural lemma will be set to be the junta 
parameter $k$ later so we need it to hold for small $n$ such as $n=1$, which requires some care
in its proof later.

We restate the main theorem of this section:

\thmidentify*

Theorem \ref{maintheorem} will follow by combining the main algorithmic component, Lemma~\ref{thm:kjunta-smallnorm} stated next, with the main structural lemma (Lemma~\ref{lem:main-structural}).

\begin{lemma} \label{thm:kjunta-smallnorm1} 
	There exists a randomized algorithm, $\emph{\FindRelevantVariables}$, which takes subcube conditional query access to an unknown distribution $p$ supported on $\{-1,1\}^n$, an integer $k\in \N$ and 
	a parameter $\eps \in (0, 1/4]$. The algorithm makes $\tilde{O}(k/\eps^2)\cdot \log n$
	queries and outputs a set $\bJ \subset [n]$ that satisfies the following guarantees:
	\begin{flushleft}\begin{enumerate}
			\item\label{en:first-cond}
			With probability at least $8/9$, for 
			every $i \in \bJ$, there is a restriction $\rho \in \{-1,1,*\}^n$ with $i \in \stars(\rho)$ such that $\mu(p_{|\rho})_i \neq 0$ \emph{(}and thus, $i$ is a relevant variable of $p$\emph{)}; 
			\item\label{en:second-cond} 
			Suppose $p$ is a $k$-junta distribution and let $\sigma=1/2$.
			With probability at least $8/9$, $\bJ$ satisfies
			\begin{align} 
			\Ex_{\brho \sim \calD_{\ol{\bJ}}(p)}\left[ \Ex_{\bnu \sim \calD_{\sigma^j}(p_{|\brho})} \Big[\big\| \mu\big((p_{|\brho})_{|\bnu}\big)\big\|_2  \Big] \right] \leq \eps,\qquad\text{for 
				every $j=1,\ldots, \lceil \log_2 2k \rceil $.}   \label{eq:guarantees}
			\end{align}
	\end{enumerate}\end{flushleft}
\end{lemma}

\begin{proofof}{Theorem~\ref{maintheorem} assuming Lemma~\ref{thm:kjunta-smallnorm1}}
	We execute $\FindRelevantVariables\hspace{0.04cm} (p, k,\tilde{\eps})$ for some 
	parameter $\tilde{\eps}$ to be specified shortly, and upon receiving $\bJ \subset [n]$ outputs $\bJ$.
	We show that when $p$ is a $k$-junta distribution, $\bJ$ satisfies the condition of Theorem \ref{maintheorem} with probability at least $2/3$.
	For this purpose it suffices to show that the condition of Theorem \ref{maintheorem} follows
	from the two conditions of Lemma \ref{thm:kjunta-smallnorm} when $\tilde{\eps}$ is set appropriately.  
	
	Let  
	$J\subset [n]$ be a set of variables for which both conditions of Lemma~\ref{thm:kjunta-smallnorm} hold (with $\tilde{\eps}$ on the right hand side in (\ref{en:second-cond}) instead of $\eps$). Since $p$ is a $k$-junta, we let $I = \{ i_1, \dots, i_k\} \subset [n]$ and $g \colon \{-1,1\}^{k} \to [0,1]$~be such that $p(x) = g(x_{i_1}, \dots, x_{i_k})$. By the first condition, we have $J\subseteq I$ and $|J| \leq k$, since a restriction $\rho \in \{-1,1,*\}^n$ with $i \in \stars(\rho)$ and $\mu(p_{|\rho})_i \neq 0$ certifies that each $i \in J$ is a relevant variable in $p$. Next consider the distribution $h=p_I$ supported on $\smash{\{-1,1\}^I}$ and suppose for the  
	sake of contradiction that $h$ is $\eps$-far from being a junta over variables in $J$. 
	Then by applying Lemma~\ref{lem:main-structural} on $h$ and $J$ with $\sigma = 1/2$ 
	(and noting that parameter $n$ in Lemma \ref{lem:main-structural} is set to $k$), we have 
	\begin{align}
	\frac{\eps}{\log^c (k/\eps)}
	\leq \sum_{j=1}^{\lceil \log_2 2k \rceil}\Ex_{\brho \sim \calD_{\ol{J}}(h)} \left[ \Ex_{\bnu \sim \calD_{\sigma^j}(h_{|\brho})}\Big[\big\|\mu\big((h_{|\brho})_{|\bnu}\big)\big\|_2 \Big] \right],\label{eq:mean-vecs}
	\end{align}
	where $c>0$ is the universal constant from Lemma \ref{lem:main-structural}.
	
	On the other hand, we claim that the right hand side of the inequality above is the same as
	$$
	\sum_{j=1}^{\lceil \log_2 2k\rceil} \Ex_{\brho \sim \calD_{\ol{J}}(p)}\left[ \Ex_{\bnu \sim \calD_{\sigma^j}(p_{|\brho})} \Big[ \big\|\mu\big((p_{|\brho})_{|\bnu}\big)\big\|_2 \Big]\right],
	$$
	after replacing $h$ with $p$.
	This is because $p$ is a $k$-junta over $I$ and thus, the mean vector 
	of $(p_{|\rho})_{|\nu}$ for any restrictions $\rho$ and $\nu$
	always has zeros in entries outside of those in $I$.
	As a result, we have %
	\begin{align*}
	\frac{\eps}{ \log^c (k/\eps)} \leq \sum_{j=1}^{\lceil \log_2 2k \rceil} \Ex_{\brho \sim \calD_{\ol{J}}(p)}\left[ \Ex_{\bnu \sim \calD_{\sigma^j}(p_{|\brho})} \Big[ \big\|\mu\big((p_{|\brho})_{|\bnu}\big)\big\|_2 \Big]\right] \leq \lceil \log_2 2k \rceil \cdot \tilde{\eps},
	\end{align*}
	where we used the second condition of Lemma \ref{thm:kjunta-smallnorm}. Hence, choosing 
	$\tilde{\eps}=\eps/\polylog(k/\eps)$ 
	gives us a contradiction.
	This shows that $h$ is $\eps$-close to being a junta over variables in $J$. Since $p$ 
	is a junta over $I$ and $h=p_{I}$, 
	$p$ is $\eps$-close to being a junta over variables in $J$ as well. 
	
	To finish the proof we note that the bound on the query complexity follows from the fact that we executed $\FindRelevantVariables \hspace{0.04cm}(p,k, \tilde{\eps})$ with $\tilde{\eps}$ picked as above.%
\end{proofof}


\begin{figure}[t!]
	\begin{framed}
		\noindent Subroutine $\FindRelevantVariables \hspace{0.04cm} (p,k, \eps)$
		
		\begin{flushleft}
			\noindent {\bf Input:} Subcube conditioning access to a distribution $p$ supported on $\{-1,1\}^n$, an integer $k\in \N$ and a proximity parameter $\eps \in (0, 1)$.
			
			\noindent {\bf Output:} A set $J \subset [n]$ of variables. 
			
			\begin{enumerate}
				\item Initialize $J = \emptyset$ (and $B=0$, which is used only in the analysis), and let 
				\[ \eps_0 = \frac{\eps}{100\cdot \log^3(k/\eps)}.   \]
				\item\label{en:big-loop} Execute the following while $|J| \leq k$:
				\begin{enumerate}
					\item Initialize $b = 1$.					\item\label{en:small-loop} Repeat the following procedure while $b \le 2k$:
					\begin{enumerate}
						\item[] \hspace{-0.3cm}Increase $B$ by $b$; run
						$\VarBudget \hspace{0.04cm}(p,k, \eps_0, b, J)$, which outputs  $J' \subset [n] \setminus J$.\vspace{0.1cm}
						\begin{enumerate}
							\item If $|J'| \geq b$, update $J$ by adding $b$ elements of $J'$
							to $J$ and go to step~\ref{en:big-loop}. 
							\item If $|J'| < b$, update $b \leftarrow 2 b$ and repeat the loop of step~\ref{en:small-loop}.\vspace{0.1cm}
						\end{enumerate}
					\end{enumerate}
					\item If $b > 2k$, output $J$.
				\end{enumerate}
				\item Output $J$. 
			\end{enumerate}
		\end{flushleft}\vskip -0.14in
	\end{framed}\vspace{-0.2cm}
	\caption{The $\FindRelevantVariables $ subroutine.}\label{fig:Learner}
\end{figure}

We present $\FindRelevantVariables$ in Figure \ref{fig:Learner}.
It uses a subroutine $\VarBudget$ which we describe in Figure \ref{fig:Varbudget} and 
analyze in the lemma below.
\begin{lemma}\label{lem:VarBudget}
	There exists a randomized algorithm, $\emph{\VarBudget} $, which takes subcube condi\-tional query access to an unknown distribution $p$ over $\{-1,1\}^n$, an integer $k\in \N$, a parameter $\eps \in$ $(0, 1/4]$, an integer $b \in [k]$, and a set $J \subset [n]$. 
	It makes
	\[ O\left(\dfrac{b}{\eps^2} \cdot \log^2\left(\frac{k}{\eps}\right) \cdot \log \left(\frac{n}{\eps}\right)\right) \]
	subcube conditional queries, and outputs a set $\bJ' \subset [n] \setminus J$ satisfying the following guarantees: 
	\begin{flushleft}\begin{enumerate}
			\item\label{en:first-condition-lem} With probability at least $1 - (\eps/n)^{9}$, for every coordinate $i \in \bJ'$, there exists a restriction $\rho \in \{-1,1,*\}^n$ with $i \in \stars(\rho)$ such that $\mu(p_{\rho})_i \neq 0$.
			\item\label{en:second-condition-lem} 
			If there exist $j \in [\lceil \log_2 2k \rceil]$ and a real number $\alpha>0$ such that\hspace{0.03cm}\footnote{Note that a trivial 
				necessary  condition for the inequality to hold is $\alpha\le 1$ and $\alpha\ge \eps/\sqrt{b}$.}
			\begin{align} 
			\Prx_{\substack{\brho \sim \calD_{\ol{J}}(p) \\ \bnu \sim \calD_{\sigma^j}(p_{|\brho})}}\left[\hspace{0.04cm} \mu\big((p_{|\brho})_{|\bnu}\big) \text{ contains at least $b$ coordinates of magnitude $\geq \frac{\eps}{\alpha \sqrt{b}}$}\hspace{0.04cm}\right] &\geq \alpha \label{eq:prob-bound} 
			\end{align}
			then the set $\bJ'$ has size at least $b$ with probability at least $1 - (\eps/k)^{9}$.
	\end{enumerate}\end{flushleft}
\end{lemma}

\begin{figure}[t!]
	\begin{framed}
		\noindent Subroutine $\VarBudget\hspace{0.04cm} (p,k, \eps, b, J)$
		
		\begin{flushleft}
			\noindent {\bf Input:} Subcube conditioning access to a distribution $p$ supported on $\{-1,1\}^n$, an integer $k\in \N$, a proximity parameter $\eps \in (0, 1/4]$, a parameter $b \in [k]$ and a set $J \subset [n]$.
			
			\noindent {\bf Output:} A set $J' \subset [n] \setminus J$ which either has size at least $b$, or is empty.
			
			\begin{itemize}
				\item Repeat the following for $j \in [\lceil \log_2 2k \rceil]$ and $a \in \{ 0, \dots, \lfloor \log_2( \sqrt{b}/\eps)\rfloor\}$ with $\alpha = 2^{-a}$:
				\begin{enumerate}
					\item[] \hspace{-0.4cm}Sample $t_\alpha$ many  pairs $\brho \sim \calD_{\ol{J}}(p)$ and $\bnu \sim \calD_{\sigma^j}(p_{|\brho})$, where 
					$$
					t_a= 100 \cdot 2^a\cdot {\log(k/\eps)}=100 \cdot {\log(k/\eps)}\big/{\alpha} 
					$$				
					\begin{enumerate}
						\item\label{en:sample-pairs-restrict} For each sampled pair $(\brho,\bnu)$, take $s_a$ samples $\bx_1, \dots, \bx_{s_a} \sim (p_{|\brho})_{|\bnu}$ with
						\begin{equation}\label{eq:settings}
						s_a=100 \cdot \left( \frac{\alpha^2 b}{\epsilon^2}\right) \cdot \log \left(\frac{n}{\eps}\right)\end{equation} (noting $\alpha^2b/\eps^2\ge 1$) and let $\hat{\mu} \in \R^{\stars(\bnu)}$ be their empirical mean given by
						$$\hat{\mu} = \frac{1}{s_a} \sum_{\ell = 1}^s \bx_{\ell}.$$
						\item Let $\bJ'$ be the set of coordinates $i \in \stars(\bnu)$ satisfying $$|\hat{\mu}_i| \geq \frac{\eps}{2\alpha \sqrt{b}}$$ and output $\bJ'$ if $|\bJ'| \geq b$.
					\end{enumerate}

				\end{enumerate}
				\item If we have not yet produced an output at the end of the main loop, output $\emptyset$.
			\end{itemize}
		\end{flushleft}\vskip -0.14in
	\end{framed}\vspace{-0.2cm}
	\caption{The $\VarBudget$ subroutine.}\label{fig:Varbudget}
\end{figure}
\begin{proof}
	We start with the first condition. 
	We observe that, for the output $\bJ'$ to violate the condition, 
	there must be an execution of step (a) for some $j,a,\rho$ and $\nu$
	such that $\smash{\mu((p_{|\rho})_{|\nu})_i}=0$ for some $i\in \stars(\nu)$ but
	the same coordinate in the average of $s_a$ samples drawn from $(p_{|\rho})_{|\nu}$
	has magnitude at least $\smash{\eps/(2\alpha\sqrt{b})}$ with $\alpha=2^{-a}$.
	Note that this coordinate in the average is just the average of $s_a$ uniformly random bits.

	Via a union bound over coordinates and a Chernoff bound, the probability that one round of 
	step (a) gives a $\bJ'$ in step (b) that violates the condition is at most
	\begin{equation}\label{eq:similar}
	n \cdot \Prx_{\bz_1, \dots, \bz_{s_a} \sim \{-1,1\}}\left[\hspace{0.05cm} \left| \frac{1}{s} \sum_{\ell=1}^{s_a} \bz_{\ell} \right| \geq \frac{\eps}{2\alpha \sqrt{b}} \hspace{0.05cm}\right]
	\le  2n\cdot  \exp\left( -\frac{s_a\eps^2}{8 \alpha^2 b}\right)\le \left(\frac{\eps}{n}\right)^{11}.
	\end{equation}
	With a union bound over all rounds of (a), the probability of $\bJ'$ 
	violating the condition is at most
	\begin{align*}
	\lceil \log_2 2k \rceil \cdot \left(\sum_{a=0}^{\lfloor \log_2( \sqrt{b}/\eps) \rfloor} 100\cdot 2^a\cdot  \log(k/\eps) \right) \cdot \left(\frac{n}{\eps}\right)^{11}  
	\leq O\left(\frac{\sqrt{b}}{\eps}\right)\cdot \log^2\left(\frac{k}{\eps}\right) \cdot \left(\frac{\eps}{n}\right)^{11} \leq \left( \frac{\eps}{n}\right)^{9}.
	\end{align*}
	
	We now turn to the second condition. 
	By assumption there are parameters $j \in [\lceil \log_2 k\rceil] $ and~$\alpha^*>$ $0$ 
	such that (\ref{eq:prob-bound}) holds (which implies that $\eps/\sqrt{b}\le \alpha^*\le 1$).
	Let $$0\le a=\lfloor \log(1/\alpha^*)\rfloor\le \lfloor \log(\sqrt{b}/\eps)\rfloor\quad\text{and}\quad\alpha=2^{-a}$$ so that $\alpha^*\le \alpha\le 2\alpha^*$.
	It suffices to show that during the main loop of 
	$\VarBudget$  with $j $ and $a$,
	at least one of the $t_a$ pairs $\brho$ and $\bnu$ sampled leads to $\bJ'$ with $|\bJ'|\ge b$
	with high probability. 
	
	For 
	this purpose we say  
	a pair $(\rho,\nu)$ of restrictions is \emph{good} if the mean vector of $(p_{|\rho})_{|\nu}$ has at least $b$ coordinates
	of magnitude at least $\smash{\eps/(\alpha^*\sqrt{b})}$.
	It follows from (\ref{eq:prob-bound}) that 
	$\smash{\brho \in \calD_{\ol{J}}(p)}$ and $\smash{\bnu\in \calD_{\sigma^{j }}(p_{|\brho})}$ are good with probability 
	at least $\alpha^*$.
	By virtue of step (a) being repeated $$t_a=100\cdot \log(k/\eps)\big/\alpha\ge 50\cdot  \log(k/\eps)\big/\alpha^*$$ times, we have that with probability at least $1-(\eps/k)^{10}$, at least one of the pairs of restrictions $\brho$ and $\bnu$
	sampled in the main loop of $j $ and $a$ is good.
	
	On the other hand, fix any such good pair $(\rho,\nu)$ and any coordinate $i\in \stars(\nu)$
	with $$\big|\mu((p_{|\rho})_{\nu})_i\big|\ge \eps\big/(\alpha^*\sqrt{b})\ge \eps\big/(\alpha\sqrt{b})$$ 
	since $\alpha\ge  \alpha^*$.
	It follows from a Chernoff bound similar to (\ref{eq:similar}) that every such coordinate $i$ is added to $\bJ'$ with probability 
	at least $1 - (\eps/n)^{10}$. 
	By a union bound over the two bad events, the main loop with $j $ and $a$
	outputs a set of size at least $b$ with probability at least 
	$1-(\eps/n)^{10} - (\eps/k)^{10} \geq 1 - (\eps /k)^{9}$. 

	Finally, the query complexity is bounded by:
	\begin{align*}
	\lceil \log_2 2k\rceil \cdot \sum_{a=0}^{\lfloor \log_2(\sqrt{b}/\eps)\rfloor}  t_as_a&\le   100^2 \cdot \lceil \log_2 2k\rceil \sum_{a=0}^{\lceil \log_2(\sqrt b /\eps)\rceil } {2^a\cdot \log\left(\frac{k}{\eps}\right)}\cdot  \frac{b}{2^{2a}\epsilon^2} \cdot  \log\left(\frac{n}{\eps}\right)\\  & = O\left(\frac{b}{\eps^2} \cdot \log^2\left(\frac{k}{\eps}\right)\cdot  \log\left(\frac{n}{\eps}\right)\right).
	\end{align*}
	as required. This finishes the proof of the lemma.
\end{proof}

Finally we use Lemma \ref{lem:VarBudget} to analyze $\FindRelevantVariables$ and prove 
Lemma~\ref{thm:kjunta-smallnorm}: 

\begin{proofof}{Lemma~\ref{thm:kjunta-smallnorm}}
	To analyze the query complexity, consider an execution of the algorithm $\FindRelevantVariables \hspace{0.04cm}(p,k, \eps)$. Given that all queries are made in calls to 
	$\VarBudget$, the number of queries made by the subroutine at any time is captured by
	$$
	B\cdot O\left(\frac{1}{\eps^2_0} \cdot \log^2\left(\frac{k}{\eps_0}\right)\cdot \log \left(\frac{n}{\eps_0}\right)
	\right)
	=\frac{B}{\eps^2}\cdot \polylog\left(\frac{k}{\eps}\right)\cdot \log n.
	$$
	using $\eps_0 = \eps / \polylog (k/\eps)$.
	So it suffices to show that $B=O(k)$ when the algorithm terminates. 
	To see this is the case we prove by induction that at the end of each loop of (b), we have
	$$
	B\le 2|J|+b.
	$$
	This clearly holds at the beginning (before the first loop of (b)) because
	$B=0$, $b=1$ and $|J|=0$.
	For the induction step, note that each iteration  of step (b)  either 
	(A) increases both $B$ and $|J|$ by $b$ and resets $b$ to $1$; or
	(B) increases $B$ by $b$, $b$ gets doubled and $|J|$ remains the same. 
	As a result, it suffices to bound $b$ and $|J|$ when the algorithm terminates.
	If the algorithm terminates because of line (c),
	then we can bound $b$ by $4k$ and $|J|$ by $k$;
	if the algorithm terminates because of line 3,
	then we can bound $b$ by $2k$ and $|J|$ by $k+b\le 3k$.

	In both cases we have $B\le 2|J|+b\le 8k$.
	This finishes the analysis of the query complexity.

	Towards proving the first guarantee, note that the total number of executions of $\VarBudget $
	is at most the value of $B$ when the algorithm terminates, and we know from the
	analysis above that it is bounded by $8k$.
	We take a union bound over all executions of $\VarBudget $, and deduce that with probability at least $8/9$, every execution satisfies the first condition in Lemma~\ref{lem:VarBudget}, from which $J$
	also satisfies the first condition in Lemma~\ref{thm:kjunta-smallnorm}
	since $J$ only contains coordinates returned by calls to $\VarBudget $.

	To prove the second guarantee, suppose $p$ is a $k$-junta distribution. 
	We can similarly take a union bound over all executions of $\VarBudget $
	and deduce that with probability at least $8/9$,
	every execution satisfies both conditions in Lemma~\ref{lem:VarBudget}.
	Let $J$ be the output of $\FindRelevantVariables$.
	Then similar to the argument above, the first condition in Lemma \ref{lem:VarBudget}
	implies that $J$ contains only relevant variables of $p$ and thus, $|J|\le k$.
	If $|J| = k$, the inequality (\ref{eq:guarantees}) is immediate since all relevant
	variables of $p$ have been identified in $J$ 
	and hence for every $\rho \in \supp(\calD_{\ol{J}}(p))$, $p_{|\rho}$ is uniform. 
	
	Suppose then that $|J| < k$ and note from Figure~\ref{fig:Learner} that 
	the algorithm terminates because of line~(c).
	This implies that for $J$, step (b) executed $\VarBudget \hspace{0.04cm}(p,k,\eps_0, b, J)$ for every $b\le 2k$ being a power of $2$ and $|J'|<b$ for every execution. 
	It then follows from the second guarantee of 
	Lemma~\ref{lem:VarBudget} that, for every $j\in [\lceil \log_2 2k\rceil]$, $b=2^\beta$ with $\beta=0, \ldots,\lfloor\log_2 2k\rfloor$ and every $\alpha>0$, (\ref{eq:prob-bound}) does not hold:
	\begin{align} 
	\Prx_{\substack{\brho \sim \calD_{\ol{J}}(p) \\ \bnu \sim \calD_{\sigma^j}(p_{|\brho})}}\left[\hspace{0.05cm} \Big|\mu\big((p_{|\brho})_{|\bnu}\big)_i\Big| \geq \frac{\eps_0}{\alpha \sqrt{b}} \text{ for at least $b$ coordinates}\hspace{0.05cm}\right] \leq \alpha. \label{eq:prob-ub}
	\end{align}
	We use (\ref{eq:prob-ub}) to show for each $j\in [\lceil \log_2 2k\rceil]$ that
	\begin{align*} 
	\Ex_{\brho \sim \calD_{\ol{J}}(p)}\left[ \Ex_{\bnu \sim \calD_{\sigma^j}(p_{|\brho})} \Big[\big\| \mu\big((p_{|\brho})_{|\bnu}\big)\big\|_2  \Big] \right] \leq \eps.
	\end{align*}
	To this end, we use
	\begin{align}\label{secondsecond}
	\Ex_{\brho \sim \calD_{\ol{J}}(p)}\left[ \Ex_{\bnu \sim \calD_{\sigma^j}(p_{|\brho})} \Big[\big\| \mu\big((p_{|\brho})_{|\bnu}\big)\big\|_2  \Big] \right] \leq \eps_0 + \int_{\eps_0}^{\sqrt{k}} \Prx_{\brho,\bnu}\Big[ \big\|\mu\big((p_{|\brho})_{|\bnu}\big)\big\|_2 \geq \gamma \Big] \hspace{0.06cm}d \gamma
	\end{align} 
	and the following claim; the proof is elementary so we delay its proof to the end.
	
	\begin{claim}\label{simpleclaim}
		Let $x\in [-1,1]^k$ with $\|x\|_2\ge \gamma$ for some $\gamma>0$. 
		Let $t=\lfloor \log_2 2k\rfloor$.
		Then there must be a $\beta=0,1,\ldots, t$ such that
		the number of $i\in [k]$ with $$|x_i|\ge \frac{\gamma}{ 2\sqrt{ 2^{\beta } t}} $$
		is at least $2^\beta$.
	\end{claim}	
	
	Letting $t=\lfloor \log_2 2k\rfloor$.
	Claim \ref{simpleclaim} implies that 	
	\begin{equation}\label{thirdthird}
	\Prx_{\brho,\bnu}\Big[ \big\|\mu\big((p_{|\brho})_{|\bnu}\big)\big\|_2 \geq \gamma \Big]
	\le \sum_{\beta=0}^t\hspace{0.1cm}
	\Prx_{\brho,\bnu}\left[\hspace{0.05cm} \Big|\mu\big((p_{|\brho})_{|\bnu}\big)_i\Big| \geq \frac{\gamma}{2\sqrt{2^{\beta} t}} \text{ for at least $2^{\beta}$ coordinates}\hspace{0.05cm} \right].
	\end{equation}
	Combining (\ref{eq:prob-ub}), (\ref{secondsecond}) and (\ref{thirdthird}), we have
	that the left hand side of (\ref{secondsecond}) is at most
	\begin{align*}
	\eps_0+&\sum_{\beta = 0}^{t}\hspace{0.1cm} \int_{\eps_0}^{\sqrt{k}}\hspace{0.03cm} \Prx_{\brho,\bnu}\left[
	\hspace{0.05cm} \Big|\mu\big((p_{|\brho})_{|\bnu}\big)_i\Big| \geq \frac{\gamma}{2\sqrt{2^{\beta} t}} \text{ for at least $2^{\beta}$ coordinates}\hspace{0.05cm} \right] d \gamma \\
	&\quad\le \eps_0 + 2\eps_0 \sqrt{t}\cdot \sum_{\beta =0}^{t} \hspace{0.1cm} \int_{\eps_0}^{\sqrt{k}} \frac{1}{\gamma}\hspace{0.06cm} d \gamma \leq \eps_0 \left(1 + 2\sqrt{t}(t+1)\cdot \ln\left(\frac{\sqrt{k}}{\eps_0}\right)\right) \leq \eps,
	\end{align*}
	using our choice of $\eps_0=\eps/(100\cdot \log^3(k/\eps))$.
	This finishes the proof of the lemma.
\end{proofof}

\ignore{
	\begin{align*}
	&\Ex_{\brho \sim \calD_{\ol{J}}(p)}\left[ \Ex_{\bnu \sim \calD_{\sigma^j}(p_{|\brho})}\left[ \| \mu((p_{|\brho})_{|\bnu} )\|_2\right]\right] \leq \frac{\eps}{2} + \sum_{\ell = 0}^{\lceil \log_2(2\sqrt{k}/\eps)\rceil} (\eps 2^{\ell}) \Prx_{\substack{\brho \sim \calD_{\ol{J}}(p) \\ \bnu \sim \calD_{\sigma^j}(p_{|\brho}) }} \left[ \|\mu((p_{|\brho})_{|\bnu})\|_2 \geq \eps 2^{\ell-1} \right] \\
	&\qquad\qquad\leq \sum_{\ell = 0}^{\lceil \log_2(\sqrt{k}/\eps)\rceil} \eps \cdot 2^{\ell} \sum_{r = 0}^{\lceil \log_2 k \rceil} \Prx_{\substack{\brho \sim \calD_{\ol{J}}(p) \\ \brho' \sim \calD_{\sigma^j}(p_{|\brho})}}\left[ |\mu((p_{|\brho})_{|\brho'})_i| \geq \frac{\eps \cdot 2^{\ell - 1}}{\sqrt{2^{r} (\lceil \log_2 k\rceil + 1)}} \text{ for $2^{r}$ coordinates} \right] + \frac{\eps}{2} \\
	&\qquad\qquad\leq  \frac{\eps}{2} +\sum_{r = 0}^{\lceil \log_2 k \rceil} \sum_{\ell=0}^{\lceil \log_2 (\sqrt{2^{r}(\log_2(k) + 1)} / \eps) \rceil + 1 } \eps \cdot \dfrac{1}{10\log^3(k/\eps)} \leq \frac{\eps}{10\log(k/\eps)}+\frac{\eps}{2}\le \eps.
	\end{align*}
	where we apply (\ref{eq:prob-ub}) with $b = 2^{r}$ and $\alpha = \sqrt{\lceil \log_2(k)\rceil + 1}/2^{\ell - 1}$. In addition, it suffices to consider $\ell \leq \lceil \log_2(\sqrt{2^r (\lceil \log_2k\rceil+1)} / \eps)\rceil + 1$, since $\mu((p_{|\brho})_{|\brho'})_i \in [-1,1]$.  }

\begin{proofof}{Claim \ref{simpleclaim}}
	Assume for contradiction that this is not the case for every $\beta=0,1,\ldots,t$.
	In particular, it means that no coordinate has $|x_i|\ge \gamma/ (2\sqrt{ t}) $ using the case with $\beta=0$.
	Therefore,
	$$
	\gamma^2\le \|x\|_2^2< 2\cdot \sum_{\beta=1}^t
	2^\beta\cdot \frac{\gamma^2}{ 4\cdot 2^{\beta } t}+k\cdot \frac{\gamma^2}{4\cdot 2^t t}
	\le \frac{\gamma^2}{2}+\frac{\gamma^2}{4t}<\gamma^2,
	$$
	a contradiction.
\end{proofof}

\newcommand{\SampleWalk}{\texttt{SampleWalk}}
\def\bS{\mathbf{S}}	\def\bi{\mathbf{i}}

\section{Lower Bounds for Learning}

The goal of this section is to prove the following lower bounds for the number of subcube conditioning
queries needed by an algorithm to solve the following two tasks
(1) to learn a set of relevant variables of a $k$-junta distribution and
(2) to learn a distribution.

Note that our lower bounds hold for the general conditioning model \cite{CFGM16, CRS15} which allows the algorithm to condition on arbitrary subsets of the domain $\{-1,1\}^n$, rather that only subcubes.

\thmlowerboundone*
\thmlowerboundtwo*

\newcommand{\bcalG}{\boldsymbol{\calG}}

Both proofs of Theorem~\ref{thm:learning-lb} and Theorem~\ref{thm:learning-lb-2} follow from reductions from the communication complexity lower bound of the following indexing problem:
\begin{itemize}
	\item Alice receives a uniformly random string $\by \sim \{-1,1\}^m$. 
	\item Bob receives a uniformly random index $\bi \sim [m]$.
	\item The task is for Alice to send a message to Bob so that Bob outputs $\by_{\bi}$.
\end{itemize}
This problem has a well known $\Omega(m)$ lower bound on the one-way communication of any protocol 
in order for Bob to succeed with probability at least $2/3$ \cite{MNSW95}.

The plan for proving Theorem~\ref{thm:learning-lb} is the following.
Our main goal is to cast the indexing problem as the problem of finding relevant variables.
Let $\calA$ be a \emph{deterministic} algorithm for the task described in Theorem \ref{thm:learning-lb} 
with $q$ general conditioning queries; it will become clear in the proof later that this is 
without loss of generality (so $\calA$ can be viewed as a depth-$q$ decision tree; see Definition \ref{def:trees}).
Setting $m=\Omega(\log {n\choose k})$, 
we show that Alice can use its input string $y\in \{-1,1\}^m$ to 
construct a $k$-junta distribution $p_y$ over $\{-1,1\}^n$ with the following recovery property:
any subset $J\subset [n]$ of no more than $k$ variables such that $p_y$ is $\eps$-close
to a junta distribution over $J$ can be used to recover $y$.
Alice uses private randomness to simulate the execution of $\calA$ on $p_y$
and sends~a~message~to~Bob that contains the sequence of $q$ samples $\bx_1,\ldots,\bx_q$.
The recovery property guarantees that whenever Bob succeeds in
finding relevant variables using $\bx_1,\ldots,\bx_q$, which happens with probability at least $4/5$,
he can use them to recover Alice's string $y$ and then $y_i$.

However, the naive protocol described above has communication complexity $qn$
and we only get $q\ge \Omega(m/n)$ which is insufficient for our goal. 
To compress this protocol, we note that distributions $p_y$ constructed from $y$ are 
in some sense very close to the uniform distribution over $\{-1,1\}^n$.
More formally, we give the following definition of $\eps$-\emph{almost uniform distributions}.

\begin{definition}
	Let $p$ be a probability distribution over $\{-1,1\}^n$ and $\eps \in (0, 1/2)$. We say that $p$~is $\eps$-almost uniform if for every $x \in \{-1,1\}^n$, $|p(x) - 2^{-n}| \leq \eps 2^{-n}$.
\end{definition}

The intuition behind the compression is that a sample from an $\eps$-almost uniform distribution
(even being conditioned on a subset of $\{-1,1\}^n$) carries with it very little information 
(roughly $O(\eps^2)$).
One can then use results from \cite{HJMD07,BravermanGarg14} (also see Corollary 7.7 in \cite{CCBook}) 
to show that the naive one-way private-coin protocol described above can be compressed into a 
public-coin protocol with $O(q\eps^2)+O(1)$ one-way communication bits.
Formally we state the following lemma:

\begin{lemma}\label{lem:compression}
	Let $\calA$ be a deterministic algorithm on distributions over $\{-1,1\}^n$
	that makes $q$ general conditioning queries.
	Then there is a one-way public-coin protocol such that, upon receiving 
	an $\eps$-almost uniform distribution $p$ over $\{-1,1\}^n$,
	Alice sends a message $\bM$ of length
	$
	O(q\eps^2)+O(1)
	$ 
	in the worst case.
	Bob can use $\bM$ to compute a sequence of $q$ strings $\bx_1,\ldots,\bx_q\in \{-1,1\}^n$
	such that the distribution of $(\bx_1,\ldots,\bx_q)$ is $(1/20)$-close
	to the distribution of the sequence of $q$ samples $\calA$
	receives when running on $p$.
\end{lemma}

We give a self-contained proof of Lemma \ref{lem:compression} in Section \ref{sec:compression} since the setting we
work on is more explicit compared to those of \cite{HJMD07,BravermanGarg14}.
The flow of the proof for Theorem \ref{thm:learning-lb-2} is similar.
The key differences lie in the construction of $p_y$ from $y$ for Alice,
and the way Bob recovers $y_i$ using the hypothesis $\hat{p}$ returned by 
the learning algorithm for $k$-junta distributions.
We prove Theorem \ref{thm:learning-lb} and Theorem \ref{thm:learning-lb-2} in Section \ref{sec:learning-lb-1} and \ref{sec:learning-lb-2}, respectively.

\subsection{Proof of Theorem~\ref{thm:learning-lb}}\label{sec:learning-lb-1}

Suppose that $\calA^*$ is a randomized algorithm which, given general conditioning query access to~any unknown $k$-junta distribution $p$ supported on $\{-1,1\}^n$, makes $q$ queries and outputs with probabi\-lity at least $4/5$ a subset $J \subset [n]$ of at most $k$ variables such that $p$ is $\eps$-close to a junta distribution over $J$. So $\calA^*$ can be viewed as a distribution of deterministic algorithms $\calA$.
Let \begin{equation}\label{choiceofm}
m= \left\lfloor\log \binom{n}{k} \right\rfloor=\Omega\left(\log \binom{n}{k}\right).
\end{equation}
Alice will interpret her input string $x \in\bits^m$ in the indexing problem
as a set $S \subset [n]$ of size $k$
and use $S$ to define the following probability distribution $p_S$ over $\{-1,1\}^n$:
\begin{align*}
p_S(x) &= \left\{ \begin{array}{cc} (1+4\eps) 2^{-n} & \prod_{i \in S} x_i = 1 \\[1ex]
(1-4\eps) 2^{-n} & \text{o.w.} \end{array} \right. .
\end{align*}
It follows directly from the definition that $p_S$ is $O(\eps)$-almost uniform.
The following claim gives us the recovery property discussed earlier:

\begin{claim}\label{cl:far-family}
	Suppose that $S \subset [n]$ is a set of size $k$ and $J\ne S \subset [n]$ is a set of size at most $k$. Then we have  $\dtv(p_{S}, g) \geq 2\eps$ for any junta distribution over variables in $J$.
\end{claim}

\begin{proof}
	Notice that since $S$ is of size $k$ and $|J| \leq k$ of size at most $k$, there exists an index $i \in S$ such that $i \notin J$. Consider this fixed $i \in S \setminus J$. We will write the probability mass functions $p_{S}$ and $g$ as functions $\{-1,1\}^J \times \{-1,1\}^{[n] \setminus (J \cup \{i\})} \times \{-1,1\} \to \R_{\geq 0}$, where the first $|J|$ indices correspond to settings of bits in $J$, the second $n - |J| - 1$ coordinates correspond to settings of bits in $[n] \setminus (J \cup \{i\})$, and the last bit determines $i$. We notice that since $g$ is a junta over variables in $J$, for any $y \in \{-1,1\}^{J}$ and any two $u_1,u_2 \in \{-1,1\}^{[n] \setminus (J\cup\{i\})}$ and $v_1,v_2 \in \{-1,1\}$, $g(y, u_1, v_1) = g(y, u_2, v_2)$. Furthermore, by definition of $p_S$, $|p_{S}(y, u_1, v_1) - p_S(y, u_1, v_2)| = 8\eps 2^{-n}$ whenever $v_1 \neq v_2$. Hence,
	\begin{align*}
	\dtv(p_{S}, g) &= \frac{1}{2} \sum_{x \in \{-1,1\}^n} \left| p_S(x) - g(x)\right| \\
	&= \frac{1}{2} \sum_{y \in \{-1,1\}^J} \sum_{u \in \{-1,1\}^{[n]\setminus (J \cup \{i\})}} \left( |p_S(y, u, 1) - g(y, u, 1)| + |p_S(y,u,-1) - g(y,u,-1)|\right) \\
	&\geq \frac{1}{2} \sum_{y \in \{-1,1\}^J} \sum_{u \in \{-1,1\}^{[n] \setminus (J \cup \{i\})}} |p_{S}(y,u,1) - p_{S}(y,u,-1)| = 2\eps.
	\end{align*}
	This finishes the proof of the claim.
\end{proof}

As a consequence of Claim~\ref{cl:far-family}, we obtain the following corollary.
\begin{corollary}\label{cl:output-set}
	Let $S \subset [n]$ be any set of size $k$, and let $J$ be any set of size at most $k$ such that $p_{S}$ is $\eps$-close to a junta distribution over $J$. Then we must have $J = S$.
\end{corollary}

\begin{proof}
	Let $g$ be the closest junta over $J$ to $p_S$, and suppose for the sake of contradiction, that $J \neq S$. Then, we apply Claim~\ref{cl:far-family} which says that $\dtv(p_S, g) \geq 2\eps$, giving the desired contradiction.
\end{proof}

\ignore{
	\begin{definition}
		For any $\eps \in (0,1)$ and $n,k \in \N$, an algorithm $\calA$ which learns the $k$ relevant variables of a $k$-junta distribution making $q$ conditional sampling queries is specified by a rooted depth-$q$ tree, where every non-leaf node $v$ has an associated set $A_v \subset \{-1,1\}^n$, as well as a child for each $x \in A_v$. Every leaf node contains a set $J \subset [n]$ of size at most $k$. An execution of the algorithm on a distribution $p$ corresponds to a walk down the tree, where at each node, we receive a sample $\bx \sim p$ conditioned on $\bx$ lying in $A_v$, and we proceed to the child of $v$ specified by $\bx$. Once we reach a leaf, we output the set $J$.
	\end{definition}
	
	Notice that the assumption that there exists an algorithm $\calA$ for learning the relevant variables of a $k$-junta distribution making $q$ queries, corresponds to the existence of such a depth-$q$ which outputs a set $J \subset [n]$ of size at most $k$ such that $p$ is $\eps$}

We are now ready to prove Theorem \ref{thm:learning-lb} by following the plan described earlier.

\begin{proofof}{Theorem~\ref{thm:learning-lb}}
	The proof proceeds via a reduction from the two-party one-way communication problem of indexing. 
	With $m$ chosen in (\ref{choiceofm}) Alice and Bob agree on a fixed 
	injective map from $\bits^m$ to subsets of $[n]$ of size $k$.    
	Alice will interpret her input string $x\in \bits^n$ as a subset $S\subset [n]$ of size $k$ using this map.
	Given that $\calA^*$ is a distribution of deterministic algorithms,
	there exists a $q$-query deterministic algorithm $\calA$ such that 
	\begin{equation}\label{fgggg}
	\Pr_{\bx\sim \bits^m}\big[\calA(p_\bS)\ \text{returns $\bS$}\big]\ge 4/5,
	\end{equation}
	where $\bx$ is drawn uniformly at random and $\bS\subset [n]$ is its corresponding subset of size $k$.
	Alice and Bob agree on such a $q$-query deterministic algorithm $\calA$.
	
	Now we describe the protocol.
	Given $x\in \bits^m$, Alice uses it to construct $p_S$ over $\{-1,1\}^m$ which is
	$O(\eps)$-almost uniform. 
	She uses Lemma \ref{lem:compression} to send a message $\bM$ of length $O(q\eps^2)+O(1)$ to Bob 
	so that Bob can use $\bM$ to obtain a sequence of $q$ strings $\bx_1,\ldots,\bx_q\in \{-1,1\}^n$
	such that the latter has distribution $(1/20)$-close to the distribution of the sequence of $q$ samples $\calA$
	receives when running on $p_S$.
	It follows from~(\ref{fgggg}) that when $\bx\sim\bits^m$,
	Bob successfully recovers $\bS$ (and thus, $\bx$ using the map they agreed on)
	by simulating $\calA$ on $\bx_1,\ldots,\bx_q$ 
	with probability at least $4/5-1/20>2/3$.
	By the $\Omega(m)$ lower bound on the indexing problem, we obtain the desired claim using (\ref{choiceofm}).

\end{proofof} 

\subsection{Proof of Theorem~\ref{thm:learning-lb-2}}\label{sec:learning-lb-2}

The lower bound $\Omega(\log {n\choose k}/\eps^2)$ follows trivially from Theorem \ref{thm:learning-lb}.
To see this, we can first learn~$p$ to within $\eps/2$ total variation distance.
Let $\hat{p}$ be the hypothesis distribution that the algorithm returns.
Then we can find its closest $k$-junta distribution $p'$ and let $S$ be the set of  
relevant variables of $p'$ with $|S|\le k$.
The algorithm can return $S$ since $\dtv(p,p')\le \dtv(p,\hat{p})+\dtv(\hat{p},p')\le \eps$.

We focus on the second part of the lower bound $\Omega(2^k/\eps^2)$ in the rest of the proof.
Note that we may assume that $k$ is asymptotically large; otherwise the second part is 
dominated by the first part. 
We follow the same flow.
Suppose that $\calA^*$ is a randomized algorithm which, given general conditioning query access to~any unknown $k$-junta distribution $p$ supported on $\{-1,1\}^n$, makes $q$ queries and outputs with probabi\-lity at least $4/5$ a hypothesis distribution $\hat{p}$ such that $\dtv(p,\hat{p})\le \eps$. 

We say a Boolean function $f:\{-1,1\}^k\rightarrow \{-1,1\}$ is \emph{good} if the number of $1$-entries
in $f$ is between $2^k/3$ and $2^{k+1}/3$.
Let $G_k$ be the set of good Boolean functions. Then it follows from Chernoff bound that 
$
\smash{|G_k|\ge 2^{2^k}(1-o_k(1)).}
$
We set $m=2^k$ and Alice interprets her input string $y\in\bits^m$ in the indexing problem 
as a good Boolean function $f:\{-1,1\}^k\rightarrow \bits$ 
by fixing a bijection between $[m]$ and $\{-1,1\}^k$ and interpreting $y$ as the truth table of $f$.

Given a string $y\in \{-1,1\}^m$ and its corresponding $f:\bits^k\rightarrow \bits$,
letting $I(y)$ be the number of $1$-entries in $f$,
Alice constructs the following $k$-junta distribution $p_y$ over $\{0,1\}^n$:
\begin{align*}
p_{y}(x)= \left\{ \begin{array}{ll} 2^{-n}\left(1+40\eps\cdot \frac{2^k}{I(y)}\right) & \text{if}\ f(x_1,\ldots,x_k) = 1 \\[1.5ex]
2^{-n}\left(1-{40\eps}\cdot \frac{2^k}{2^k - I(y)}\right) & \text{if}\ f(x_1,\ldots,x_k) = -1 \end{array} \right. 
\end{align*}
Note that when $f$ is good, $p_y$ is an $O(\eps)$-almost uniform $k$-junta distribution; as it becomes
clear later Alice constructs $p_y$ only when $f$ is good.
The following claim gives us the recovery property:

\begin{claim}\label{cl:right-val}
	Given a good $y\in \{-1,1\}^m$ and $p_y$ defined above, let $\hat{p}$ be any distribution on $\{-1,1\}^n$ which has $\dtv(p_{y}, \hat{p}) \leq \eps$. Then,
	\begin{align*}
	\Prx_{\bx \sim \{-1,1\}^n}\Big[ \sign\left(\hat{p}(\bx) - 2^{-n} \right) \neq \sign\left( p_{y}(\bx) - 2^{-n}\right)\Big] \leq \frac{1}{20}.
	\end{align*}
\end{claim}
\begin{proof}
	Notice that for every $x \in \{-1,1\}^n$ where $\sign\left(\hat{p}(x) - 2^{-n} \right) \neq \sign\left( p_{y}(x) - 2^{-n}\right)$, we have $|\hat{p}(x) - p_{y}(x)| \geq 40\eps \cdot 2^{-n}$. Hence, 
	\begin{align*}
	\hspace{-0.2cm}\eps &\geq \dtv(p_y, \hat{p}) = \frac{1}{2} \sum_{x \in\{-1,1\}^n} |p_y(x) - \hat{p}(x)| \geq 20\eps \cdot \Prx_{\bx \sim \{-1,1\}^n}\Big[ \sign\left(\hat{p}(\bx) - 2^{-n} \right) \neq \sign\left( p_{y}(\bx) - 2^{-n}\right)\Big].
	\end{align*}
	This finishes the proof of the claim.
\end{proof}

\def\by{\mathbf{y}}

\begin{proofof}{Theorem~\ref{thm:learning-lb-2}}
	Again, the proof proceeds via a reduction from the two-party one-way communication problem of indexing 
	over $\{-1,1\}^m$ where $m=2^k$. 
	Let $y\in \{-1,1\}^m$ be the input string of Alice.
	As alluded to earlier, in the case that $y$ is not good, Alice just aborts the protocol
	and they fail the task with probability $o_k(1)$ because $y$ is drawn uniformly at random from $\{-1,1\}^m$.
	In the case that $y$ is good, Alice uses it to construct $p_y$, 
	a $k$-junta distribution over $\{-1,1\}^n$ that is $O(\eps)$-almost uniform.
	Given that $\calA^*$ is a randomized algorithm for learning $k$-junta distributions over $\{-1,1\}^n$,
	there exists a deterministic algorithm with $q$ general conditioning queries such that
	$$
	\Pr_{\by}\big[\calA(p_\by)\ \text{returns a hypothesis that is $\eps$-close to $p_{\by}$}\big]\ge 4/5,
	$$
	where $\by$ is uniform over good strings. Alice and Bob agree on such an $\calA$.

	The protocol goes as before.
	When $y$ is good,
	Alice uses Lemma \ref{lem:compression} to send a message $\bM$ of length $O(q\eps^2)+O(1)$ to Bob 
	so that Bob can use $\bM$ to obtain a sequence of $q$ strings $\bx_1,\ldots,\bx_q\in \{-1,1\}^n$
	such that their distribution is $(1/20)$-close to the distribution of
	the sequence of $q$ samples $\calA$
	receives when running on $p_S$.
	It follows from~(\ref{fgggg}) that when $\by\sim \{-1,1\}^m$,
	Bob successfully learns a hypothesis distribution $\hat{\textbf{{p}}}$ that 
	is $\eps$-close to $p_{\by}$, 
	by simulating $\calA$ on $\bx_1,\ldots,\bx_q$, 
	with probability at least $4/5-1/20-o_k(1)$.
	We now apply Claim~\ref{cl:right-val} to conclude that if this occurs, Bob can output
	the correct $\bi$-th bit of $\by$
	with  probability at least $9/10$ given that $\bi$ is independent and uniform..
	As a result, over the randomness of $\by$ and $\bi$,
	Bob outputs the correct $\by_\bi$ with probability at least $4/5-1/20-o_k(1)-1/20\ge 2/3$.
	By the $\Omega(m)=\Omega(2^k)$ lower bound on the indexing problem, we obtain the desired claim.
\end{proofof}

\subsection{Compressing batches of conditional samples}\label{sec:compression}

We prove Lemma \ref{lem:compression} in the rest of the section.
%
Recall that $\calA$ is a deterministic (adaptive) algorithm, where each query (a subset of $\{-1,1\}^n$) 
depends on all samples received from previous queries.

We use the following definition to capture such a $q$-query deterministic algorithm:

\begin{definition}\label{def:trees}
	For $n,q \in \N$, we say a $q$-\emph{query tree} $\calT$ is a rooted depth-$q$ tree. Every non-leaf node $v \in \calT$ contains a subset $A_v \subseteq \{-1,1\}^n$, as well as a child node $v_x$ for every $x \in A_v$. Given a distribution $p$ over $\{-1,1\}^n$, an \emph{execution} of $\calT$ on $p$ is a random walk $(v_1,\dots, v_q)$ down the tree, specifying a sequence of $q$ samples $(\bx_1,\dots, \bx_q)$: starting at the root node and proceeding down the tree, for the current node $v_i$, sample $\bx_i \sim p$ conditioned on $\bx_i \in A_{v_i}$, and let $v_{i+1} = (v_i)_{\bx_i}$. Let $\calE_{p, \calT}$ be the distribution supported on $(\{-1,1\}^n)^q$ which outputs the samples $(\bx_1, \dots, \bx_q)$ of an execution of $\calT$ on $p$.
\end{definition}

We consider a protocol, $\SampleWalk$ which, without communication, generates an execution of a
given $q$-query tree $\calT$, and Alice decides whether or not to ``accept'' the samples
at the end. In more detail, $\SampleWalk$ takes as input a distribution $p$ over $\{-1,1\}^n$, a $q$-query tree $\calT$, and an error tolerance $\delta \in (0,1)$, and using public randomness, will output a root-to-leaf walk of $\calT$ specified by nodes $(v_1, \dots, v_q)$ and $(\bx_1, \dots, \bx_q)$, or ``reject''. The protocol, $\SampleWalk$ follows the ``rejection sampling'' paradigm. (See Figure~\ref{fig:protocol} for a precise description of the protocol.)

\begin{figure}[H]
	\begin{framed}
		\noindent Protocol $\SampleWalk \hspace{0.04cm} (p, \calT, \delta)$
		
		\begin{flushleft}
			\noindent {\bf Input:} A distribution $p$ supported on $\{-1,1\}^n$, a $q$-query tree $\calT$, and a parameter $\delta \in (0,1)$. Furthermore, we assume access to a public string of infinite uniformly random bits.
			
			\noindent {\bf Output:} A root-to-leaf walk down the decision tree $\calT$ specified by nodes $(v_1, \dots, v_q)$ and samples $(\bx_1, \dots,  \bx_q)$, or ``reject''.\protect\footnotemark 
			
			\begin{enumerate}
				\item\label{en:prot-1} Starting at the root of $\calT$ and walking down the tree, Alice considers the current node in $v \in \calT$, and the query $A_v \subset \{-1,1\}^n$. She uses public randomness to generate a sample $\bx_v \sim A_v$ drawn \emph{uniformly} from $A_v$, and considers the child node of $\calT$ specified by $\bx_v$. Notice that this builds a walk $(v_1, \dots, v_q)$ and $(\bx_1,\dots, \bx_q)$, and in particular, this step is completely independent from $p$, and draws a sample from $\calE_{\calU, \calT}$.
				\item\label{en:prot-3} Alice samples a private bit which is $1$ with probability 
				\[ \min\left(1, \delta \cdot \frac{\calE_{p, \calT}(\bx_1,\dots,\bx_q)}{\calE_{\calU, \calT}(\bx_1,\dots, \bx_q)}  \right)  \] 
				and $-1$ otherwise. If Alice's sampled bit is $1$, Alice ``accepts'' the sample $(\bx_1, \dots, \bx_q)$ and the nodes $(v_1, \dots, v_q)$, if it is $-1$, Alice ``rejects''. 
			\end{enumerate}
		\end{flushleft}\vskip -0.14in
	\end{framed}\vspace{-0.2cm}
	\caption{The $\SampleWalk$ Protocol.}\label{fig:protocol}
\end{figure}
\footnotetext{We note that outputting $(v_1, \dots, v_q)$ is unnecessary, as the samples $(\bx_1, \dots, \bx_q)$ uniquely determine a root-to-leaf walk down the tree $\calT$. We maintain the notation just for notational simplicity.}

\begin{definition}
	For a $q$-query tree $\calT$, we let $\calD_{p, \calT,\delta}^{\circ}$ be a distribution supported on $(\{-1,1\}^n)^q \cup \{ \bot \}$ given by the samples $(\bx_1,\dots,\bx_q)$ forming the output of one execution of $\emph{\SampleWalk}(p, \calT, \delta)$, or $\bot$ if it outputs ``reject''. We let $\calD_{p, \calT, \delta}$ be the distribution $\calD_{p, \calT, \delta}^{\circ}$ conditioned on it not outputting $\bot$.
\end{definition}


\begin{lemma}\label{lem:communication-compression}
	There exists a sufficiently small constant $\zeta \in (0,1)$ such that for any $\eps,\delta \in (0,1/2)$ and
	\[ q \leq \left\lfloor \frac{\zeta \log(1/\delta)}{\eps^2} \right\rfloor, \]
	the following holds. Let $\calT$ be a $q$-query tree and $p$ be $\eps$-almost uniform. Then,
	\begin{align*}
	\dtv(\calD_{p, \calT, \delta}, \calE_{p, \calT}) \leq \delta \qquad\text{and}\qquad \Prx\left[ \calD_{p,\calT,\delta}^{\circ} \text{ outputs $\bot$}\right] \leq 1 - \delta / 2.
	\end{align*}
\end{lemma}

\begin{proof}
	In particular, notice that in order for an execution of $\SampleWalk(p, \calT,\delta)$ to output ``reject'', two events must occur:
	\begin{itemize}
		\item The first event is that the samples $(\bx_1, \dots, \bx_q)$ sampled in Step~\ref{en:prot-1} satisfy 
		\begin{align}
		\calE_{p, \calT}(\bx_1,\dots, \bx_q) < \calE_{\calU, \calT}(\bx_1,\dots,\bx_q) \cdot \frac{1}{\delta}. \label{eq:large-prob}
		\end{align}
		\item The second event is that a random bit sampled in Step~\ref{en:prot-3} is set to $-1$, and the probability that his occurs is
		\begin{align*} 
		1 - \delta \cdot \frac{\calE_{p, \calT}(\bx_1,\dots, \bx_q)}{\calE_{\calU,\calT}(\bx_1,\dots, \bx_q)}. 
		\end{align*}
	\end{itemize}
	We let $\calR \subset (\{-1,1\}^n)^q$ be the set of strings which satisfy (\ref{eq:large-prob}), i.e.,
	\[ \calR = \left\{ (x_1,\dots, x_q) \in (\{-1,1\}^n)^q : \calE_{p,\calT}(x_1,\dots, x_q) < \frac{1}{\delta} \cdot \calE_{\calU,\calT}(x_1,\dots,x_q) \right\}, \]
	and notice that 
	\begin{align}
	\Prx\left[\calD_{p,\calT,\delta}^{\circ} \text{ outputs $\bot$} \right] = \sum_{x \in \calR} \calE_{\calU, \calT}(x)\left( 1 - \delta \cdot \dfrac{\calE_{p,\calT}(x)}{\calE_{\calU, \calT}(x)}\right) = \Prx_{\bx \sim \calE_{\calU, \calT}}[\bx \in \calR] - \delta \cdot \Prx_{\bx \sim \calE_{p, \calT}}\left[ \bx \in \calR \right], \label{eq:rejection-prob}
	\end{align}
	so for simplicity in the notation, let 
	\[ \alpha \eqdef \Prx_{\bx \sim \calE_{\calU, \calT}}\left[ \bx \in \calR\right] \qquad\text{and}\qquad \beta \eqdef\Prx_{\bx \sim \calE_{p, \calT}}\left[\bx \in \calR\right].\]
	Furthermore, whenever $x \in \calR$, 
	\begin{align*}
	\calD_{p,\calT,\delta}(x) &= \sum_{k=1}^{\infty} \calE_{\calU, \calT}(x) \cdot \left( \delta \cdot \frac{\calE_{p, \calT}(x)}{\calE_{\calU, \calT}(x)} \right) \cdot \calD_{p, \calT,\delta}^{\circ}(\bot)^{k-1} = \delta \cdot \calE_{p, \calT}(x) \left( \frac{1}{1-\calD_{p,\calT,\delta}^{\circ}(\bot)}\right) \\
	&= \left( \frac{\delta}{1 - \alpha + \delta \beta}\right) \calE_{p,\calT}(x),
	\end{align*}
	and whenever $x \notin \calR$, Step~\ref{en:prot-3} always accepts the sample, so 
	\begin{align*}
	\calD_{p, \calT,\delta}(x) &= \left( \frac{1}{1 - \alpha + \delta \beta} \right) \cdot \calE_{\calU,\calT}(x).
	\end{align*}
	Thus, we may write
	\begin{align}
	\dtv\left(\calD_{p,\calT,\delta}, \calE_{p,\calT} \right) &= \frac{1}{2} \sum_{x \in (\{-1,1\}^n)^q}\left| \calD_{p,\calT,\delta}(x) - \calE_{p,\calT}(x)\right| \nonumber \\
	&\leq \frac{1}{2} \sum_{x \notin \calR} \left( \calD_{p,\calT,\delta}(x) + \calE_{p,\calT}(x)\right) + \frac{1}{2} \sum_{x \in \calR} \calE_{p, \calT}(x) \left| \frac{\delta}{1 - \alpha + \delta \beta} - 1\right|\nonumber  \\
	&= \frac{1}{2} \left(\dfrac{1 - \alpha}{1 - \alpha + \delta \beta} + (1 - \beta)\right) + \frac{1}{2} \beta \left| \frac{\delta (1-\beta) - (1 - \alpha)}{1 - \alpha + \delta \beta}\right|,\label{eq:dtv-compute}
	\end{align}
	so it suffices to show
	\begin{align*}
	1 - \delta^2/2 \leq \alpha , \beta\leq 1
	\end{align*}
	in order to conclude that (\ref{eq:dtv-compute}) is at most $\delta$, and that (\ref{eq:rejection-prob}) is at most $1 - \delta/2$. In order to do so, we use the fact that $p$ is $\eps$-almost uniform to upper bound $1 - \alpha$ and $1 - \beta$. Notice that if $x \notin \calR$, then, considering the unique path $(v_1, \dots, v_q)$ in $\calT$ specified by $x$, we have
	\begin{align}
	\frac{1}{\delta} \leq \dfrac{\calE_{p, \calT}(x)}{\calE_{\calU, \calT}(x) } &= \prod_{i=1}^q \left( \dfrac{p(x_i)}{\frac{1}{|A_{v_i}|} \sum_{y \in A_{v_i}} p(y)}\right) = \prod_{i=1}^q \left(1 + \dfrac{p(x_i) - \Ex_{\bz \sim A_{v_i}}[p(\bz)]}{\Ex_{\bz \sim A_{v_i}}[p(\bz)]} \right) \nonumber \\
	&\leq \exp\left( \sum_{i=1}^q \dfrac{p(x_i) - \Ex_{\bz \sim A_{v_i}}[p(\bz)]}{\Ex_{\bz \sim A_{v_i}}[p(\bz)]} \right). \label{eq:quantity}
	\end{align}
	We first upper bound $1-\alpha$ by considering the random sequence $\bY_1, \dots, \bY_q$ generated by starting at the root $v_1$ and walking down the tree $\calT$, while sampling $\bx_i \sim A_{v_i}$, setting $\bY_i = (p(\bx_i) - \Ex_{\bz\sim A_{v_i}}[p(\bz)]) / \Ex_{\bz\sim A_{v_i}}[p(\bz)]$, and letting $v_{i+1} = (v_{i})_{\bx_i}$. We upper-bound $1-\alpha$ by giving an upper bound for the probability that $\sum_{i=1}^q  \bY_i \geq \ln(1/\delta)$, which in turn upper bounds $1-\alpha$ by (\ref{eq:quantity}). Notice that partial sums $\{ \sum_{i=1}^t \bY_i \}_{t \in [q]}$ form a 0-centered martingale, and since $p$ is $\eps$-almost uniform, 
	\begin{align*}
	|\bY_i| \leq \max_{\substack{v \in \calT \\ x \in A_{v}}} \left| \dfrac{p(x) - \Ex_{\bz \sim A_v}\left[ p(\bz)\right]}{\Ex_{\bz \sim A_v}[p(\bz)]} \right| \leq \max_{\substack{v \in \calT \\ x \in A_v}} \left| \dfrac{ \Ex_{\bz \sim A_v}[p(x) - p(\bz)]}{\Ex_{\bz \sim A_{v}}[p(\bz)]} \right| \leq \frac{2\eps}{1 -\eps} \leq 4\eps.
	\end{align*}
	We may apply Azuma's inequality to conclude
	\begin{align*}
	\Prx_{\bx \sim \calE_{\calU, \calT}}\left[ \sum_{i=1}^q \bY_i \geq \ln(1/\delta) \right] &\leq \exp\left( -\frac{\ln^2(1/\delta)}{2 \cdot 16 \eps^2 \cdot q} \right) \leq \delta^2 / 2
	\end{align*}
	by setting of $q$ with $\zeta$ being a sufficiently small constant, and hence lower bounds $\alpha$ by $1 - \delta^2 / 2$. In order to upper bound $1- \beta$, we consider the sequence of random variables $\bY_1', \dots, \bY_q'$ generated by starting at the root $v_1$ and walking down the tree $\calT$, but now we sample $\bx_i \sim p$ conditioned on $\bx_i \in A_{v_i}$, setting $\bY_i = (p(\bx_i) - \Ex_{\bz \sim A_{v_i}}[p(\bz)]) / \Ex_{\bz \sim A_{v_i}}[p(\bz)]$, and writing
	\[ \bY_i' = \bY_i - \dfrac{\Ex_{\bz' \sim p}[p(\bz') \mid \bz' \in A_{v_i}] - \Ex_{\bz \sim A_{v_i}}[p(\bz)]}{\Ex_{\bz \sim A_{v_i}}[p(\bz)]}, \]
	where the subsequent node $v_{i+1} = (v_i)_{\bx_i}$. Notice that now the partial sums $\{ \sum_{i=1}^t \bY_i' \}_{t \in [q]}$ have expectation $0$, form a martingale, where $\bY_i'$ are obtained by shifting $\bY_i$ by its expectation, $\Ex_{\bx \sim \calE_{p, \calT}}[\bY]$. Furthermore, we may upper bound this shift by importance sampling,
	\begin{align}
	\dfrac{\Ex_{\bz' \sim p}\left[ p(\bz') \mid \bz' \in A_{v_i}\right] - \Ex_{\bz \sim A_{v_i}}\left[ p(\bz)\right]}{\Ex_{\bz \sim A_{v_i}}[p(\bz)]} &= \dfrac{\Ex_{\bz \sim A_{v_i}}[p(\bz)^2] - \Ex_{\bz \sim A_{v_i}}[p(\bz)]^2}{\Ex_{\bz \sim A_{v_i}}[p(\bz)]^2 } \nonumber \\
	&= \dfrac{\Ex_{\bz \sim A_{v_i}}\left[ \left( p(\bz) - \Ex_{\bz' \sim A_{v_i}}[p(\bz')]\right)^2\right]}{\Ex_{\bz \sim A_{v_i}}[p(\bz)]^2} \leq \frac{4\eps^2}{1-\eps} \leq 8\eps^2. \label{eq:bias}
	\end{align}
	so that similarly to the computation above, $|\bY_i'| \leq 4 \eps + 8\eps^2 \leq 12 \eps$. We may again, apply Azuma's inequality, where we notice that the expectation of 
	\begin{align*}
	\Prx_{\bx \sim \calE_{p,\calT}}\left[ \sum_{i=1}^q \bY_i \geq \ln(1/\delta) \right] &\leq \Prx_{\bx \sim \calE_{p,\calT}}\left[\sum_{i=1}^q \bY_i' \geq \ln(1/\delta) - 8q\eps^2 \right] \\
	&\leq \Prx_{\bx \sim \calE_{p,\calT}}\left[ \sum_{i=1}^q \bY'_i \geq \ln(1/\delta) / 2\right] \leq \exp\left(- \frac{\ln^2(1/\delta)}{2 \cdot 4 \cdot 144 \eps^2 q} \right) \leq  \delta^2 / 2,
	\end{align*}
	where we used a small enough constant $\zeta > 0$ so that $8 q \eps^2 \leq \ln(1/\delta) / 2$, as well as for the final inequality to hold.
\end{proof}

We now use Lemma \ref{lem:communication-compression} to prove Lemma \ref{lem:compression}:

\begin{proofof}{Lemma \ref{lem:compression}}
	We start with the easy case when $q<1/\eps^2$. In this case, we apply Lemma~\ref{lem:communication-compression} with $\delta = 1 / (40)^{1/\zeta}$. Notice that $q\leq \lfloor \zeta \log(1/\delta) / \eps^2 \rfloor$, so we let $\calT$ be $\calA$, and Lemma~\ref{lem:communication-compression} implies a single call to $\SampleWalk(p, \calA, \delta)$ succeeds in outputting a sample $(\bx_1, \dots, \bx_q)$ from $\calD_{p, \calA, \delta}$ with probability at least $\delta / 2$, and if it does succeed, the output distribution is at most $\delta$-far from the distribution producing a sequence of $q$ samples an execution of $\calA$ on $p$. Alice and Bob use public randomness to execute $\SampleWalk(p, \calA, \delta)$ for $t = O(1/\delta)$ iterations, and Alice communicates the index of the first execution where $\SampleWalk(p,\calA,\delta)$ did not output ``reject'', or the final index if all executions outputted ``reject''. Notice that the distribution of the first time $\SampleWalk(p_{S}, \calT, \delta)$ accepts is exactly $\calD_{p_{S},\calT, \delta}$. Furthermore, this uses $O(\log(1/\delta)) = O(1)$ bits of communication, and that the total variation distance between the samples $(\bx_1,\dots, \bx_q)$ from this protocol and an execution of $\calA$ on $p$ is at most $\delta + (1-\delta / 2)^{t} \leq 1/20$, where the first $\delta$ captures the case when some $\SampleWalk(p,\calA,\delta)$ does not reject, and $(1-\delta/2)^t$ is the probability that all $\SampleWalk(p, \calA, \delta)$ output ``reject''.
	
	When $q \geq 1/\eps^2$, we apply Lemma~\ref{lem:communication-compression} with 
	\[ \delta = \dfrac{1}{\eps^2 q \cdot 100^{1/\zeta}}. \]
	As per setting of (what we refer to as $q'$) from Lemma~\ref{lem:communication-compression}, where $q' = \lfloor \zeta \log(1/\delta)/\eps^2 \rfloor \geq 2$ and hence $q' \geq \zeta \log(1/\delta) / (2\eps^2)$.  Alice and Bob break up the $q$-query algorithm $\calA$ into $\lceil q / q' \rceil$ many $q'$-query trees. The trees are adaptively chosen so as to simulate an execution of $\calA$. For each $q'$-query tree $\calT$, Alice and Bob use public randomness to execute $\SampleWalk(p_S, \calT, \delta)$ for $O(1/\delta)$ iterations such that with probability at least $1/2$, at least one accepts. Alice then communicates $O(\log(1/\delta))$ bits to Bob, indicating the first index where $\SampleWalk(p_{S},\calT, \delta)$ accepts, or a special message indicating none accepted. If some execution accepts, then Bob re-constructs the samples $\bx_1,\dots,\bx_{q'}$ utilizes those samples to simulate the walk down $\calT$. If $\SampleWalk(p_{S},\calT,\delta)$ never accepts, Alice and Bob try again on the same tree. 
	
	Notice that by Lemma~\ref{lem:communication-compression}, since the distribution over the leaves of $\calT$ is $\delta$-close in total variation distance from that of a true execution of $\calT$ on $p$, after $\lceil q / q' \rceil$ successive executions of Lemma~\ref{lem:communication-compression}, the distribution over the leaves of $\calA$ is at most $\delta \lceil q / q' \rceil$-close to that of a true execution of $\calA$ on $p$, where we have
	\begin{align*}
	\delta \left\lceil \frac{q}{q'} \right\rceil \leq \frac{1}{\eps^2 q \cdot 100^{1/\zeta}} \left( \frac{q \cdot 2\eps^2}{\zeta \log(1/\delta)} + 1 \right) \leq \frac{3}{100}
	\end{align*}
	
	In order to upper bound the communication complexity, notice that each round of $\lceil q / q' \rceil$ sends $O(\log(1/\delta))$ bits and succeeds with probability at least $1/2$; which means that the expected communication complexity of a round is $O(\log(1/\delta))$. Hence, the expected communication complexity of the whole protocol is therefore 
	\[ O\left(\left\lceil \frac{q}{q'} \right\rceil \log(1/\delta)\right) \leq O\left( \frac{q\log(1/\delta)}{q'}  + \log(1/\delta) \right) = O\left(q\eps^2 + \log(q\eps^2)\right) \leq O(q\eps^2).\] 
	In order to bound the worst-case communication complexity, we use Markov's inequality. Specifically, by losing another constant factor, we may assume the protocol sends $O(q\eps^2)$ bits except with probability at most $1/100$; in this case, Alice sends an arbitrary bits. Then, the distribution over the samples that Bob may reconstruct is $(3/100 + 1/100)$-close to that of a true execution of $\calA$ on $p$.
\end{proofof}

\newcommand{\TestingJuntas}{\texttt{TestingJuntas}}
\newcommand{\MeanTester}{\texttt{MeanTester}}

\section{Testing Algorithm}

\begin{figure}[t!]
	\begin{framed}
		\noindent Subroutine $\TestingJuntas \hspace{0.04cm} (p,k,\eps)$
		
		\begin{flushleft}
			\noindent {\bf Input:} Subcube conditioning access to a distribution $p$ supported on $\{-1,1\}^n$, an integer $k\in \N$ and a proximity parameter $\eps \in (0, 1/4]$.
			
			\noindent {\bf Output:} Either \texttt{accept} or \texttt{reject}.
			
			\begin{enumerate}
				\item Let $c$ be the universal constant in the main structural lemma. We let
				\begin{equation}\label{eq:defeqs}
				\eps'=\frac{\eps}{\lceil \log_2 2n\rceil \cdot \log^c(n/\eps)},\quad 
				r=\big\lceil \log ( 2\sqrt{n}/\eps')\big\rceil \quad \text{and}\quad \eps^*=\frac{\eps'}{1600 r}.
				\end{equation}
				\item Execute $\FindRelevantVariables\hspace{0.04cm}(p,k,\eps^*)$ 
				and let $J$ be the set it returns.
				\item If $|J|>k$, \texttt{reject}. 
				\item  
				For each $j \in [\lceil \log_2 2n \rceil]$ and $\ell\in [
				r]$ with $r=\lceil \log ( 2\sqrt{n}/\eps')\rceil$:
				\begin{enumerate}
					\item[] \hspace{-0.4cm}Repeat the following $L\cdot R$ times, where 
					$$L=\frac{4r\sqrt{n}}{2^\ell \eps'} \quad\text{and}\quad R= O\left(\log \left(\frac{n}{\eps'}\right)\right)$$ 
					\begin{enumerate}
						\item[(A)] Sample  $\brho \sim \calD_{\ol{J}}(p)$ and $\bnu \sim \calD_{\sigma^j}(p_{|\brho})$, execute $\MeanTester\hspace{0.04cm}((p_{|\brho})_{|\bnu},k,2^{-\ell})$
						for\\ $R$ times and take the majority of answers.\vspace{0.1cm}

					\end{enumerate}
					
					\item[] \hspace{-0.4cm}\texttt{Reject} if for at least $R/2$ rounds of (A), the majority
					of answers is ``\texttt{Not a Junta}''.
				\end{enumerate}
				\item \texttt{Accept} if this line is reached.
			\end{enumerate}
			
		\end{flushleft}\vskip -0.14in
	\end{framed}\vspace{-0.2cm}
	\caption{The $\TestingJuntas$ algorithm for testing junta distributions.}\label{fig:testing}
\end{figure}

We use $\FindRelevantVariables$ and $\MeanTester$ to give an algorithm for 
testing $k$-junta distributions.
The algorithm, $\TestingJuntas$, is described in Figure \ref{fig:testing}; we prove the following theorem:

\thmtestingalg*

\begin{proofof}{Theorem \ref{thm:testing}}
	We start with the soundness case to show that $\TestingJuntas$ rejects with probability at least $2/3$
	when $p$ is far from $k$-juntas. 
	Assume without loss of generality that the set $J$ returned by $\FindRelevantVariables$
	has size at most $k$; otherwise $\TestingJuntas$ rejects.
	
	Given $|J|\le k$ and $p$ is $\eps$-far from $k$-junta distributions, the main structural lemma implies that
	\begin{equation*} \sum_{j=1}^{ \lceil \log_2 2n \rceil} \Ex_{\brho \sim \calD_{\ol J}(p)} \left[ \Ex_{\bnu \sim \calD_{\sigma^j}(p_{|\brho})}\Big[\big\|\mu((p_{|\brho})_{|\bnu})\big\|_2 \Big] \right] \geq \dfrac{\eps}
	{ \log^{c} (n/\eps) }.
	\end{equation*}
	As a result, there exists a $j\in \lceil \log_2 2n\rceil$ (using the choice of $\eps'$ in (\ref{eq:defeqs})) such that  
	$$
	\Ex_{\brho \sim \calD_{\ol J}(p)} \left[ \Ex_{\bnu \sim \calD_{\sigma^j}(p_{|\brho})}\Big[\big\|\mu((p_{|\brho})_{|\bnu})\big\|_2 \Big] \right] \geq \eps'.$$
	Fix such a $j$ and we apply the following 
	claim (which is elementry and we delay its proof):
	
	\begin{claim}\label{simplesimple2}
		Let $\bX$ be a random variable that takes values between $0$ and $1$.
		If $\E[\bX]\ge \delta$ for some $\delta\in (0,1)$, then there exists an $\ell\in [\lceil \log (2/\delta)\rceil]$ such that
		$$
		\Pr\big[\bX\ge 2^{-\ell}\big]\ge \frac{2^{\ell}\delta}{4\lceil \log (2/\delta)\rceil}
		$$
	\end{claim}  
	
	Scaling down by $\sqrt{n}$ and applying Claim \ref{simplesimple2}, 
	there is an $\ell\in [r]$ with $r=\lceil \log (2\sqrt{n}/\eps')\rceil$ such that
	\begin{equation}\label{hehe100}
	\Pr_{\brho,\bnu}\left[\big\|\mu((p_{|\brho})_{|\bnu})\big\|_2\ge \sqrt{n}\big/2^\ell
	\right]\ge \frac{2^{\ell}\eps'}{4r\sqrt{n}}.
	\end{equation}
	It follows from a Chernoff bound that, with probability at least $1-o_n(1)$,
	the number of rounds of (A) in which
	$\brho,\bnu$ satisfy (\ref{hehe100}) is at least $2R/3$ (since the expectation is at least $R$). 
	It follows from the promise we get from $\MeanTester$ (i.e., each run returns ``\texttt{Not a Junta}''
	with probability at least $2/3$ when the event in (\ref{hehe100}) holds) that 
	with probability at least $1-o_n(1)$, 
	the majority of answers returned by $\MeanTester$ is ``\texttt{Not a Junta}''
	in each of these $2R/3$ rounds of (A).
	So overall the algorithm rejects with probability at least $1-o_n(1)$.
	This finishes  the soundness case.
	
	Next we work on the completeness case to show that $\TestingJuntas$ accepts with probability at least $2/3$
	when $p$ is a $k$-junta distribution. Suppose $p$ is a $k$-junta distribution, and let $K \subset [n]$ be the set of at most $k$ relevant variables (which is unknown to the algorithm).
	First it follows from Lemma~\ref{thm:kjunta-smallnorm1} that with probability at least $7/9$, the output $J$ of $\FindRelevantVariables$ satisfies both conditions of Lemma~\ref{thm:kjunta-smallnorm1}. So let $|J| \leq k$, and for every $j \in [\lceil \log_2(2k)\rceil ]$, 
	\begin{align}
	\Ex_{\brho \sim \calD_{\ol{J}}(p)}\left[ \Ex_{\bnu \sim \calD_{\sigma^j}(p_{|\brho})}\left[ \|\mu((p_{|\brho})_{|\bnu})\|_2 \right]\right] \leq \eps^*. \label{eq:up-to-k}
	\end{align}
	We will now use this fact, as well as the following simple claim (whose proof we defer), to derive the bound (\ref{eq:up-to-k}) for all $j \in [\lceil \log_2(2n)\rceil ]$, and not just up to $\lceil \log_2(2k)\rceil$.
	
	\begin{claim}\label{thm:kjunta-smallnorm}
		Fix $m \in \N$ and let $h$ be any distribution over $\{-1,1\}^m$. For any $0 \leq \sigma_2 \leq \sigma_1 \leq 1/m$, we have
		\begin{align*}
		\Ex_{\bnu \sim \calD_{\sigma_2}(h)}\left[ \|\mu(h_{|\bnu})\|_2 \right] \leq \Ex_{\bnu \sim \calD_{\sigma_1}(h)}\left[ \|\mu(h_{|\bnu})\|_2 \right]
		\end{align*}
	\end{claim}
	
	For every $\rho \in \supp(\calD_{\ol{J}}(p))$, let $h^{(\rho)}$ be the distribution over $\{-1,1\}^{K\setminus J}$ given by $(p_{|\rho})_{K \setminus J}$. Since $p$ is a junta over variables in $K$, for every $\rho \in \supp(\calD_{\ol{J}}(p))$, the distribution of $p_{|\rho}$ over variables outside of $K$ is always uniform, irrespective of the restriction $\rho$. Hence, for any $\sigma' \in (0,1)$, the non-zero coordinates of the mean vector $\mu((p_{|\rho})_{\bnu})$ for $\bnu \sim \calD_{\sigma'}(p_{|\rho})$ are always supported on those coordinates in $K$. Hence, for every $\sigma' \in (0, 1)$, 
	\[ \Ex_{\bnu \sim \calD_{\sigma'}(p_{|\rho})}\left[ \|\mu((p_{|\rho})_{\bnu})\|_2 \right] = \Ex_{\bnu \sim \calD_{\sigma'}(h^{(\rho)})}\left[ \|\mu(h^{(\rho)}_{|\bnu})\|_2\right]. \]
	We let $j^* = \lceil \log_2(2k)\rceil$ and note that $\sigma^{j^*} \leq 1/k$. By Claim~\ref{thm:kjunta-smallnorm}, we have that for $j' \in [\lceil \log_2(2n)\rceil]$ with $j' \geq j^*$,
	\begin{align*}
	\Ex_{\bnu \sim \calD_{\sigma^{j'}}(p_{|\rho})}\left[\|\mu((p_{|\rho})_{|\bnu})\|_2 \right] = \Ex_{\bnu \sim \calD_{\sigma^{j'}}(h^{(\rho)})}\left[ \|\mu(h_{|\bnu}^{(\rho)})\|_2\right] &\leq \Ex_{\bnu \sim \calD_{\sigma^{j^*}}}\left[\|\mu(h_{|\bnu}^{(\rho)})\|_2 \right] = \Ex_{\bnu \sim \calD_{\sigma^{j^*}}}\left[ \|\mu((p_{|\rho})_{\bnu})\|_2\right].
	\end{align*}
	Averaging over $\brho \sim \calD_{\ol{J}}(p)$ implies that for all $j \in [\lceil \log_2(2n)\rceil ]$, 
	\begin{align*}
	\Ex_{\brho \sim \calD_{\ol{J}}(p)}\left[ \Ex_{\bnu \sim \calD_{\sigma^j}(p_{|\brho})}\left[\|\mu((p_{|\brho})_{|\bnu})\|_2 \right]\right] \leq \eps^*,
	\end{align*}
	which in turn, implies that for all $j \in [\lceil \log_2(2n)\rceil]$, and all $\ell \in [r]$,
	\begin{align}
	\Prx_{\brho, \bnu}\left[\|\mu((p_{|\brho})_{|\bnu})\|_2 \geq \sqrt{n} / (100 \cdot 2^{\ell}) \right] \leq \frac{2^{\ell} \eps^* \cdot 100}{\sqrt{n}} \leq \frac{2^{\ell} \eps'}{16 r \sqrt{n}} \label{hehe101}
	\end{align}

	using our choice of $\eps^*$ in (\ref{eq:defeqs}). 
	Fix $j$ and $\ell$.  
	It follows from a Chernoff bound that
	with probability at least
	$1-e^{-\Omega(R)}$,
	the number of rounds of (A) that satisfy the event in (\ref{hehe101}) is at most $R/2$ (because the expectation is 
	at most $R/4$).
	The latter implies that the number of rounds of (A) that violate the event in (\ref{hehe101})
	is at least $LR-R/2$.
	For each of these $LR-R/2$ rounds of (A), the majority of runs of $\MeanTester$ in (A) 
	returns ``\texttt{Is a Junta}'' with probability at least
	$ 
	1-e^{-\Omega(R)}
	$ by a Chernoff bound. 
	By a union bound we have that all these $LR-R/2$ rounds have majority being ``\texttt{Is a Junta}'' 
	with probability at least
	$ 
	1-(LR-R/2)\cdot  e^{-\Omega(R)} .
	$ 
	It follows that the main loop with $j$ and $\ell$ rejects with probability at most
	$$
	1-e^{-\Omega(R)}-(LR-R/2)\cdot e^{-\Omega(R)}\le 1-LR\cdot e^{-\Omega(R)}.
	$$
	Using a union bound over all main loops, 
	the algorithm rejects with probability at most
	$$
	\frac{2}{9}+\lceil \log 2n\rceil \cdot r\cdot LR\cdot e^{-\Omega(R)}<\frac{1}{3}.
	$$
	
	Finally we bound the number of queries.
	Notice that both $\eps'$ and $\eps^*$ are $\eps/\polylog(n/\eps)$.
	Hence the number of queries made by the call to $\FindRelevantVariables^*$ is 
	$\tilde{O}(k/\eps^2)\cdot \polylog(n)$.
	On the other hand, the number of queries made by calls to $\MeanTester$ is (using $r=\lceil \log_2 (2\sqrt{n}/\eps')\rceil$)
	\begin{align*}
	&\lceil \log_2 2n\rceil \cdot 
	\sum_{\ell=1}^r \frac{4r\sqrt{n}}{2^\ell\eps'}\cdot O\left(\log^2\left(\frac{n}{\eps'}\right)\right)
	\cdot (k+\sqrt{n})\cdot \max\left\{\frac{2^{2\ell}}{n},\frac{2^\ell}{\sqrt{n}}\right\}
	\\&\hspace{0.8cm}=(k+\sqrt{n})\cdot \polylog\left(\frac{n}{\eps}\right)\cdot \sum_{\ell=1}^r
	\frac{\sqrt{n}}{2^\ell \eps}  \cdot \max\left\{\frac{2^{2\ell}}{n},\frac{2^\ell}{\sqrt{n}}\right\}
	=\tilde{O}\left(\frac{k+\sqrt{n}}{\eps^2}\right).
	\end{align*}
	The upper bound on the running time can simply be verified from Figure~\ref{fig:testing} and Theorem~\ref{thm:MeanTesting++}. This finishes the proof of the theorem.
\end{proofof}

\begin{proofof}{Claim \ref{simplesimple2}}
	Let $r=\lceil \log (2/\delta)\rceil$, and assume for contradiction that
	the claim is not true for any $\ell\in [r]$.
	Then we have
	$$
	\delta\le \E[\bX]< \sum_{\ell=1}^r \frac{2^{\ell}\delta}{4r}\cdot \frac{2}{2^{\ell}}+1\cdot \frac{1}{2^r} =\delta,
	$$
	a contradiction.
\end{proofof}

\begin{proofof}{Claim~\ref{thm:kjunta-smallnorm}}
	We simply note that for any restriction $\nu \in \{-1,1,*\}^m$ with $\stars(\nu) = S$, 
	\begin{align*}
	\Prx_{\bnu \sim \calD_{\sigma_1}(h)}\left[ \bnu = \nu \right] &= \Prx_{\bS \sim \calS_{\sigma_1}}\left[ \bS = S \right] \cdot \Prx_{\bx \sim h_{\ol{S}}}\left[ \bx = \nu_{\ol{S}}\right] \geq \Prx_{\bS \sim \calS_{\sigma_2}}\left[ \bS = S \right] \cdot \Prx_{\bx \sim h_{\ol{S}}}\left[ \bx = \nu_{\ol{S}}\right] = \Prx_{\bnu \sim \calD_{\sigma_2}(h)}\left[ \bnu = \nu\right],
	\end{align*}
	where we used the fact that 
	\begin{align*}
	\frac{d}{d\sigma}\left[ \Prx_{\bS \sim \calS_{\sigma}}\left[ \bS = S \right]\right] &= \sigma^{|S|-1}\left(1 - \sigma \right)^{m-|S|-1} \left( |S| - \sigma m\right) > 0
	\end{align*}
	whenever $0 \leq \sigma \leq 1/m$.
\end{proofof}

\section{Lower Bound for Testing}

In this section, we prove the following theorem showing a lower bound for testing whether a product distribution is an $k$-junta distribution with $k=n/2$. We first state the theorem and proceed to show it implies Theorem~\ref{thm:lb}.

\begin{theorem}\label{thm:n/2-junta}
	There exist two absolute constants $\eps_1>0$ and $C_1\in \N$ such that for all $0<\eps\le \eps_1$ and $n\ge C_1^2$, any algorithm which receives samples from an unknown product distribution $p$ supported on $\{-1,1\}^n$ and distinguishes with probability at least $2/3$ between the case $p$ is an $(n/2)$-junta distribution and the case $p$ is $\eps $-far from being an $(n/2)$-junta distribution must observe at least $\tilde{\Omega}(n)/\eps^2$ many samples from $p$.
\end{theorem}

\begin{proofof}{Theorem~\ref{thm:lb} assuming Theorem~\ref{thm:n/2-junta}}
	We first inspect the proof of Theorem~4.8 from \cite{CDKS17}, which presents a lower bound on the sample complexity of testing whether an unknown product distribution is uniform or far from uniform. Specifically, they show that there are two constants 
	$\eps_2>0$ and $C_2\in \N$ such that for any $\eps\in (0,\eps_2]$ and $n\ge C_2$,
	there are two distributions $\calY$ and $\calN$, supported on product distributions over $\{-1,1\}^n$ such that no algorithm can determine whether a draw $\bp$ belongs to $\calY$ or $\calN$ with probability greater than $2/3$ without observing $\Omega(\sqrt{n} / \eps^2)$ samples from $\bp$.
	Moreover, the distribution $\calY$ always outputs $\calU_n$ and 
	the distribution $\calN$ always outputs a distribution $\bp$ that is $\eps$-far from 
	being a $(n/2)$-junta distribution.
	We are done if $k\le \sqrt{n}$ so we are left with the case when $k\ge \sqrt{n}$.
	In the rest of the proof we prove a lower bound of $\tilde{\Omega}(k)/\eps^2$ 
	with a reduction to Theorem \ref{thm:n/2-junta}.
	
	We now prove Theorem \ref{thm:lb}  
	by setting the two constants $\eps_0=\min(\eps_1,\eps_2)$ and $C_0=\max(C_1^2,C_2)$.
	Let $\eps\in (0,\eps_0]$, $n\ge C_0$ and $0\le k\le n/2$.
	Since $\calU_n$ is trivially a $k$-junta distribution and $k\le n/2$,
	the properties of $\calY$ and $\calN$ from \cite{CDKS17} imply a lower bound
	of $\Omega(\sqrt{n}/\eps^2)$ for distinguishing between the case $p$ is a $k$-junta distribution
	and the case $p$ is $\eps$-far from a $k$-junta distribution. 
	
	Note that $k\ge \sqrt{n}\ge C_1$.
	Consider an unknown product distribution $g$ over $\{-1,1\}^{2k}$ and the task
	of distinguishing the case $g$ is a $k$-junta distribution and the case $g$ is $\eps$-far
	from a $k$-junta distribution.
	By Theorem \ref{thm:n/2-junta}, any algorithm for this task must observe $\tilde{\Omega}(k)/\eps^2$ samples from $g$.
	On the other hand, let $g'$ be the distribution supported on $\{-1,1\}^n$ 
	defined using $g$ as follows: To draw $\bx\sim g'$ we first draw a sample $\by\sim g$ and 
	set $\by$ to be the first $2k$ bits of $\bx$; the last $n-2k$ bits of $\bx$ are drawn
	independently and uniformly at random.
	Notice that if $g$ is a $k$-junta, then $g'$ is a $k$-junta, and if $g$ is $\eps$-far from a $k$-junta
	, then $g'$ is $\eps$-far from a $k$-junta. 
	Given that sample access to $g'$ can be simulated using sample access to $g$,
	the task of distinguishing between the case $g'$ is a $k$-junta and 
	the case $g'$ is $\eps$-far from $k$-junta is at least as hard as the task for $g$.
	From this reduction we get a sample complexity lower bound of $\tilde{\Omega}(k)/\eps^2$.
\end{proofof}

The proof of Theorem~\ref{thm:n/2-junta} follows from the following lemma by simply noticing that any algorithm which receives $s$ independent samples from an unknown product distribution $p$ over $\{-1,1\}^n$ can be simulated by an algorithm which receives a sample from the product distribution $\Bin(s, p_1) \times \dots \times \Bin(s,p_n)$. 

\newcommand{\Ryes}{\calR_{\text{yes}}}
\newcommand{\Rno}{\calR_{\text{no}}}

\begin{lemma}\label{lem:moment-matching}
	There 
	exists an absolute constant $\eps_0>0$ 
	such that for all
	$\eps\in (0,\eps_0]$ and $n \in \N$, 
	there exist two distribution $\Dyes$ and $\Dno$ supported on product distributions over $\{-1,1\}^n$ satisfying
	\begin{equation}\label{eq:in-k-junta-and-far}
	\Prx_{\bp \sim \Dyes}\big[ \bp \in \Junta{n/2} \big] \geq 1 - o_n(1) \quad\text{and}\quad \Prx_{\bp \sim \Dno}\big[ \dtv(\bp, \Junta{n/2} ) \geq \eps \big] \geq 1 - o_n(1). 
	\end{equation}
	Moreover, letting $s =\lceil n / (\eps^2 \log^{12} n)\rceil$, the two distributions $\Ryes = \calR(s, \Dyes)$ and $\Rno = \calR(s,\Dno)$ supported on $\N^{n}$
	satisfy $\dtv\left( \Ryes, \Rno \right) = o_n(1)$, where $\calR(s,\calD)$ is specified by letting
	\begin{equation}\label{defeq}
	\Prx_{\boldr \sim \calR(s,\calD)}\left[ \boldr = r \right] = \Ex_{\bp \sim \calD}\left[ \prod_{i =1}^n \Prx_{\bell \sim \Bin(s, \bp_i)}\left[ \bell = r_i \right]\right],\quad\text{for every $r \in \N^n$.} 
	\end{equation}
\end{lemma}

The proof of Lemma \ref{lem:moment-matching} constitutes the next two subsections.
We give the construction of $\Dyes$ and $\Dno$ and prove (\ref{eq:in-k-junta-and-far}) in Section \ref{sec:construction},
and bound the distance between $\Ryes$ and $\Rno$ in Section \ref{sec:BoundDistance}.

\subsection{Construction of $\Dyes$ and $\Dno$}\label{sec:construction}

Let $p$ be a product distribution over $\{-1,1\}^n$. 
We prove the following lemma that lowerbounds $\dtv(p,\calU_n)$ using $\|\mu(p)\|_2$:

\begin{lemma}\label{dtvlowerbound}
	There is two constants $c_1^*, c_2^*>0$ such that any product distribution $p$ over $\{-1,1\}^n$ satisfies
	$$
	\dtv(p,\calU_n)\ge \left(\frac{1}{8}- \frac{c_1^* \|\mu(p)\|_\infty}{\|\mu(p)\|_2}\right)
	\cdot \min\left(c_2^*,\frac{\|\mu(p)\|_2}{4}\right).
	$$
\end{lemma}

We delay the proof of Lemma \ref{dtvlowerbound} to Section \ref{sec:dtvlowerbound}.
We fix the constant $\eps_0 \in \R_{\geq 0}$ in Lemma \ref{lem:moment-matching} to be 
\begin{align}
\eps_0 = \frac{c^*_2}{9}. \label{eq:setting-eps-0}
\end{align}

For $n \in \N$, let $\ell = \lceil \log n / \log \log n\rceil$. Given any vector $\alpha \in \R^{\ell}$ 
we let $A(\alpha)$ be the Vandermonde matrix defined with respect to $\alpha$, and $e_1 \in \R^{\ell}$ be the first basis vector:
\[ A(\alpha) = \left[ \begin{array}{ccccc} \alpha_1^0 & \alpha_2^0 & \alpha_3^0 & \dots & \alpha_{\ell}^0 \\[0.8ex]
\alpha_1^1 & \alpha_2^1 & \alpha_3^1 & \dots & \alpha_{\ell}^1 \\[0.8ex]
\alpha_1^2 & \alpha_2^2 & \alpha_3^2 & \dots & \alpha_{\ell}^2 \\[0.2ex]
\vdots & \vdots & \vdots & \ddots & \vdots  \\[0.8ex]
\alpha_1^{\ell-1} & \alpha_2^{\ell-1} & \alpha_3^{\ell-1} & \dots & \alpha_{\ell}^{\ell-1} \end{array} \right] \qquad \text{and}\qquad e_1 = \left[ \begin{array}{c} 1 \\[0.4ex] 0 \\[0.4ex] 0 \\[-0.3ex] \vdots \\[0.4ex] 0 \end{array}\right]. \]
Recall the following closed form for the determinant of a Vandermonde matrix $A(\alpha)$:
\[ \det\big(A(\alpha)\big) = \prod_{\substack{i,j \in [\ell] \\ i < j}} (\alpha_j - \alpha_i),\]
so that $\det(A(\alpha)) \neq 0$ whenever coordinates of $\alpha$ are distinct. For the rest of the section, consider the vector $\alpha \in \R^{\ell}$ given by letting
\begin{align}
\alpha_j &= j^3 \qquad\forall j\in[\ell], \label{eq:alpha-def}
\end{align}
and let $z \in \R^{\ell}$ be the unique solution to the system of $\ell$ linear equations where $A(\alpha) z = e_1$. Let 
\[ \calW = \left\{ j \in [\ell] :  z_j \geq 0 \right\} \qquad\text{and}\qquad \calV = [\ell] \setminus \calW. \]
We will need the following technical claim about $z$; we delay its proof to Subsection~\ref{sec:ell-1-z}. 
\begin{claim}\label{cl:z-bound}
	There is an absolute constant $C^* > 0$ such that for any $\ell \in \N$, the solution $z \in \R^{\ell}$ to the Vandermonde system $A(\alpha) z = e_1$ with $\alpha$ as in (\ref{eq:alpha-def}) satisfies $\|z\|_1 \leq C^*$.
\end{claim}


We now  describe $\Dno$ and $\Dyes$ using $\alpha$, $\calW$ and $\calV$ given above.
Let $\tau \in \R_{\geq 0}$ be set as 
\begin{align}
\tau &= \min\left\{ 36\sqrt{C^*}\cdot \eps, \frac{\sqrt{n}}{2\ell^3} \right\}, \label{eq:tau-setting}
\end{align}
and notice that for large $n$, $\tau = 36\sqrt{C^*} \eps = \Theta(\eps)$.
First we let $\bp \sim \Dno$ be the product distribution supported on $\{-1,1\}^n$ given by letting for each $i \in [n]$, be independently set to 
\begin{align} 
\Prx_{\bx \sim \bp}\left[ \bx_i = 1 \right] = \frac{1}{2} + \frac{\bgamma_i \cdot \tau}{\sqrt{n}} \qquad\text{such that } \bgamma_i = \left\{ \begin{array}{ll} 0 & \text{w.p.\  } 1 - \dfrac{\sum_{j \in \calW} z_j}{\|z\|_1} \\[2ex]
j^3 & \text{w.p.\  } \dfrac{z_j}{\|z\|_1} \text{ for $j \in \calW$.} \end{array} \right. . \label{eq:sample-dno}
\end{align}
Notice that probabilities above are smaller than $1$ since $\bgamma_i \leq \ell^3$, for $\ell = \lceil \log n / \log \log n\rceil$ and the setting of $\tau$. 
On the other hand, we let $\bq \sim \Dyes$ be the product distribution supported on $\{-1,1\}^n$ given by letting for each $i \in [n]$, be independently set to
\begin{align}
\Prx_{\bx \sim \bq}\left[ \bx_i = 1\right] = \frac{1}{2} + \frac{\bdelta_i \cdot \tau}{\sqrt{n}} \qquad\text{such that }\bdelta_i = \left\{ \begin{array}{ll} 0 & \text{w.p.\ } 1 - \dfrac{\sum_{j \in \calV} -(z_j)}{\|z\|_1} \\[2ex]
j^3 & \text{w.p.\ } \dfrac{-z_j}{\|z\|_1}  \text{ for $j \in \calV$.}\end{array} \right. . \label{eq:sample-dyes}
\end{align}

Again, we note that the probabilities are at most $1$ since $\bdelta_i \leq \ell^3$ as well. We record a claim that follows directly from the definition of $z$, $\calW$ and $\calV$:
\begin{claim}\label{cl:matching-moments}
	For all $k = 1, \dots, \ell-1$, we have
	\begin{align} 
	\Ex_{\bdelta_i}\big[\bdelta_i^k\big] &= \Ex_{\bgamma_i}\big[\bgamma_i^k\big]. \label{eq:expression}
	\end{align}
\end{claim}
\begin{proof}
	The proof follows from the fact that 
	\begin{align*}
	\Ex_{\bgamma}\big[ \bgamma_i^k \big] -\Ex_{\bdelta_i}\big[ \bdelta_i^k \big] = \frac{1}{\|z\|_1} \sum_{j=1}^{\ell} \alpha_j^k z_j = \frac{1}{\|z\|_1} (A(\alpha) z)_{k+1} = 0,
	\end{align*}
	since $A(\alpha)z = e_1$.
\end{proof}

We show in the next two claims that (\ref{eq:in-k-junta-and-far}) holds when $n$ is sufficiently large.

\begin{claim}\label{lem:yes-dist}
	We have 
	$\bp \in \Junta{n/2}$ with probability at least $1 - o_n(1)$ over the draw of $\bp \sim \Dyes$.
\end{claim}
\begin{proof}
	Let $\bp \sim \Dyes$, and let $\bA \subseteq [n]$ be the set of coordinates $i \in [n]$ with $\bdelta_i \neq 0$. We will show that, when $n$ is sufficiently large, $|\bA| \leq n/2$ with probability $1- o_n(1)$, which implies that $\bp\sim \Dyes$ is an $(n/2)$-junta for $\calU_n$ with probability at least $1-o_n(1)$. 
	
	To see this is the case, we notice that each $\bdelta_i$ is 0 with probability 
	\begin{align*}
	1 - \frac{\sum_{j \in \calV} -z_j}{\|z\|_1} &= \frac{1}{2}\left(1 + \dfrac{\sum_{j \in \calW} z_j + \sum_{j \in \calV} z_j}{\|z\|_1}\right) = \frac{1}{2} + \frac{1}{2\|z\|_1} \geq \frac{1}{2} + \frac{1}{2C^*},
	\end{align*}
	where we used the fact that $z$ was the solution to $(A(\alpha)z)_1 = 1$ to deduce that $\sum_{j } z_j = 1$. Hence,~for large $n$, we apply a Chernoff bound to deduce that $|\bA| \leq n/2$ except with probability $o_n(1)$.
\end{proof}

\begin{claim}\label{lem:no-dist}
	We have $\bp$ is $\eps$-far from $\Junta{n/2}$ with probability at least $1-o_n(1)$ over the draw of $\bp \sim \Dno$.
\end{claim}
\begin{proof}
	By a similar computation, as the proof of Claim~\ref{lem:yes-dist}, if we let $\bA$ be the subset of coordinates $i \in [n]$ with $\bgamma_i = 0$ in $\bp\sim\Dno$, we have $$|\bA| \leq n \left(\frac{1}{2} - \frac{1}{4C^*}\right)$$ except with probability $o_n(1)$. Consider a fixed distribution $p$ in the support of $\Dno$ where the above event occurs, i.e., the set $A \subset [n]$ of coordinates with zero $\gamma_i$ (specifying the marginal distributions of $p$ as in (\ref{eq:sample-dno})) is smaller than $n/2 - n/(4C^*)$. Let $q$ be any
	$(n/2)$-junta distribution and let $S$ be the influential variables of $q$'s p.d.f with $|S| \leq n/2$.
	We have that, for each $i\in \ol{A}\cap \overline{S}$, 
	\[ |\mu(p)_i| \geq 2\tau \gamma_i /\sqrt{n} \geq 2\tau / \sqrt{n}. \] 
	Let $T$ be $\ol{A}\cap \overline{S}$ with
	$$t \eqdef |T|=|\ol{A}\cap \overline{S}|\ge n\left(\frac{1}{2}+\frac{1}{4C^*}\right) -
	\frac{n}{2} \ge \frac{n}{4C^*}.
	$$
	
	Consider the distributions $p_{T}$ and $q_{T}$ given by taking a sample and projecting onto the coordinates in $T$. Since $T \subset \ol{S}$, and the p.d.f of $q$ is constant for any setting of variables in $S$, the distribution $q_T$ is the uniform distribution over $t$ bits. 
	We note
	\begin{align}
	\dtv(p_{T}, \calU_t) &= \frac{1}{2} \sum_{x \in \{-1,1\}^T} |p_T(x) - q_T(x)| = \frac{1}{2} \sum_{x \in \{-1,1\}^T}\left|\sum_{y \in \{-1,1\}^{\ol{T}}} p(x,y) - q(x,y) \right| \nonumber \\
	&\leq \frac{1}{2} \sum_{z \in \{-1,1\}^n} |p(z) - q(z)| = \dtv(p, q), \label{eq:projection-to-actual}
	\end{align}
	where $p(x,y) = p(z)$ with $z_i = x_i$ for $i \in T$ and $z_i = y_i$ for $i \notin T$, and $q(x,y)$ is defined analogously.
	We now apply 
	Lemma~\ref{dtvlowerbound} to deduce a lower bound on $\dtv(p_T, \calU_t)$, and by (\ref{eq:projection-to-actual}) lower bound $\dtv(p,q)$. Since $p$ is a product distribution, $\mu(p)_i = \mu(p_T)_i$ for all $i \in T$, and we have
	\begin{align}
	\|\mu(p_{T})\|_{\infty} &\leq \frac{2\tau \ell^3}{\sqrt{n}} \qquad\text{and}\qquad \|\mu(p_T)\|_2 \geq \sqrt{t} \cdot \frac{2\tau}{\sqrt{n}} = \frac{\tau}{\sqrt{C^*}}. \label{eq:ell-infty-ell-2}
	\end{align}
	Applying Lemma~\ref{dtvlowerbound}, we have 
	\begin{align*}
	\dtv(p_{T}, \calU_t) &\geq \left(\frac{1}{8} - o_n(1) \right) \cdot \min\left(c_2^*, \frac{\tau}{ 4\sqrt{C^*}} \right) \geq \min\left(\frac{c_2^*}{9}, \frac{\tau}{36 \sqrt{C^*}} \right),
	\end{align*}
	once $n$ is a large enough constant.
	Finally, by the setting of $\eps_0$ in (\ref{eq:setting-eps-0}), and $\tau$ in (\ref{eq:tau-setting}), $\dtv(p_{T}, \calU_t) \geq \min(\eps_0, \eps) = \eps$ for large enough $n$. Since the distribution $q$ was an arbitrary $(n/2)$-junta distribution, this concludes the proof.
\end{proof}

\subsection{Statistical Distance Between $\Ryes$ and $\Rno$}\label{sec:BoundDistance}

Let $s=\lceil  n / (\eps^2  \log^{12} n)\rceil$.
We show that distributions $\Ryes = \calR(s,\Dyes)$ and $\Rno = \calR(s, \Dno)$ as defined 
in (\ref{defeq}) using $\Dyes$ and $\Dno$ satisfy 
\begin{equation}\label{finaleq}
\dtv(\Ryes, \Rno) \leq o_n(1).
\end{equation}

Recall that $\Ryes$ is the distribution supported on $\{ 0, \dots, s\}^n$ given by first sampling 
$\bdelta_1,\ldots,\bdelta_n$ independently according to (\ref{eq:sample-dyes})
and then sampling from the product distribution
\begin{align}
\boldr &\sim \prod_{i=1}^n \Bin\left(s, \bq_i\right),\qquad\text{where}\quad
\text{$\bq_i \eqdef \Prx_{\bx \sim \bq}\left[ \bx_i = 1\right]=\frac{1}{2} + \frac{\bdelta_i \cdot \tau}{\sqrt{n}}$}. \label{eq:sample-r}
\end{align}
Notice that we always have 
\[ \frac{1}{2} \leq \bq_i \leq \frac{1}{2} + \frac{\tau \ell^3}{\sqrt{n}} \leq \frac{1}{2}+O\left(\frac{\eps \log^3 n}{\sqrt{n}} \right) \]
once $n$ is a large enough constant. 

Similarly, $\Rno$ is the distribution supported on $\{0, \dots, s\}^n$ given by first sampling $\bgamma_1,\ldots,\bgamma_n$ according to (\ref{eq:sample-dno}), and then sampling from the product distribution
\begin{align*}
\boldr \sim \prod_{i=1}^n \Bin\left(s, \bp_i\right),\qquad\text{where}\qquad \bp_i \eqdef\Prx_{\bx \sim \bp}\left[ \bx_i = 1\right] = \frac{1}{2} + \frac{\bgamma_i \cdot \tau}{\sqrt{n}},
\end{align*} 
and similarly, we have $1/2 \leq \bp_i \leq 1/2 + O(\eps \log^3 n / \sqrt{n})$. 
In particular, if we denote the set $B \subset \{0, \dots, s\}^n$ given by
\[ B = \left\{ r = (r_1, \dots, r_n) \in \{0,\dots, s\}^n : \exists\hspace{0.03cm} j \in [n], \left|r_j - \frac{s}{2}\right| \geq \sqrt{s} \log^2 n \right\}. \] 
It follows from our choice of $s$, that for every $i \in [n]$ and any fixed setting of $\bp_1, \dots, \bp_n$ and $\bq_1, \dots, \bq_n$,
\begin{align*}
\frac{s}{2} \leq \Ex_{\boldr_i \sim \Bin(s, \bp_i)}\left[ \boldr_i \right] , \Ex_{\boldr_i \sim \Bin(s, \bq_i)}\left[ \boldr_i \right] \leq \frac{s}{2} + O\left(\frac{s \eps \log^3 n}{\sqrt{n}}\right) = \frac{s}{2} + O\left( \sqrt{s}\right),
\end{align*}
so that via a Chernoff bound and a union bound,
\[ \Prx_{\boldr \sim \Ryes}[\boldr \in B],\ \Prx_{\boldr \sim \Rno}[\boldr \in B] = o_n(1).  \]
Therefore, in order to show $\dtv(\Ryes, \Rno) = o_n(1)$, it suffices to show that for every $r \notin B$,\vspace{0.15cm}
\begin{align}
\dfrac{\Prx_{\boldr \sim \Ryes}\left[ \boldr = r\right]}{\Prx_{\boldr \sim \Rno}\left[ \boldr = r\right]} &= \dfrac{\Ex_{\bdelta_1,\dots, \bdelta_n}\left[ \prod_{i=1}^n \left( \binom{s}{r_i} \left( \frac{1}{2} + \frac{\bdelta_i \tau}{\sqrt{n}}\right)^{r_i} \left( \frac{1}{2} - \frac{\bdelta_i \tau}{\sqrt{n}}\right)^{s-r_i}\right) \right]}{\Ex_{\bgamma_1,\dots, \bgamma_n}\left[ \prod_{i=1}^n \left( \binom{s}{r_i} \left( \frac{1}{2} + \frac{\bgamma_i \tau}{\sqrt{n}}\right)^{r_i} \left( \frac{1}{2} - \frac{\bgamma_i \tau}{\sqrt{n}}\right)^{s-r_i}\right) \right]} \leq 1 + o_n(1). \label{eq:ratio}\\[-3ex] \nonumber
\end{align}
Toward this goal, consider a fixed $r \notin B$, and notice that since $\bdelta_1, \dots, \bdelta_n$ are drawn independently, the numerator in (\ref{eq:ratio}) is 
\begin{align}
&\prod_{i=1}^n\hspace{0.06cm} \Ex_{\bdelta_i}\left[\binom{s}{r_i} \left(\frac{1}{2} + \frac{\bdelta_i \tau}{\sqrt{n}} \right)^{r_i} \left( \frac{1}{2} - \frac{\bdelta_i \tau}{\sqrt{n}}\right)^{s - r_i} \right] \nonumber \\
&\qquad\qquad\qquad= \prod_{i=1}^n \binom{s}{r_i} \cdot \frac{1}{2^s}\cdot\Ex_{\bdelta_i}\left[ \left(1 - \left( \frac{2\bdelta_i \tau}{\sqrt{n}}\right)^2 \right)^{m_i} \left(1 - \sgn(t_i) \cdot \frac{2 \bdelta_i \tau}{\sqrt{n}} \right)^{|t_i|}\right], \label{eq:simplify}
\end{align}
where $t_i = s - 2r_i$ and $m_i = \min\left\{ r_i, s - r_i \right\}$; notice that $|t_i| \leq 2 \sqrt{s}\log^2 n$ since $r \notin B$. Similarly, the denominator in (\ref{eq:ratio}) may be expressed as (\ref{eq:simplify}) by replacing $\bdelta_i$ with $\bgamma_i$. We analyze (\ref{eq:ratio}) by considering each term in the product; in particular, it suffices to show that for every $i \in [n]$,
\begin{align}
\dfrac{\Ex_{\bdelta_i}\hspace{-0.04cm}\left[\left( 1 - 4 \bdelta_i^2 \tau^2/n\right)^{m_i} \left( 1 - \sgn(t_i) \cdot 2\bdelta_i \tau / \sqrt{n}\right)^{|t_i|}\right]}{\Ex_{\bgamma_i}\hspace{-0.04cm}\left[\left( 1 - 4 \bgamma_i^2 \tau^2/n\right)^{m_i} \left( 1 - \sgn(t_i) \cdot 2\bgamma_i \tau / \sqrt{n}\right)^{|t_i|}\right]} \leq 1 + o_n(1/n). \label{eq:each-prod}
\end{align}
Using the choice of $s$ and the fact that both $\bdelta_i$ and $\bgamma_i$ are 
no larger than $\log^3 n$, we always have
\begin{align}
\left( 1 - \frac{4 \bdelta_i^2 \tau^2}{n}\right)^{m_i} , \left(1 - \sgn(t_i)\cdot \frac{ 2\bdelta_i \tau}{\sqrt{n}} \right)^{|t_i|}, \left(1 - \frac{4\bgamma_i^2 \tau^2}{n} \right)^{m_i} , \left( 1 - \sgn(t_i)\cdot \frac{2\bgamma_i \tau}{\sqrt{n}} \right)^{|t_i|} =1\pm o_n(1).\label{eq:bounds}
\end{align}
In addition, we have,
\begin{align}
\left( 1 - \frac{4 \bdelta_i^2 \tau^2}{n}\right)^{m_i} &= \sum_{k=0}^{m_i} \binom{m_i}{k} \left( \frac{-4\bdelta_i^2 \tau^2}{n}\right)^k\nonumber\\[0.5ex] 
&= \sum_{k=0}^{\ell/4-1} \binom{m_i}{k} \left( \frac{-4\bdelta_i^2 \tau^2}{n}\right)^k + \sum_{k=\ell/4}^{m_i} \binom{m_i}{k}\left( \frac{-4\bdelta_i^2 \tau^2}{n}\right)^k. \label{eq:haha} 
\end{align}
For each term in the second sum, 
we upperbound $\bdelta_i \leq \ell^3$ and use the approximation of $\binom{m_i}{k} \leq (em_i/k)^k$. 
We also use $k\ge \ell/4$, $m_i\ge s/3$ and the choice of $\ell=\lceil \log n/\log \log n \rceil$.
As a result, the absolute value of the $k$th term in the second sum is at most
\begin{equation}\label{hehe4}
\left(\frac{em_i\cdot 4\ell^6 \cdot O(\eps^2)}{kn}\right)^k
\le \left(O\left(\frac{s\ell^5 \eps^2}{n}\right)\right)^k
\le \left(\frac{1}{\log^6 n}\right)^k. 
\end{equation}
As a result, the absolute value of the second sum is at most
$$
\sum_{k=\ell/4}^{m_i} \left(\frac{1}{\log^6 n}\right)^k\le 2\cdot \left(\frac{1}{\log^6 n}\right)^{\frac{\log n}{4\log \log n}}=o_n(1/n).
$$ 
In fact, we have shown, by negating all terms in (\ref{eq:haha}) of degree (in $\bdelta_i$) at least $\ell/4$,
\begin{align} 
\left( 1 - \frac{4 \bdelta_i^2 \tau^2}{n}\right)^{m_i} = \sum_{k=0}^{\ell/4 - 1} \binom{m_i}{k} 
\left(\frac{-4\tau^2}{n}\right)^k  \cdot \bdelta_i^{2k} \pm o_n(1/n). \label{eq:gamma-bound-1}
\end{align}
Similarly, 
\begin{align*}
\left(1 - \frac{2\sgn(t_i) \bdelta_i \tau}{\sqrt{n}} \right)^{|t_i|} &= \sum_{k=0}^{|t_i|} \binom{|t_i|}{k} \left( \frac{-2 \sgn(t_i) \bdelta_i \tau}{\sqrt{n}}\right)^k \\[0.5ex]
&=\sum_{k=0}^{\ell/2-1}\binom{|t_i|}{k} \left( \frac{-2 \sgn(t_i) \bdelta_i \tau}{\sqrt{n}}\right)^k
+\sum_{k=\ell/2}^{|t_i|} \binom{|t_i|}{k} \left( \frac{-2 \sgn(t_i) \bdelta_i \tau}{\sqrt{n}}\right)^k.
\end{align*}
Analogously to (\ref{hehe4}), the absolute value of the second sum can be bounded from above by
\begin{align*}
\sum_{k=\ell/2}^{|t_i|} \left( O\left( \frac{|t_i|}{k}\cdot \frac{ \eps \ell^3}{ \sqrt{n}}\right)\right)^k 
\leq 2\left(O\left(\dfrac{1}{\log^2 n \log^2 (\log n)}\right) \right)^{\frac{\log n}{2\log\log n}} = o_n(1/n) 
\end{align*}
and we have
\begin{align} 
\left(1 - \frac{2\sgn(t_i) \bdelta_i \tau}{\sqrt{n}} \right)^{|t_i|} = \sum_{k=0}^{\ell/2-1} \binom{|t_i|}{k} \left(\frac{-2\sgn(t_i) \tau}{\sqrt{n}} \right)^k\cdot  \bdelta_i^k \pm o_n(1/n). \label{eq:gamma-bound-2}
\end{align}
Analogously, we may conclude that 
\begin{align}
\left(1 - \frac{4\bgamma_i^2 \tau^2}{n} \right)^{m_i} &= \sum_{k=0}^{\ell/4-1} \binom{m_i}{k} 
\left(\frac{-4\tau^2}{n}\right)^k \cdot \bgamma_i^{2k} \pm o_n(1/n)\qquad \text{and} \nonumber \\[0.6ex]
\left( 1 - \frac{2\sgn(t_i)\bgamma_i \tau}{\sqrt{n}} \right)^{|t_i|} &= \sum_{k=0}^{\ell/2-1} \binom{|t_i|}{k} \left(\frac{-2\sgn(t_i) \tau}{\sqrt{n}}\right)^{k}\cdot  \bgamma_i^k \pm o_n(1/n). \label{eq:gamma-bound-3}
\end{align}
It follows from (\ref{eq:bounds}) and all four approximations in 
(\ref{eq:gamma-bound-1}), (\ref{eq:gamma-bound-2}) and (\ref{eq:gamma-bound-3})
that all four sums on the right hand side are $1\pm o_n(1)$,
and note that all these inequalities hold with probability $1$ (over the draw of $\bdelta_i$ and $\bgamma_1$).
Putting (\ref{eq:gamma-bound-1}), (\ref{eq:gamma-bound-2}), (\ref{eq:gamma-bound-3}) 
and (\ref{eq:bounds}) together, we have
\begin{align}
&\Ex_{\bdelta_i}\left[ \left( 1 - 4 \bdelta_i^2 \tau^2/n\right)^{m_i} \left( 1 - \sgn(t_i) \cdot 2\bdelta_i \tau / \sqrt{n}\right)^{|t_i|} \right] \label{eq:delta-part}\\[0.5ex]
&\qquad\leq \Ex_{\bdelta_i}\left[\left( \sum_{k=0}^{\ell/4 - 1} \binom{m_i}{k} 
\left(\frac{-4\tau^2}{n}\right)^k\cdot 
\bdelta_i^{2k}+ o_n(1/n)\right) \left( \sum_{k=0}^{\ell/2-1} \binom{|t_i|}{k} \left(\frac{-2\sgn(t_i) \tau}{\sqrt{n}} \right)^k \cdot  \bdelta_i^k + o_n(1/n) \right)\right]  \nonumber \\[1ex]
&\qquad\leq \Ex_{\bdelta_i}\left[\left( \sum_{k=0}^{\ell/4 - 1} \binom{m_i}{k} 
\left(\frac{-4\tau^2}{n}\right)^k\cdot \bdelta_i^{2k}\right) \left( \sum_{k=0}^{\ell/2-1} \binom{|t_i|}{k} \left(\frac{-2\sgn(t_i) \tau}{\sqrt{n}} \right)^k\cdot  \bdelta_i^k \right) \right] + o_n(1/n), \nonumber \\[1.5ex]
&\Ex_{\bgamma_i}\left[ \left( 1 - 4 \bgamma_i^2 \tau^2/n\right)^{m_i} \left( 1 - \sgn(t_i) \cdot 2\bgamma_i \tau / \sqrt{n}\right)^{|t_i|} \right] \label{eq:gamma-part} \\[0.5ex]
&\qquad\geq \Ex_{\bgamma_i}\left[\left( \sum_{k=0}^{\ell/4 - 1} \binom{m_i}{k} 
\left(\frac{-4\tau^2}{n}\right)^k \cdot \bgamma_i^{2k}\right) \left( \sum_{k=0}^{\ell/2-1} \binom{|t_i|}{k} \left(\frac{-2\sgn(t_i) \tau}{\sqrt{n}} \right)^k\cdot  \bgamma_i^k \right) \right] - o_n(1/n). \nonumber
\end{align}
Hence, notice that (\ref{eq:delta-part}) and (\ref{eq:gamma-part}) are both $1\pm o_n(1)$, and can each be expressed as the same linear function of the first $\ell-1$ moments of $\bdelta_i$ and $\bgamma_i$ up to additive errors $\pm o_n(1/n)$. Since the first $\ell-1$ moments of $\bdelta_i$ and $\bgamma_i$ are equal by Claim~\ref{cl:matching-moments}, we have shown (\ref{eq:ratio}), which completes the proof of (\ref{finaleq}).

\subsection{Proof of Lemma \ref{dtvlowerbound}}\label{sec:dtvlowerbound}

We will use the fact that $e^{-x}\le 1-x/2$ for all $x\in [0,1]$,
which implies that
\begin{equation}\label{hehe3}
e^{-x}\le \max\big(e^{-1},1-x/2\big)
\end{equation} 
for all $x\ge 0$.
We set the constant $c^*$ in Lemma \ref{dtvlowerbound} to be $1-e^{-1}$.

Let $\mu=\mu(p)$ for convenience and
we assume without loss of generality that $\mu_i \ge 0$ for all $i\in [n]$.
A sample $\bx \sim p$ has all coordinates set independently, where the $i$th coordinate of $\bx_i$ is $1$ with probability $(1+\mu_i)/2$ and $-1$ with probability $(1-\mu_i)/2$.
Given any $x\in \{-1,1\}^n$, we have 
$$
p(x)=\prod_{\substack{i \in [n] \\x_i=1}} \left(\frac{1+\mu_i}{2}\right)\cdot 
\prod_{\substack{i \in [n] \\ x_i=-1}}\left(\frac{1-\mu_i}{2} \right)
=\frac{1}{2^n}\cdot \prod_{\substack{i \in [n] \\ x_i=1}} (1+\mu_i)\cdot \prod_{\substack{i \in [n] \\ x_i=-1}} (1-\mu_i).
$$
We say a string $x\in \{-1,1\}^n$ is \emph{good} if 
$$
\sum_{i\in [n]}  \mu_i x_i\le -\frac{\|\mu\|_2}{ 2}.
$$
The proof proceeds in two steps.
First we show that there exists a constant $c_1^* > 0$ such that when $\bx$ is drawn uniformly at random from $\{-1,1\}^n$, 
$\bx$ is good with probability at least 
$$
\frac{1}{4}-\frac{c_1^*\|\mu\|_\infty}{\|\mu\|_2}.
$$
Next we show there exists a constant $c_2^* > 0$ that every good string $x\in \{-1,1\}^n$ satisfies
$$
\left|p(x)-\frac{1}{2^n}\right|\ge \frac{1}{2^n}\cdot \min\left(c_2^*,\frac{\|\mu\|_2}{4}\right).
$$
The lemma follows since
\begin{align*}
\dtv(p,\calU_n)
&=\frac{1}{2}  \sum_{x \in \{-1,1\}^n} \left|p(x)-\frac{1}{2^n}\right|
\ge \frac{1}{2} \sum_{\substack{x \in \{-1,1\}^n \\ \text{good $x$}}} 
\left|p(x)-\frac{1}{2^n}\right| \\
&\ge \frac{1}{2}\cdot \Prx_{\bx \sim \{-1,1\}^n}\big[\text{$\bx$ is good}\big]
\cdot \min\left(c^*,\frac{\|\mu\|_2}{4}\right).
\end{align*}

For the first step,
we let $\bx\sim \{-1,1\}^n$ be drawn uniformly at random  
and write $\by_i=\mu_i\bx_i$. We recall the Berry--Ess\'een theorem:

\begin{theorem}[Berry--Ess\'een]
	There exists a universal constant $c_1^* > 0$ such that letting $\bs=\by_1+\cdots+\by_n$, where $\by_1,\ldots,\by_n$ be independent real-valued
	random variables with $\E[\by_i]=0$ and $\Var[\by_i]=\sigma_i^2$,
	and suppose that $|\by_i|\le \tau$ with probability $1$ for all $i\in [n]$.
	Let $\bg$ be a Gaussian random variable with mean $0$ and 
	variance $\sum_{i\in [n]} \sigma_j^2$, matching those of $\bs$.
	Then for all $\theta\in \mathbb{R}$ we have 
	$$
	\Big|\Pr[\bs\le \theta]-\Pr[\bg\le \theta]\Big|\le \frac{c_1^* \tau}{\sqrt{\sum_{i\in [n]}\sigma_i^2}}.
	$$
\end{theorem}

Note that in our case, $\tau=\|\mu\|_\infty$ and $\sigma_i^2=\mu_i^2$ and thus, the variance of $\bg$ is $\|\mu\|_2^2$.

Recall the following fact about Gaussian anti-concentration:

\begin{fact}[Gaussian anti-concentration]
	Let $\bg$ be a Gaussian random variable with variance $\sigma^2$. Then for all $\kappa>0$
	it holds that $$\sup_{\theta\in \mathbb{R}} \Big\{\Prx\big[|\bg-\theta|\le \kappa\sigma\big]\Big\}
	\le \kappa.$$
\end{fact}

Setting $\kappa=1/2$ and $\theta=0$ (and using the symmetry of $\bg$), we have that
$$
\Prx_{\bg \sim\calN(0, \|\mu\|_2^2)}\Big[ \bg \le - \|\mu\|_2/2\Big]\ge 1/4.
$$
It follows from Berry-Ess\'een that
$$
\Prx_{\bx \sim \{-1,1\}^n}\left[\sum_{i\in [n]} \mu_i\bx_i\le -\frac{\|\mu\|_2}{2}\right]\ge \frac{1}{4}-\frac{ c_1^* \|\mu\|_\infty}{\|\mu\|_2}.
$$
This finishes the proof of the first step.
For the second we use the fact that $e^x\ge 1+x$ for all $x\in \mathbb{R}$ and thus,
for each good $x\in \{-1,1\}^n$ we have
$$
2^n\cdot p(x)\le \prod_{\substack{i \in [n] \\ x_i=1}} e^{\mu_i}\cdot \prod_{\substack{i \in [n] \\ x_i=1}} e^{-\mu_i}
=e^{\sum_{i\in [n]}  \mu_i x_i }\le e^{-\|\mu\|_2/2}\le \max\left(e^{-1},1-\|\mu\|_2/4\right),
$$
where we used (\ref{hehe3}) in the last inequality.
As a result, we have 
$$
\big|1-2^n\cdot p(x)\big|= 1-2^n\cdot p(x) 
\ge  \min\left(c_2^*,\|\mu\|_2/4\right) 
$$
since we set $c_2^*=1-e^{-1}$. This finishes the proof of the lemma.

\subsection{Proof of Claim \ref{cl:z-bound}}\label{sec:ell-1-z}

Applying Cramer's rule, we have
\begin{align} 
|z_i| &= \left|\dfrac{\det(A_i)}{\det(A)}\right| = \prod_{j \in [\ell] \setminus \{i\}} \left|\dfrac{\alpha_j}{\alpha_i - \alpha_j} \right|,  \label{eq:z-val}
\end{align}
where $A_i$ is the $\ell \times \ell$ matrix given by replacing the $i$-th column with $e_1$, and notice that $A_i$ is the Vandermonde matrix $A(\alpha^{(i)})$, with $\alpha^{(i)} \in \R^{\ell}$ being the vector which is exactly $\alpha_j$ on all $j \neq i$ and $0$ when $j = i$. We now show that there exists a constant $i_0 \in \N$ (which does not depend on $\ell$) such that for all $\ell \in \N$, the sequence $\{ |z_i| \}_{i \geq i_0}$ is geometrically decreasing with constant bounded away from $1$. This suffices to bound $\|z\|_1$, since 
\begin{align*}
\|z\|_1 = \sum_{i=1}^{\ell} |z_i| \leq \sum_{i=1}^{i_0 - 1} |z_i| + \sum_{i=i_0}^{\ell} |z_i| \leq (i_0 - 1) \max_{i <i_0} |z_i| + O(|z_{i_0}|),
\end{align*}
and for every $i \in [\ell]$, we can upperbound the logarithm of (\ref{eq:z-val}) by
\begin{align*}
\log_2\big(|z_i|\big) &\leq (i-1) \log_2 i + \sum_{j > i} \log_2\left( 1 + \frac{i^3}{j^3 - i^3}\right) \leq i^3\left(1 + \sum_{j > i} \frac{1}{j^3 - i^3} \right) \lsim i^3.
\end{align*}
The first inequality follows from the fact that $| j^3 / (i^3 - j^3)| \leq i$ for $j <i$; the second inequality follows from upperbounding $\log(1+x) \leq x$ for $x \geq 0$; the last inequality follows from the fact 
$j^3-i^3>(j-i)^3$ for all $j>i$, and the sums to a constant.
From the upper bound on $\log_{2}|z_i|$, we may conclude $\|z\|_{1} \leq 2^{O(i_0^3)}$. 

In order to pick $i_0 \in \N$, notice that for all $i \in \N$, we use (\ref{eq:z-val}) on $z_{i+1}$ and $z_i$ to obtain
\begin{align*}
\frac{|z_{i+1}|}{|z_i|}
= \prod_{j=0}^{i-1}\frac{i^3-j^3}{(i+1)^3-j^3} \cdot \prod_{j=i+2}^{\ell} \frac{j^3-i^3}{j^3-(i+1)^3}.
\end{align*}
We first handle the case when $\ell\ge 2i+1$. In this case we break the product into
\begin{equation}\label{hehe5}
\frac{|z_{i+1}|}{|z_i|}
=\prod_{k=1}^i \left(\frac{i^3-(i-k)^3}{(i+1)^3-(i-k)^3}\cdot \frac{(i+1+k)^3-i^3}{(i+1+k)^3-(i+1)^3}\right)
\prod_{j=2i+2}^\ell \frac{j^3-i^3}{j^3-(i+1)^3}.
\end{equation}
Using $a^3-b^3=(a-b)(a^2+ab+b^2)$, the factor for each $k\in [i]$ in the first product becomes
\begin{align}
\frac{3i^2-3ki+k^2}{3i^2-3(k-1)i+k^2-k+1}\cdot
\frac{3i^2+3(k+1)i+(k+1)^2}{3i^2+3(k+2)i+k^2+3k+3}.\label{hehe5}
\end{align}
Noting that  
the denominator of the first factor is
$$
(i+1)^2+(i+1)(i-k)+(i-k)^2\le (2i+1-k)^2
$$
we can bound the first factor of (\ref{hehe5}) by 
$$
1-\frac{3i-k+1}{(i+1)^2+(i+1)(i-k)+(i-k)^2}\le 
1-\frac{3i-k+1}{(2i+1-k)^2}\le 1-\frac{1}{2i+1-k}.
$$
Similarly we have that the second factor of (\ref{hehe5}) is 
$$
1-\frac{3i+k+2}{(i+1+k)^2+(i+1+k)(i+1)+(i+1)^2}
\le 1-\frac{3i+k+2}{(2i+2+k)^2}\le 1-\frac{1}{2i+k+2}.
$$
As a result, the first product in (\ref{hehe5}) is at most (using $1+x\le e^x$)
$$
\exp\left(-\sum_{k=1}^i \left(\frac{1}{2i+1-k}+\frac{1}{2i+k+2}\right)\right).
$$
We note that by re-indexing terms,
\begin{align*}
\sum_{k=1}^{i} \left( \frac{1}{2i + 1 - k} + \frac{1}{2i + k + 2} \right) = \sum_{h=i+1}^{3i + 2} \frac{1}{h} - \left(\frac{1}{2i+1} + \frac{1}{2i+2} \right) \geq \int_{i+1}^{3i+2} \frac{1}{x} \cdot dx - \frac{1}{i} \mathop{\longrightarrow}^{i\to \infty} \ln(3)
\end{align*}
where the sum approaches $\ln 3$ as $i$ grows so we fix our $i_0$ to be sufficiently large
so that when $i\ge i_0$ the above sum is at least $1 + \frac{1}{20}$.
For the second product of (\ref{hehe5}) we rewrite it as
\begin{align*}
\prod_{j=2i+2}^{\ell} \frac{j^3-i^3}{j^3-(i+1)^3}&=
\prod_{k=i+1}^{\ell-i-1} \left(1+\frac{3i^2+3i+1}{(i+1+k)^3-(i+1)^3}\right)\\
&\le \prod_{k=i+1}^{\ell-i-1} \left(1+\frac{3i^2+3i+1}{3(i+1)k^2+k^3}\right)\\
&\le \prod_{k=i+1}^{\ell-i-1} \left(1+\frac{i}{k^2}\right)\le \exp\left(\sum_{k\ge i+1} \frac{i}{k^2}\right)
\le e,
\end{align*}
where the third inequality used $3i^2+3i+1\le i(3i+3+k)$ and the last inequality
used the fact that $\sum_{k\ge i+1} 1/k^2\le 1/i$.
As a result, in this case ($i\ge i_0$ and $\ell\ge 2i+1$) we have that
$|z_{i+1}|/|z_i|\le e^{-1/20}$.
We are almost done. For the case when $\ell<2i+1$, we simply note that
$$
\frac{|z_{i+1}|}{|z_i|}
\le \prod_{k=1}^i \frac{i^3-(i-k)^3}{(i+1)^3-(i-k)^3}\cdot \frac{(i+1+k)^3-i^3}{(i+1+k)^3-(i+1)^3}
$$
since we added more factors that are at least $1$.
Since $i\ge i_0$, the same argument used earlier implies that the ratio is at most $e^{-1-1/20}$.

\section{Robust Mean Testing for $k$-Juntas}

In this section, we consider a robust distribution testing algorithm which distinguishes between a given distribution $p$ having a mean vector $\mu(p)$ with large $\ell_2$ norm, and $p$ being a $k$-junta distribution and having a mean vector with small $\ell_2$ norm. Our tester is similar to the mean testing algorithm of \cite{CCKLW20}, however we will require a tighter analysis of the completeness case, which in our setting is more general. The goal of this section is to demonstrate an algorithm
that draws a small number of samples from $p$ to distinguish these two cases with
probability at least $2/3$. We restate the main theorem of this section:

\MeanTestingRestatable*

To describe the testing algorithm we start with some notation.
\begin{definition}
	Given $x\in \{-1,1\}^n$, we write $x\otimes y$ to denote the tensor product of $x$ and $y$:
	$$
	x \otimes x = (x_1 x_1, x_1 x_2 , \dots x_1 x_n , x_2 x_1, \dots x_n x_n)  \in \{-1,1\}^{n^2}.
	$$
	We also write $x^{\otimes {r}}$ to denote the tensor product of $r$ copies of $x$:
	$$
	x^{\otimes r}=\underbrace{x\otimes x\otimes \cdots \otimes x}_{r}.
	$$
	
	Given a distribution $p$ over $\{-1,1\}^n$, we define the \emph{tensor-distribution} $\odot(p)$ of $p$,
	a distribution over $\{-1,1\}^{n^2}$, as 
	the distribution of $\bx\otimes \bx$ with $\bx\sim p$.
	We define $\odot^r(p)$ recursively as $\odot^r (p) = \odot(\odot^{r-1}(p))$
	with $\odot^0(p)=p$, which is a 
	distribution of dimension ${n^{2^r}}$. 
	We call  $\odot^r(p)$ the \emph{$r$-th order tensor distribution} of $p$ and note that, 
	equivalently, $\odot^r(p)$ is the distribution of $\bx^{\otimes 2^r}$ with $\bx\sim p$.
\end{definition}

The following claim follows from the definition
of tensor-distributions since
$\mu(\odot^{r+1}(p))$ is the vectorization of the covariance matrix $\Sigma(\odot^r(p))$.

\begin{claim}\label{clm:meanVar}
	Given $p$ over $\{-1,1\}^n$ and $r\ge 0$, 
	we have $\|\mu(\odot^{r+1}(p))\|_2=\|\Sigma(\odot^r(p))\|_F$.
\end{claim}

Let $p$ be a distribution over $\{-1,1\}^n$.
The main test statistic used by our algorithm first draws $2q$ samples $\bX_1,\dots, \bX_q$ and $\bY_1,\dots,\bY_q$ independently from $p$, for some $q$ to be specified, and construct
\[	\overline{\bX} = \frac{1}{q}\sum_{i=1}^{q} \bX_i\qquad
\text{and}\qquad\overline{\bY} = \frac{1}{q}\sum_{i=1}^{q}\bY_i.	\]
\noindent
We then set 
$\bZ = \langle \overline{\bX}, \overline{\bY} \rangle$.
We use the following lemma (Lemma 4.1) from \cite{CCKLW20}:

\begin{proposition}\label{prop:expvar}
	Let $p$ be a distribution over $\{-1,1\}^n$. Then we have
	\begin{align*}
	\E\big[\bZ\big] &= 
	\big\|\mu( p)\big\|_2^2 \\[0.5ex]
	\mathbf{Var}\big[\bZ\big] &\leq \frac{1}{q^2}\cdot  \big\|\Sigma( p)\big\|_F^2 + \frac{4}{q} \cdot
	\big\|\mu( p)\big\|_2^2\cdot  \big\|\Sigma( p)\big\|_F.	
	\end{align*}
\end{proposition}

We will use the above test statistic for higher order tensor distributions of $p$.
For $r \geq 0$, given $2q$ samples $\bX_1,\dots, \bX_q$ and $\bY_1,\dots,\bY_q$ from $p$, we use them
to obtain $2q$ samples $\bX_1^{(r)},\dots, \bX_q^{(r)}$ and $\bY_1^{(r)},\dots,\bY_q^{(r)}$ from $\odot^r(p)$, by setting $$\bX_i^{(r)} = \bX_i^{\otimes {2^r}}\in \{-1,1\}^{n^{2^r}}.$$ We can then similarly form their averages  $\overline{\bX}^{(r)}, \overline{\bY}^{(r)}$ and set $\bZ^{(r)} = \langle \overline{\bX}^{(r)}, \overline{\bY}^{(r)} \rangle$.

We record the following corollary from the above proposition:

\begin{corollary}\label{coro111}
	Let $p$ be a distribution over $\{-1,1\}^n$ and $r\ge 0$. Then we have
	\begin{align*}
	\E\left[\bZ^{(r)}\right] &= 
	\big\|\mu(\odot^{r}(p))\big\|_2^2\\[0.5ex]
	\mathbf{Var}\left[\bZ^{(r)}\right] &\leq \frac{1}{q^2}\cdot  \big\|\Sigma(\odot^{r}(p))\big\|_F^2 + \frac{4}{q}\cdot \big\|\mu(\odot^{r}(p))\big\|_2^2\cdot  \big\|\Sigma(\odot^{r}(p))\big\|_F .
	\end{align*}
\end{corollary}

Next, we set 
$$
q=C\cdot \max\left\{\frac{k + \sqrt{n}}{\eps^2 n } , \frac{1+ k/\sqrt{n} }{\eps} \right\}
$$
for some sufficiently large constant $C>0$, and
define a sequence $(\tau_r)_{r\ge 0}$ with 
$\tau_0={\eps^2n}/{2}$ and 
\begin{align}\label{eq:tauRecursive}
\tau_r=\frac{1}{5000} \cdot q^2 \tau_{r-1}^2
\end{align}
for each $r\ge 1$.
Setting $a = 1/5000$, we have the following closed form for $\tau_r$:
\begin{equation}\label{eqn:tau}
\tau_r = \frac{1}{aq^2} \left( \frac{aq^2 \eps^2 n}{2}\right)^{2^r}.
\end{equation}




\begin{figure}[t!]	
	\begin{algorithm}[H]
		\caption{Robust Junta Mean Tester}\label{alg:MeanTester}
		\SetKwInOut{Input}{input}
		\SetKwInOut{Output}{output}
		\Input{Sample access to distribution $p$ over $\{-1,1\}^n$ and a distance
			parameter $\eps\in (0,1)$}
		Set $r_0 = \lceil \log \log n\rceil $. \\
		Draw a sequence of $2q$ samples $\bS = (\bX_{1},\dots,\bX_{q}, \bY_{1},\dots,\bY_{q})$  
		from $p$ independently\\
		\For{$r=0,1,2,\dots r_0$}{
			
			Using samples from $\bS$ to compute	
			$\overline{\bX}^{(r)},\overline{\bY}^{(r)}$ and $\bZ^{(r)}$  \\
			\If{$\bZ^{(r)} > \tau_r$ }{
				\Output{\texttt{Not a $k$-Junta}}
			}
			
		}
		\If{All $r_0$ tests pass}{
			\Output{\texttt{Is a $k$-Junta}}	
		}
		
	\end{algorithm}
	\caption{Robust Junta Mean Tester}\label{fig:mean}
\end{figure}

Our main algorithm is presented in Figure \ref{fig:mean} and we prove Theorem~\ref{thm:MeanTesting++} in the rest of the section. We divide the proof of correctness into a soundness and completeness case. The two cases are addressed in Sections \ref{sec:completenes} and \ref{sec:soundness} respectively,
where we prove the following two lemmas:

\begin{lemma}[Soundness]\label{lem:soundness}
	Suppose $p$ is a distribution over $\{-1,1\}^n$ satisfying $\|\mu(p)\|_2 \geq  \eps \sqrt{n}$.
	Then there exists an $r \in \{0,1,\dots,r_0\}$ such that
	\[	\bpr{ \bZ^{(r)} > \tau_r } \geq \frac{2}{3}.	\]
\end{lemma}  
\begin{lemma}[Completeness]\label{lem:smallk}
	Suppose $p$ is a $k$-junta distribution over $\{-1,1\}^n$
	with $\|\mu(p)\|_2 \leq  {\eps} \sqrt{n}/100$. 
	Then for every $r \in \{0,1,\dots,r_0\}$, we have 	
	\[	\bpr{ \bZ^{(r)} > \tau_r } \leq \frac{1}{25}  \cdot \left(\frac{1}{2}\right)^{2^r-1}	.\]
\end{lemma}

\begin{proofof}{Theorem~\ref{thm:MeanTesting++}}
	The soundness case follows directly from Lemma \ref{lem:soundness}.
	For completeness,  
	we can apply a union bound over all $r \in \{0,1,\dots,r_0\}$, giving 
	\begin{equation}
	\begin{split}
	\bpr{\bZ^{(r)} > \tau_r\ \text{for some $r\in \{0,1,\ldots,r_0\}$}}  \leq \sum_{r \geq 0} \frac{1}{25}  \cdot \left(\frac{1}{2}\right)^{2^r-1}  
	<1/3. 
	\end{split}
	\end{equation}
	
	The sample complexity of the algorithm follows directly from our choice of $q$ in (\ref{eq:q-setting}).	
	Finally, we demonstrate that $\bZ^{(r)}$ from the $r$-th order tensor distribution can be 
	computed in polynomial time in $n$ and $q$ --- much faster than the naive $O(n^{2^r})$ time 
	required to compute samples $\bX^{(r)}_i$ from $\odot^r(p)$ using samples $\bX_i$ from $p$. 
	To do this, we will use the following \textit{mixed-product} property of tensor products.
	
	\begin{fact}[\cite{van2000ubiquitous}]
		If $A,B,C,D$ are matrices with such that the products $AC$ and $BD$ are well-defined, then we have $(A \otimes B)(C \otimes D) = (AC \otimes BD)$.
	\end{fact}
	
	Let $X_1,\ldots,X_q,Y_1,\ldots,Y_q$ be strings in $\{-1,1\}^n$.
	Then our target $Z^{(r)}$ can be written as
	\begin{equation}
	\begin{split}
	Z^{(r)} &= \frac{1}{q^2} \left\langle \sum_{i=1}^q X_i^{\otimes {2^r}} , \sum_{i=1}^q Y_i^{\otimes {2^r}}\right\rangle \\ 
	& = \frac{1}{q^2} \sum_{1 \leq i , j \leq q} \left(X_{i}^{\otimes {2^r}}\right)^T Y_{j}^{\otimes {2^r}} \\
	& = \frac{1}{q^2} \sum_{1 \leq i , j \leq q} \left(X_{i }^T \otimes X_{i }^T \otimes \cdots \otimes X_{i}^T\right) \left(Y_{j} \otimes Y_{j} \otimes \cdots \otimes Y_{j}\right) \\
	& = \frac{1}{q^2} \sum_{1 \leq i,j \leq q} \left(X_{i }^TY_{j} \otimes X_{i }^T Y_{j}\otimes \cdots \otimes X_{i }^T Y_{j}\right)\\ 
	& = \frac{1}{q^2} \sum_{1 \leq i,j \leq q} \langle X_{i }, Y_{j} \rangle^{2^r}. 
	\end{split}
	\end{equation}
	To compute $Z^{(r)}$ for each $r=0,1,\dots,r_0$, we can first construct the $q\times q$ matrix  $M$
	with $M_{i,j} =  \langle X_{i}, Y_{j} \rangle$ in time $O(q^2 n)$.  Then each $Z^{(r)}$ is just the average of $2^r$-th power of entries of $M$, namely $Z^{(r)} = (1/q^2)\cdot \sum_{i,j} M_{i,j}^{2^r}$.
	The time needed to compute $Z^{(r)}$ from $M$ for $r=0,1,\dots,r_0$ is $o(q^2 n)$ (recall that 
	$r=\lceil \log \log n\rceil$).
	This completes the analysis of running time of our algorithm.  \end{proofof}

\subsection{Soundness: Proof of Lemma \ref{lem:soundness}}\label{sec:soundness}
We first prove the following lemma, which we will iteratively apply in the soundness case.

\begin{lemma}\label{lem:soundinduction}
	Let $p$ be a distribution supported on $\{-1,1\}^n$ and $r\ge 0$. 
	Suppose that
	$$\|\mu(\odot^r (p))\|_2^2 \geq 2 \tau$$ for some $\tau>0$ and
	$\pr{\bZ^{(r)}  \leq \tau  } \geq 1/3$.
	Then we have $\|\mu(\odot^{r+1} (p))\|_2^2 \geq  (\tau q/24)^2 .$
\end{lemma}
\begin{proof}
	By Proposition \ref{prop:expvar}, we have $\E\big[{\bZ^{(r)}}\big] = \|\mu(\odot^r (p))\|_2^2 \ge 2 \tau$. Thus 
	\begin{equation}\label{eqn:soundnesstau}
	\begin{split}
	\frac{1}{3} &\leq \Pr\big[{\bZ^{(r)} \leq \tau}\big] \leq \Pr\Bigg[{ \left|\bZ^{(r)} - \E\big[{\bZ^{(r)}}\big] \right| \geq \frac{\E\big[{Z^{(r)} }\big]}{2}  }\Bigg]\\&\leq \frac{4}{\|\mu(\odot^r(p)) \|_2^4}\left(\frac{1}{q^2}\cdot	\|\mu(\odot^{r+1} (p) ) \|_2^2 + \frac{4}{q}\cdot \|\mu(\odot^{r} (p) ) \|_2^2  \cdot  \|\mu(\odot^{r+1} (p) ) \|_2 \right) ,
	\end{split}
	\end{equation}
	where in the last inequality we applied Chebyshev's inequality. It follows that at least one of the two terms on the last line of equation (\ref{eqn:soundnesstau}) must be greater than $1/6$. Thus $	\|\mu(\odot^{r+1} (p) ) \|_2^2  \geq \tau^2 q^2 / 3$ or $	\|\mu(\odot^{r+1} (p) ) \|_2 \geq \tau q/24$, from which the lemmas follows.
\end{proof}

\begin{proofof}{Lemma~\ref{lem:soundness}}
	Assume for the sake of contradiction that  
	$\pr{\bZ^{(r)}  \leq \tau_r  } \geq 1/{3}$ for every $r =0,1, \dots,r_0$. 
	We apply Lemma \ref{lem:soundinduction} to prove by induction on $r$ that 	
	$	\|\mu(\odot^{r} (p) ) \|_2^2 \geq 2 \tau_r$ for every $r=0,1,2,\dots,r_0+1$.
	The base case of $r=0$ follows from the choice of $\tau_0=\eps^2 n/2$ and the assumption
	that $\|\mu(\odot^{0}(p))\|_2=\|\mu(p)\|_2\ge \eps\sqrt{n}$.
	For the induction step, we have by the inductive hypothesis that 
	$	\|\mu(\odot^{r} (p) ) \|_2^2 \geq 2 \tau_r$ for some $r\le r_0$.
	It follows from Lemma \ref{lem:soundinduction} and  $\pr{\bZ^{(r)}  \leq \tau_r  } \geq 1/{3}$
	that 
	$$	\big\|\mu(\odot^{r+1} (p) ) \big\|_2^2 \geq \left(\frac{\tau_r q}{24}\right)^2
	\ge \frac{1}{2500}\cdot q^2\tau_r^2=2\tau_{r+1}.$$
	
	Now to get a contradiction, we note that	
	\[		\big\|\mu(\odot^{r_0+1} (p) ) \big\|_2^2 \geq  \frac{2}{aq^2} \left(\frac{a q^2 \eps^2 n}{2}\right)^{2^{r_0+1}}=   q^{2^{r_0 +2} - 2} \cdot \left(\eps \sqrt{n} \right)^{2^{r_0+2}} \cdot \left(\frac{a}{2}\right)^{2^{r_0+1} - 1}.	\]
	Given that $q\ge C/\eps$ and $q\ge C/(\eps^2 \sqrt{n})$ in (\ref{eq:q-setting}),
	we have 
	$$q^{2^{r_0 +2}-2}
	\ge \left(\frac{C}{\eps}\right)^{2^{r_0+2}-4}\cdot \left(\frac{C}{\eps^2\sqrt{n}}\right)^2
	=\left(\frac{ 1}{\eps}\right)^{2^{r_0+2}}\cdot \frac{1}{n}\cdot C^{2^{r_0+2}-2}
	$$
	and thus, 
	$$
	\|\mu(\odot^{r_0+1} (p) ) \|_2^2 \geq n^{2^{r_0+1}}\cdot \frac{1}{n}\cdot  C^{2^{r_0+2}-2}\cdot 
	\left(\frac{a}{2}\right)^{2^{r_0+1}-1},
	$$
	which, after setting $C$ to be a large enough constant and recalling that $r_0=\lceil \log \log n\rceil$, contradicts the fact that we always have $\|\mu(\odot^{r_0+1} (p) ) \|_2^2 \allowbreak \leq n^{2^{r_0+1}}$. This completes the proof of the lemma.
\end{proofof}
\noindent

\subsection{Completeness: Proof of Lemma \ref{lem:smallk}}
\label{sec:completenes}
We will now need the following bound on the mean vector in the completeness case.

\begin{proposition}\label{prop:mean}
	Suppose $p$ is a $k$-junta distribution over $\{-1,1\}^n$.
	Then for each $r\ge 1$ we have 
	\[ 
	\big\|\mu(\odot^r (p))\big\|_2^2 \leq  \left( 2 \cdot \max \{n, k^2\}\cdot  2^r \right)^{2^{r-1} }. \]
\end{proposition}
\begin{proof}
	For $r=0$, the result holds because $ \mu(p) $ is $k$-sparse when $p$ is a $k$-junta distribution. 
	
	Next consider the case when $r>0$. Let $R = 2^{r}$ and let $S\subseteq [n]$ be the set of 
	influential variables with $|S|=k$. (Note that if the number of influential variables is 
	smaller than $k$ we can always add more variables to $S$ to make it size $k$.) 
	Without loss of generality we assume $S=[k]$ and 
	by the definition of $k$-junta distributions, there is 
	a distribution $p'$ over $\{-1,1\}^k$ such that $\bx=(\bx_1,\ldots,\bx_n)\sim p$ can be drawn
	by first drawing $(\bx_1,\ldots,\bx_k)\sim p'$ and then drawing each $\bx_i$, $i>k$,
	independently and uniformly at random from $\{-1,1\}$.
	
	Now we consider the mean vector $\mu(\odot^r(p))$. Note that it has $n^R$ entries and each entry is indexed
	by an $R$-tuple $I=(i_1,\dots,i_R) \in [n]^R$: the entry indexed by $I$ is given by
	$$
	\E_{\bx\sim p} \big[\bx_{i_1}\cdots \bx_{i_R}\big].
	$$
	We define $Q \subseteq [n]^R$ as the set of all $R$-tuples $I=(i_1,\dots,i_R) \in [n]^R$ such that
	every $j \notin S$ appears an even number of times in $I$. 
	Given that every $\bx_j$, $j\notin S$, is drawn independently from other variables and is uniform over
	$\{-1,1\}$, we have that 
	entries of $\mu(\odot^r (p))$ are zero outside of those indexed by tuples in $Q$.
	On the other hand, every nonzero entry of $\mu(\odot^r (p))$ trivially 
	has magnitude no larger than $1$.
	As a result, $\|\mu(\odot^r (p)) \|_2^2 \le |Q|$ and we bound $|Q|$ in the rest of the proof.
	
	To this end, let $Q_i \subseteq Q$ be the set of $I=(i_1,\dots,i_R) \in Q$ such that $\{  \ell \in [R] \; | \; i_\ell \notin S \}| = i$. Then $$|Q_i|\le \binom{R}{i}\cdot  k^{R-i}\cdot  L_i,$$ where $L_i$ is the number of ordered $i$-tuples, each entry selected from $[n]$ (note that we relaxed it from $[n]\setminus S$ to $[n]$ to
	simplify the presentation since this can only make $L_i$ bigger), in which every $j\in [n]$ appears an even number of times. Note that $L_i$ is trivially $0$ when $i$ is odd. We can bound $L_i$ by noting that to pick   a tuple $(i_1,\dots,i_R) \in Q_j$, we can first pick $i_1 \in [n]$, and then pick an index $i_j$ for some $j > 1$ and set $i_j = i_1$. Next, we pick $i_{2}\in [n]$ (or $i_3$ if $i_2$ was chosen to be $i_j$ in the first round) and then pick an unused index $i_{j'}$ for some $j' > 2$ and set $i_{j'}=i_2$, and so on. Thus,
	\[	L_i \leq n(i-1) \cdot n(i-3) \cdots n 
	=\left(\frac{n}{2}\right)^{i/2}\cdot \frac{i!}{(i/2)!}
	\leq \left(\frac{n}{2}\right)^{i/2} \cdot i^{i/2}\]
	when $i$ is even. 
	Using that $|Q_{0}| = k^R$, we have
	\[|Q| \leq\sum_{\ell=0}^{R/2}| Q_{2 \ell} | \leq k^R + \sum_{\ell=1}^{R/2} \binom{R}{2\ell} \cdot\left( n \ell\right)^{\ell} 
	\cdot k^{R-2\ell}.	\]
	Letting $\alpha = \max\{k^2, n \}$ so that $k\le \sqrt{\alpha}$ and $n\le \alpha$, we have
	\begin{align*}
	|Q|&\le \alpha^{R/2}+\sum_{\ell=1}^{R/2} \binom{R}{2\ell}\cdot \ell^\ell\cdot \alpha^{R/2} \le \alpha^{R/2}\left(1+(R/2)^{R/2}\cdot \sum_{\ell=1}^{R/2} \binom{R}{2\ell}\right) 
	\le \alpha^{R/2}\cdot (R/2)^{R/2}\cdot 2^R,
	\end{align*}
	which completes the proof.\end{proof}

We now start the proof of Lemma \ref{lem:smallk}.

\begin{proofof}{Lemma~\ref{lem:smallk}}
	Again, we set $R = 2^r$. We show for each $r\in \{0,1,\ldots,r_0\}$ that
	\begin{equation}\label{eq:mainmain}\E\left[{ \bZ^{(r)}} \right]= \big\|\mu(\odot^r (p))\big\|_2^2 \leq  \frac{1}{100}\left(\frac{1}{2}\right)^{R - 1}\cdot  \tau_r 
	\qquad\text{and}\qquad
	\textbf{Var}\left[ \bZ^{(r)} \right] \leq \frac{1}{100}\left(\frac{1}{2}\right)^{R - 1} \tau_r^2.
	\end{equation}
	Assuming this, by Chebyshev's inequality we have
	\begin{equation}
	\begin{split}
	\bpr{  \bZ^{(r)}  > \tau_r } \leq 	\bpr{ \left|\bZ^{(r)} -\E\left[{ \bZ^{(r)}}\right]\right|  > \tau_r/2 }   \leq 4\cdot\frac{ \textbf{Var}[\bZ^{(r)}]}{\tau_r^2} 
	\leq \frac{1}{25}  \cdot \left(\frac{1}{2}\right)^{R-1} 
	\end{split}
	\end{equation}
	and this finishes the proof of the lemma.

	We start with the case when $r=0$.
	The first part of (\ref{eq:mainmain}) follows trivially from the assumption that
	$\|\mu(p)\|_2 \leq  {\eps} \sqrt{n}/100$,
	and the second part follows from Lemma \ref{prop:mean}.
	To see the latter, we have from Claim \ref{clm:meanVar} and Lemma \ref{prop:mean} that
	$$
	\textbf{Var}\left[ \bZ^{(0)} \right]\le \frac{1}{q^2}\cdot \left(4\cdot \max(n,k^2)\right)
	+\frac{4}{q}\cdot \frac{\eps^2n}{10000}\cdot \sqrt{\left(4\cdot \max(n,k^2)\right)}
	\le \frac{1}{100}\left(\frac{1}{2}\right)^{R-1}\cdot \tau_1^2 ,
	$$
	where the last inequality used the choice of $\tau_1$, $\eps\le 1$, and 
	$q\ge C (k+\sqrt{n})/(\eps^2 n)$ for some sufficiently large constant $C$.
	
	Moving to the general case when $r\ge 1$, we have $R=2^r\ge 2$.
	Letting $\beta=\max(n,k^2)$ and using $q\ge C\sqrt{\beta}/(\eps^2n)$ and $q\ge C\sqrt{\beta}/(\eps \sqrt{n})$,
	we have
	$$
	q^{2R-2}=q^{2R-4}\cdot q^2
	\ge \left(\frac{C\sqrt{\beta}}{\eps \sqrt{n}}\right)^{2R-4}\cdot \left(\frac{C\sqrt{\beta}}{\eps^2n}\right)^2
	=\left( {C^2\beta} \right)^{ R-1}\cdot \left(\frac{1}{\eps^2n}\right)^{ R } .
	$$ 
	Plugging this in the closed form (\ref{eqn:tau}) of $\tau_r$, we have
	$$
	\tau_r=\frac{1}{aq^2}\left(\frac{aq^2\eps^2n}{2}\right)^R
	\ge \frac{1}{2}\cdot \left(\frac{aC^2\beta}{2}\right)^{ R-1}.
	$$
	
	Using Proposition \ref{prop:mean}, we have 
	$
	\E\big[ \bZ^{(r)}\big] \leq\left( 2R \beta \right)^{R/2 }
	$ and thus,
	$$
	\frac{\E [ \bZ^{(r)} ]}{\tau_r}\le \left(2R\cdot \left(\frac{2}{aC^2}\right)^{R-1}\cdot 2^{R/2}\right)
	\cdot \left(\frac{R}{\beta}\right)^{R/2-1}.
	$$
	Note that $r\le r_0=\lceil \log \log n\rceil$ and thus $R/\beta<1$ when $n$ is sufficiently large.
	As a result we have 
	$$
	\frac{\E [ \bZ^{(r)} ]}{\tau_r}\le 2R\cdot \left(\frac{2}{aC^2}\right)^{R-1}\cdot 2^{R/2}
	\le 2R\cdot \left(\frac{4}{aC^2}\right)^{R-1} 
	\le \frac{1}{100}\left(\frac{1}{2}\right)^{R-1},
	$$
	when $C$ is sufficiently large.
	This completes the proof of the first part of (\ref{eq:mainmain}). 
	For the second part, by Corollary \ref{prop:expvar}
	and using the first part of (\ref{eq:mainmain})  and
	the recursive definition of $\tau_r$ in (\ref{eq:tauRecursive}), we have  
	\begin{align*}
	\textbf{Var}\left[\bZ^{(r)}\right]& \leq \frac{1}{q^2}\cdot \big\|\mu(\odot^{r+1} (p))\big\|_2^2 + \frac{4}{q}\cdot\big\| \mu(\odot^{r} (p))\big\|_2^2\cdot \big\|\mu(\odot^{r+1} (p))\big\|_2\\
	&	\leq \frac{1}{100\cdot q^2 \cdot 2^{2R - 1} }\cdot \tau_{r+1} + \frac{1}{ 250\cdot q \cdot 2^{R - 1}}\cdot \tau_r \cdot\sqrt{\tau_{r+1}}\\
	&	= \frac{1}{100\cdot q^2\cdot 2^{2R - 1} }\cdot  \left(\frac{q^2 \tau_r^2}{5000 }\right) + \frac{1}{ 250\cdot q \cdot 2^{R - 1}}\cdot \tau_r \cdot \sqrt{\frac{q^2 \tau_r^2}{5000 }}   
	< \frac{1}{100  }  \left(\frac{1}{2}\right)^{R-1} \cdot\tau_r^2.
	\end{align*}
	This finishes the proof of the lemma.
\end{proofof}

\section{Proof of the Main Structural Lemma: Lemma~\ref{lem:main-structural}}\label{sec:structural}

In this section, we prove the main structural lemma. The goal is to relate the distance in total variation from a distribution which is far from being a $k$-junta to the expected Euclidean distance of its mean vector  after applying random restrictions.

The proof of Lemma \ref{lem:main-structural} uses the following results from \cite{CCKLW20}, which we reproduce below.
\begin{lemma}[Lemma~1.4 in \cite{CCKLW20}]\label{lem:base-case}
	Let $p$ be a distribution over $\{-1,1\}^n$. For any $\sigma \in (0,1)$, 
	\begin{align*}
	\dtv(p, \calU) \leq \Ex_{\bS \sim \calS_{\sigma}}\left[ \dtv(p_{\ol{\bS}}, \calU)\right] + \Ex_{\brho \sim \calD_{\sigma}(p)}\left[ \dtv(p_{|\brho}, \calU) \right].
	\end{align*}
\end{lemma} 

\begin{lemma}[Implicit in \cite{CCKLW20}]\label{lem:one-star}
	Let $p$ be a distribution over $\{-1,1\}^n$. Then we have
	\begin{align*}
	\frac{\dtv(p, \calU)}{n \log n} &\lsim \Ex_{\substack{\bi \sim [n] \\ \brho \sim \calD_{\{ \bi\}}(p)}}\Big[ \big\|\mu(p_{|\brho})\big\|_2 \Big].
	\end{align*}
\end{lemma}
\begin{proof}
	We follow Subsection~1.1.2 in \cite{CCKLW20}.
	Let $f \colon \{-1,1\}^n \to [-1, \infty)$ be $$f(x) = 2^n \cdot p(x) - 1.$$ Then by the first part of (4) in \cite{CCKLW20} (scaled by $1/n$), we have
	\begin{align*}
	\frac{\dtv(p, \calU)}{n\log n} &\lsim \frac{1}{n} \cdot \Ex_{\bx \sim \{-1,1\}^n}\left[ \sqrt{\sum_{i=1}^n \left(\big(f(\bx) - f(\bx^{(i)})\big)^+\right)^2}\hspace{0.1cm} \right]  \\[0.5ex] &= \frac{1}{n} \cdot \Ex_{\bx \sim p}\left[ \sqrt{\sum_{i=1}^n \left(\dfrac{\big(f(\bx) - f(\bx^{(i)})\big)^+}{f(\bx) + 1}\right)^2} \hspace{0.1cm}\right]  \\[0.8ex]
	&\leq \frac{1}{n} \cdot \Ex_{\bx \sim p}\left[\hspace{0.06cm} \sum_{i=1}^n \left|\dfrac{\big(f(\bx) - f(\bx^{(i)})\big)^+}{f(\bx) + 1}\right| \hspace{0.06cm}\right] \\[0.8ex] &\leq \frac{2}{n}\cdot  \sum_{i=1}^n \Ex_{\bx \sim p}\left[ \hspace{0.06cm}\left|\frac{p(\bx) - p(\bx^{(i)})}{p(\bx) + p(\bx^{(i)})} \right|\hspace{0.06cm}\right] = 2\Ex_{\substack{\bi \sim [n] \\
			\brho \sim \calD_{\{\bi\}}(p)}}\Big[ \big|\mu(p_{|\brho})_{\bi}\big|\Big],
	\end{align*}
	where the first inequality uses a robust version of Pisier's inequality on $f$ (see Theorem 1.7 and (3)
	in \cite{CCKLW20});
	the next equation follows from importance sampling; the third inequality uses  Jensen's inequality. 
	Finally we note that since $p_{|\brho}$ is supported on a single bit, the absolute value is the same as the Euclidean norm.
\end{proof}

We point out that the two lemmas above hold even when $n$ is a small constant.
The next theorem from \cite{CCKLW20} holds only when $n$ is sufficiently large.

\begin{theorem}[Theorem~1.5 in \cite{CCKLW20}]\label{thm:restriction-thm}
	Let $p$ be a distribution over $\{-1,1\}^n$. For any $\sigma \in (0,1)$,
	\begin{align}
	\Ex_{\brho \sim \calD_{\sigma}(p)}\Big[ \big\|\mu(p_{|\brho})\big\|_2 \Big] \geq \frac{\sigma}{\poly(\log n)} \cdot \tilde{\Omega}\left( \Ex_{\bS \sim \calS_{\sigma}}\left[ \dtv(p_{\ol{\bS}}, \calU)\right] - 2e^{-\min(\sigma, 1-\sigma) n/10}\right). \label{eq:restriction-thm}
	\end{align}
\end{theorem}

We are now ready to prove Lemma \ref{lem:main-structural}.

\begin{proofof}{Lemma~\ref{lem:main-structural}}
	Let $q$ be the junta distribution on $J$ such that 
	its projection $q_J$ is the same as $p_J$ (equivalently, one can draw $\bx\sim q$ by
	first drawing a string from $\{0,1\}^J$ from $p_J$ and then drawing every other bit
	independently and uniformly at random).
	Given our assumption that $p$ is $\eps$-far from every junta distribution over $J$, we have
	\begin{equation}\label{hehe12}
	\eps\le \dtv(p,q)=\Ex_{\brho \sim \calD_{\ol J}(p)}\Big[\dtv\big(p_{|\brho},q_{|\brho}\big)\Big]
	=\Ex_{\brho \sim \calD_{\ol J}(p)}\Big[\dtv\big(p_{|\brho},\calU \big)\Big].
	\end{equation}
	In the rest of the proof we consider 
	a restriction $\rho \in \{-1,1,*\}^n$ with $\stars(\rho) =  {\overline{J}}$
	and lowerbound $\dtv(p_{|\rho},\calU)$.
	For simplicity of notation, we let $g = p_{|\rho}$ be the distribution supported over $\smash{\{-1,1\}^{\ol{J}}}$. The goal is to obtain a lower bound for $\dtv(g, \calU)$ in terms of mean vectors of random restrictions of $g$, which is then plugged into (\ref{hehe12}) to finish the proof of Lemma \ref{lem:main-structural}.
	
	Let $m=|\overline{J}|$. We start with the case when $m$ satisfies 
	$
	m\le C\cdot \log (m/\eps)
	$
	for some constant $C>0$.
	We apply Lemma \ref{lem:one-star} on $g$ (with the parameter $n$ set to $m$). 
	There is a constant $\hat{c}$ such that 
	$$
	\dtv(g,\calU) 
	\le \hat{c}\hspace{0.04cm} \log^2(m/\eps)\cdot \Ex_{\substack{\bi \sim [n] \\ \bnu \sim \calD_{\{ \bi\}}(p)}}\Big[ \big\|\mu(g_{|\bnu})\big\|_2 \Big].
	$$
	Letting $j=\lceil \log_2 2m\rceil$,
	the probability of $\brho\sim \calD_{\sigma^j}(g)$ having exactly one $*$ is at least
	$$
	m\cdot \sigma^j\cdot (1-\sigma^j)^{m-1}\ge m\cdot \frac{1}{4m}\cdot 
	\left(1-\frac{1}{2m}\right)^{m-1}
	\ge \frac{1}{8},
	$$  
	and when this happens, the $*$ is distributed uniformly at random.
	As a result, we have\begin{equation}\label{hehehe100}
	\dtv(g,\calU)\le \hat{c}\hspace{0.04cm}\log^2(m/\eps) \cdot \Ex_{\substack{\bi \sim [n] \\ \bnu \sim \calD_{\{ \bi\}}(p)}}\Big[ \big\|\mu(g_{|\bnu})\big\|_2 \Big] \le 8\hspace{0.03cm}\hat{c}\hspace{0.04cm}\log^2(m/\eps)\cdot \Ex_{\bnu \sim \calD_{\sigma }(g)}\Big[ \big\|\mu(g_{|\bnu})\big\|_2 \Big] 
	\end{equation}
	The lemma then follows by combining (\ref{hehe12}) and (\ref{hehehe100}).
	We now turn to the case when 
	\begin{equation}\label{choiceofm}
	|\overline{J}|=m\ge C\cdot \log(m/\eps)
	\end{equation} 
	for some sufficiently large constant $C>0$. 
	We first prove by induction that for any $t \in \N$, 
	\begin{align}
	\dtv(g, \calU) &\leq \Ex_{\bnu \sim \calD_{\sigma^t}(g)}\Big[ \dtv\big(g_{|\bnu}, \calU\big) \Big] + \sum_{j=1}^t \Ex_{\bnu \sim \calD_{\sigma^{j-1}(g)}}\Bigg[ \Ex_{\bS \sim \calS_{\sigma} (\stars(\bnu))}\Big[\dtv\big((g_{|\bnu})_{\ol{\bS}}, \calU\big) \Big] \Bigg]. \label{eq:induction}
	\end{align}
	Lemma~\ref{lem:base-case} provides the base case when $t = 1$, as a draw from the
	distribution $\calD_{1}(g)$ always outputs the all-$*$ restriction $(*, *, \dots, *)$. 
	For the induction step with $t > 1$, notice that
	\begin{align}
	\dtv(g, \calU) &\leq \Ex_{\bnu \sim \calD_{\sigma^{t-1}}(g)}\Big[ \dtv\big(g_{|\bnu}, \calU\big)\Big] + \sum_{j=1}^{t-1}\hspace{0.05cm} \Ex_{\bnu \sim \calD_{\sigma^{j-1}}(g)}\left[ \Ex_{\bS \sim \calS_{\sigma}(\stars(\bnu))}\Big[ \dtv\big((g_{\bnu})_{\ol{\bS}}, \calU\big)\Big]\right] \label{eq:step-1}\\
	&\leq \Ex_{\bnu \sim \calD_{\sigma^{t-1}(g)}}\left[ \Ex_{\bS \sim \calS_{\sigma}(\stars(\bnu))}\Big[ \dtv\big( (g_{|\bnu})_{\ol{\bS}}, \calU\big) \Big] +\Ex_{\bnu' \sim \calD_{\sigma}(g_{|\bnu})}\Big[ \dtv\big((g_{|\bnu})_{|\bnu'}, \calU\big) \Big] \right] \label{eq:step-2}\\[0.5ex]
	&\qquad\qquad +  \sum_{j=1}^{t-1}\hspace{0.05cm} \Ex_{\bnu \sim \calD_{\sigma^{j-1}}(g)}\left[ \Ex_{\bS \sim \calS_{\sigma}(\stars(\bnu))}\Big[ \dtv\big((g_{\bnu})_{\ol{\bS}}, \calU\big)\Big]\right], \nonumber
	\end{align}
	where we first applied the inductive hypothesis in (\ref{eq:step-1}) and then Lemma~\ref{lem:base-case} to the distribution $g_{|\bnu}$ supported on $\smash{\{-1,1\}^{\stars(\bnu)}}$ in (\ref{eq:step-2}). We get (\ref{eq:induction}) by noticing that the distribution over distributions $(g_{|\bnu})_{|\bnu'}$~where $\bnu \sim \calD_{\sigma^{t-1}}(g)$ and $\bnu' \sim \calD_{\sigma}(g_{|\bnu})$ is equivalent to $g_{|\bnu}$ with $\bnu \sim \calD_{\sigma^t}(g)$.
	
	Next for each restriction $\nu \in \{-1,1, *\}^n$ we let 
	\[ \alpha(\nu) = \Ex_{\bS \sim \calS_{\sigma}(\stars(\nu))}\Big[ \dtv\big((g_{|\nu})_{\ol{\bS}}, \calU\big)\Big], \]
	and let $G_t \subset \{-1, 1,*\}^n$ for each $t\in \N$ be the set of restrictions $\nu \in \{-1,1,*\}^n$ that satisfy 
	\[ \alpha(\nu) \geq \max\left\{ \frac{\eps}{6t},\hspace{0.05cm} 4\hspace{0.02cm} e^{-|\stars(\nu)|/20} \right\}. \]
	For each restriction $\nu\notin G_t$ we trivially have
	$$
	\alpha(\nu)\le \frac{\eps}{6t}+4\hspace{0.02cm} e^{-|\stars(\nu)|/20}.
	$$ 
	For each $\nu\in G_t$ we have 
	$$
	\alpha(\nu)-2\hspace{0.02cm}e^{-|\stars(\nu)|/20}\ge \alpha(v)/2\ge \eps/(12t).
	$$
	We can then apply Theorem \ref{thm:restriction-thm} to get
	$$
	\alpha(\nu)\le 
	\left(c_0\cdot \big(\log n\cdot \log(12t/\eps)\big)^{c_1}\right)\cdot 
	\Ex_{\bnu' \sim \calD_{\sigma}(g_{|\nu})} \Big[ \big\| \mu\big((g_{|\nu})_{|\bnu'} \big)\big\|_2 \Big]  
	$$
	for some universal constants $c_0$ and $c_1$.
	Therefore, we have for every $\nu\in \{-1,1,*\}^n$ that
	$$
	\alpha(\nu)\le 
	\left(c_0\cdot \big(\log n\cdot \log(12t/\eps)\big)^{c_1}\right)\cdot 
	\Ex_{\bnu' \sim \calD_{\sigma}(g_{|\nu})} \Big[ \big\| \mu\big((g_{|\nu})_{|\bnu'} \big)\big\|_2 \Big]  
	+\frac{\eps}{6t}+4\hspace{0.02cm} e^{-|\stars(\nu)|/20}.
	$$
	Combining this bound with (\ref{eq:induction}), we get
	\begin{align}
	\dtv(g, \calU) &\leq \Ex_{\bnu \sim \calD_{\sigma^t}(g)}\Big[ \dtv\big(g_{|\bnu}, \calU\big)\Big] \label{eq:line-1}\\
	&\qquad + \left(c_0\cdot \big(\log n\cdot \log(12t/\eps)\big)^{c_1}\right)\cdot\sum_{j=1}^t   \Ex_{\bnu \sim \calD_{\sigma^{j-1}}(g)}\left[  \Ex_{\bnu' \sim \calD_{\sigma}(g_{|\bnu})} \Big[ \big\| \mu\big((g_{|\bnu})_{|\bnu'} \big)\big\|_2 \Big] \right]  \label{eq:line-2}\\
	&\qquad + \frac{\eps}{6} + 4\sum_{j=1}^t \Ex_{\bnu \sim \calD_{\sigma^{j-1}}(g)}\left[ e^{-|\stars(\bnu)| / 20}\right] .\label{eq:line-3}
	\end{align}
	
	Setting (where $C$ is the constant from (\ref{choiceofm}))
	\begin{equation}\label{settingt}
	t = \left\lfloor \log\left(\frac{m}{C\cdot\log (m/\eps)}\right)\right \rfloor + 1 
	\end{equation} in the rest of the proof.
	We upper bound the right-hand side of (\ref{eq:line-3}) by noting that $|\stars(\bnu)|$, when $\bnu \sim \calD_{\sigma^{j-1}}(g)$ is a sum of $n$ independent random variables, where each is set to $1$ with probability $\sigma^{j-1}$. Thus, we have
	\begin{align*}
	\sum_{j=1}^t \Ex_{\bnu \sim \calD_{\sigma^{j-1}}(g)}\left[e^{-|\stars(\bnu)|/20} \right] &= \sum_{j=1}^t \left( \Ex_{\bX \sim \Ber(\sigma^{j-1})}\left[e^{-\bX / 20}\right] \right)^{m}    
	{ =\sum_{j=1}^t \left(1-\sigma^{j-1}\left(1-e^{-1/20}\right)\right)^{m}}
	\\ & \leq\sum_{j=1}^t \left( 1 - \frac{\sigma^{j-1}}{100} \right)^{m} 
	\le t\cdot  \exp\left( -\frac{\sigma^{t-1} m}{100}\right) \leq \frac{\eps}{24},
	\end{align*}
	using our choice of $t$ with $\sigma^{t-1}m \geq C\cdot \log(m/\eps) $ and a sufficiently large constant $C$. Therefore, the right-hand side of (\ref{eq:line-3}) can be bounded from above by $\eps/3$. 
	
	Next we upperbound (\ref{eq:line-2}).
	Using again the fact that $(g_{|\bnu})_{|\bnu'}$ with $\bnu \sim \calD_{\sigma^{j-1}}(g)$ and $\bnu' \sim \calD_{\sigma}(g_{|\bnu})$ is distributed as $g_{|\bnu}$ with $\bnu \sim \calD_{\sigma^j}(g)$, the right-hand side of (\ref{eq:line-2}) may be upper bounded by
	\begin{align}
	\left(c_0\cdot \big(\log n\cdot \log(12t/\eps)\big)^{c_1}\right)\cdot\sum_{j=1}^t   \Ex_{\bnu \sim \calD_{\sigma^j}(g)}\Big[ \big\| \mu(g_{|\bnu})\big\|_2\Big].
	\end{align}
	Finally we bound the right-hand side of (\ref{eq:line-1}) by considering the set of restrictions $F \subset \{-1,1,*\}^n$ where $\nu \in \{-1,1,*\}^n$ is in $F$ iff $|\stars(\nu)| \leq 2C\cdot \log(m/\eps)$, and note that by the setting of $t$,  
	\[ \Prx_{\bnu \sim \calD_{\sigma^t}(g)}\big[ \bnu \notin F \big] \leq \frac{\eps}{6}. \]
	Using the trivial bound of $\dtv(g_{|\nu}, \calU) \leq 1$, we have
	\begin{align*}
	\Ex_{\bnu \sim \calD_{\sigma^t}(g)}\Big[ \dtv\big(g_{|\bnu}, \calU\big) \Big]
	\le \frac{\eps}{6}+\Ex_{\bnu \sim \calD_{\sigma^t}(g)}\Big[ \dtv\big(g_{|\bnu}, \calU\big) \cdot \ind\left\{ \bnu \in F\right\} \Big]
	\end{align*}
	We apply Lemma~\ref{lem:one-star} to every $g_{|\bnu}$ with $\bnu \in F$. So there exists a universal constant $c_2$ such that
	\begin{align*}
	\Ex_{\bnu \sim \calD_{\sigma^t}(g)}\Big[ \dtv\big(g_{|\bnu}, \calU\big) \cdot \ind\left\{ \bnu \in F\right\} \Big]
	\leq c_2\cdot   \log^2(m/\eps) \cdot \Ex_{\bnu \sim \calD_{\sigma^{t}}(g)}\left[\Ex_{\substack{\bi \sim \stars(\bnu) \\ \bnu'\sim \calD_{\{ \bi\}}(g_{|\bnu})}}\Big[ \big\|\mu\big((g_{|\bnu})_{|\bnu'}\big)\big\|_2\Big]  \right].
	\end{align*}
	Note that the distribution on $(g_{|\bnu})_{|\bnu'}$ is equivalent to the distribution $g_{|\bnu}$ which draws $\bi \sim [n]$ and then sets $\bnu \sim \calD_{\{\bi\}}(g)$. Hence, we can upperbound (\ref{eq:line-1}) by
	\begin{align*}
	\frac{\eps}{6} + c_2 \cdot \log^2(m/\eps)\cdot \Ex_{\substack{\bi \sim [n] \\ \bnu \sim\calD_{\{\bi\}}(g)}}\Big[ \big\|\mu\big(g_{|\bnu}\big)\big\|_2 \Big]  
	\leq \frac{\eps}{6} + 4 c_2 \cdot  \log^2(m/\eps) \cdot   
	\Ex_{\bnu \sim \calD_{\sigma^{r}}(g)}\Big[ \big\|\mu\big(g_{|\bnu}\big)\big\|_2\Big] \end{align*}
	where $r = \lceil \log_2 m \rceil$.
	The inequality used the fact that $\nu\sim \calD_{\sigma^r}(g)$ has
	$\stars(\nu)=1$ with probability at least $1/4$ and when this happens, the star is 
	distributed uniformly at random.

	Finally, noting that $t<r$, we combine the upper bounds for (\ref{eq:line-1}), (\ref{eq:line-2}), and (\ref{eq:line-3}) to get 
	\begin{align*}
	\dtv(g, \calU) \leq { \frac{\eps}{2}} + c_3\cdot  \log^{c_4}(n/\eps)\cdot \sum_{j=1}^{ \lceil \log_2 n \rceil} \Ex_{\bnu \sim \calD_{\sigma^j}(g)}\Big[ \big\|\mu(g_{|\bnu})\big\|_2 \Big] 
	\end{align*}
	for some universal constants $c_3$ and $c_4$.
	It follows from  (\ref{hehe12}) that
	\begin{align*}
	\eps \leq \frac{\eps}{2} + \polylog( n/\eps)\cdot  \sum_{j=1}^{ \lceil \log_2 n \rceil} \Ex_{\brho \sim \calD_{\ol J}(p)}\left[ \Ex_{\bnu \sim \calD_{\sigma^j}(p_{|\brho})} \Big[\big\| \mu\big((p_{|\brho} )_{|\bnu}\big)\big\|_2 \Big]\right],
	\end{align*} 
	which completes the proof.
\end{proofof}

\bibliographystyle{alpha}
\bibliography{waingarten}

\end{document}